\documentclass[sigconf, nonacm]{acmart}
\pdfoutput=1
\usepackage{balance}
\usepackage[full]{optional}
\usepackage{graphicx}
\usepackage{booktabs}
\usepackage{hyperref}
\usepackage{cleveref}
\crefformat{footnote}{#2\footnotemark[#1]#3}

\usepackage{amsthm}
\usepackage{algorithm}
\usepackage[noend]{algpseudocode}
\let\oldReturn\Return
\renewcommand{\Return}{\State\oldReturn}
\newcommand{\para}[1]{\vskip 0.06in\noindent {\bf #1} }
\newcommand{\etal}{\textit{et al}. }

\usepackage[style=base]{caption}
\usepackage{amssymb}
\usepackage{multirow}
\usepackage{graphicx}
\usepackage{mwe}
\usepackage[labelformat=simple]{subcaption}

\usepackage{amsfonts}
\usepackage{xcolor}

\theoremstyle{definition}

\newtheorem{lemma}{Lemma}

\newtheorem{definition}{Definition}
\newtheorem{example}{Example}
\newtheorem{property}{Property}

\DeclareMathOperator*{\argmin}{arg\,min}

\def\BibTeX{{\rm B\kern-.05em{\sc i\kern-.025em b}\kern-.08emT\kern-.1667em\lower.7ex\hbox{E}\kern-.125emX}}
    
\begin{document}

\fancyhead{}

\title{On VR Spatial Query for Dual Entangled Worlds}
\titlenote{\textbf{A shorter version of this paper has been accepted for publication in the 28th ACM International Conference on Information and Knowledge Management (CIKM 2019)}.}

\author{
Shao-Heng Ko\textsuperscript{1}, 
Ying-Chun Lin\textsuperscript{2},
Hsu-Chao Lai\textsuperscript{1}\textsuperscript{3},
Wang-Chien Lee\textsuperscript{4},
De-Nian Yang\textsuperscript{1}\textsuperscript{5}
}
\affiliation{\textsuperscript{1} Institute of Information Science, Academia Sinica, Taipei, Taiwan}
\affiliation{\textsuperscript{2} Department of Computer Science, Purdue University, West Lafayette, USA}
\affiliation{\textsuperscript{3} Department of Computer Science, National Chiao Tung University, Hsinchu, Taiwan}
\affiliation{\textsuperscript{4} Department of Computer Science and Engineering, The Pennsylvania State University, State College, USA}
\affiliation{\textsuperscript{5} Research Center for Information Technology Innovation, Academia Sinica, Taipei, Taiwan}
\affiliation{\textsuperscript{1} \{arsenefrog, hclai0806, dnyang\}@iis.sinica.edu.tw \quad \textsuperscript{2} lin915@purdue.edu \quad \textsuperscript{4} wlee@cse.psu.edu}

\begin{abstract}
With the rapid advent of Virtual Reality (VR) technology and virtual tour applications, there is a research need on spatial queries tailored for simultaneous movements in both the physical and virtual worlds. Traditional spatial queries, designed mainly for one world, do not consider the entangled dual worlds in VR. In this paper, we first investigate the fundamental shortest-path query in VR as the building block for spatial queries, aiming to avoid hitting boundaries and obstacles in the physical environment by leveraging Redirected Walking (RW) in Computer Graphics. Specifically, we first formulate \textit{Dual-world Redirected-walking Obstacle-free Path} (DROP) to find the minimum-distance path in the virtual world, which is constrained by the RW cost in the physical world to ensure immersive experience in VR. We prove DROP is NP-hard and design a fully polynomial-time approximation scheme, \textit{Dual Entangled World Navigation} (DEWN), by finding Minimum Immersion Loss Range (\textit{MIL Range}). Afterward, we show that the existing spatial query algorithms and index structures can leverage DEWN as a building block to support $k$NN and range queries in the dual worlds of VR. Experimental results and a user study with implementation in HTC VIVE manifest that DEWN outperforms the baselines with smoother RW operations in various VR scenarios.
\end{abstract}

\maketitle

\section{Introduction}
\label{sec:intro}

With the growing availability of Virtual Reality (VR) devices, innovative VR applications in virtual social, travel, and shopping domains have emerged. This technological trend of VR not only attracts business interests from prominent vendors such as Facebook and Alibaba\footnote{Facebook: \url{https://youtu.be/YuIgyKLPt3s}; Alibaba:\url{https://cnn.it/2GkXUDX}.} but also brings a new wave of research in the academia. While current research on VR mostly originated from Computer Graphics, Multimedia, and HCI, focusing on constructing vivid VR worlds \cite{MP16,AB16,SK16}, the needs for research and support from the database community are also imminent. 

Traditional research on spatial data management has contributed significantly to various applications in the physical world. For example, for mobile users on a journey, the information about the closest gas stations along a routing path can be found by spatial queries \cite{YL14}. These queries are also needed in the virtual worlds in VR applications where moving between point-of-interests (POIs) is a basic operation. For example, in VR campus touring\footnote{CampusTours: \url{https://campustours.com/}; UNSW 360: \url{https://ocul.us/2VBzGlC}.} and VR architecture/indoor navigation\footnote{IrisVR: \url{https://irisvr.com/}; VR for Architects: \url{https://bit.ly/2JlwiVq}.} applications, spatial queries can be issued to find POIs and guide users to move to them. However, in many VR applications where users move in both the virtual and physical worlds, the simple one-world setting may no longer sustain, rendering the aforementioned queries useless. To study this problem, we revisit a number of spatial queries widely used in many VR applications to develop new algorithms by considering factors in the dual entangled virtual and physical worlds.

Traditional VR applications adopt simple stand-and-play approaches, e.g., teleportation \cite{EB16}, which have users to stand still in the physical world and rely on handheld devices, e.g., joysticks, to move to the destination. However, unlike previous generation of VR Head Mound Displays (HMDs), which are tied to computers with cable wires, the new VR devices are either \textit{wireless}\footnote{HTC Vive Pro: \url{https://bit.ly/2AM0vUM}; DisplayLink XR: \url{https://bit.ly/2HdI2FJ}.} or \textit{standalone}\footnote{HTC Vive Focus: \url{https://bit.ly/2US4DwI}; Oculus Go: \url{https://www.oculus.com/go/}.} devices. As this new wave of technology unties VR devices from a fixed computer, \textit{mobile VR} \cite{ST17, LP18, SY18} and \textit{room-scale VR} \cite{LE17VR, LE18, ZY18} recently attract massive attention in HCI and Computer Graphics research communities, as they allow untethered walking\footnote{A number of demo videos on \textit{walking} with \textit{wireless} VR can be found at \url{https://bit.ly/2vWP9gG}, \url{https://bit.ly/2LIlgeT}, and \url{https://bit.ly/2HojNX3}.} in VR to improve user experience. Indeed, research \cite{DC09, JS16, RWsurvey} finds that stand-and-play approaches do not facilitate immersive experience intended in VR. On the contrary, \textit{walking} is able to bring benefits to the users' cognition in virtual environments (VEs) \cite{RR11}, because users can experience correct stimulations \cite{RWsurvey} in order to reduce the side-effect of motion sickness. To avoid hitting physical obstacles, various hardware and HCI solutions leveraging saccadic movement \cite{SQ18}, space partition \cite{MS18} and Galvanic vestibular stimulation \cite{SM16} are proposed recently.

Usually, users in VR applications are severely constrained \cite{LE17VR, LE18, WG18} by the small size and setting of physical space, e.g., living room, during exploration of massive VEs. As a result, if the movement in the virtual world is simply realized by a directly matched walk in the real world, users may easily get  hindered by boundaries of the small physical space.\footnote{See also \url{https://bit.ly/2YuKSgU} and \url{https://bit.ly/2Ebcfox} on this issue.} To address this issue, {\em Redirected Walking} (RW) \cite{SR01,RWsurvey,RO16,WG18} has been proposed to steer users away from physical boundaries and obstacles by slightly tailoring the walking direction and speed displayed in HMDs.\footnote{A series of demo videos elaborating Redirected Walking can be found at \url{https://bit.ly/2JGv8D8} and \url{https://bit.ly/2H6UCb4}.} For example, when a user intends to walk straightly in the virtual world, RW continuously adjusts the walking direction displayed in the HMD to guide the user walking along a curve in the physical world in a small room. 

It has been successfully demonstrated that the human visual-vestibular system does not conceive those minor differences if the RW operations (detailed later) are carefully controlled \cite{SS13,FM16,RWsurvey}, and RW provides the most immersive user experiences compared to joystick and teleportation-based locomotion techniques \cite{LE18,RWsurvey}. However, when a path in the virtual world (called {\em v-path}) is identified by directly employing the shortest-path query, the walking path in the physical world (called {\em p-path}) may involve many RW operations that may incur motion sickness \cite{FS10,SS13,FM16}, thereby deteriorating the user experience.

In this paper, therefore, we first formulate a new query, namely {\em Dual-world Redirected walking Obstacle-free Path} (DROP), to find the minimum-distance v-path from the current user location to the destination that is \textit{RW-realizable} by a corresponding obstacle-free p-path, bounded by a preset total cost on Redirected Walking (RW cost) to restrict the loss of immersive experience in VR. Specifically, given the current positions of the user, the layouts of both the virtual and physical worlds, and a destination in the virtual world, DROP finds a v-path and an RW-realized obstacle-free p-path such that (i) the length of v-path is minimized, and (ii) the total cost incurred by RW operations does not exceed a preset threshold. We introduce the notion of \textit{Minimum Immersion Loss} (MIL) to represent the RW cost for realizing a short walk in dual worlds.

\begin{figure}[t]
\centering
  \begin{subfigure}[b]{0.54\columnwidth}
    \includegraphics[width=\columnwidth]{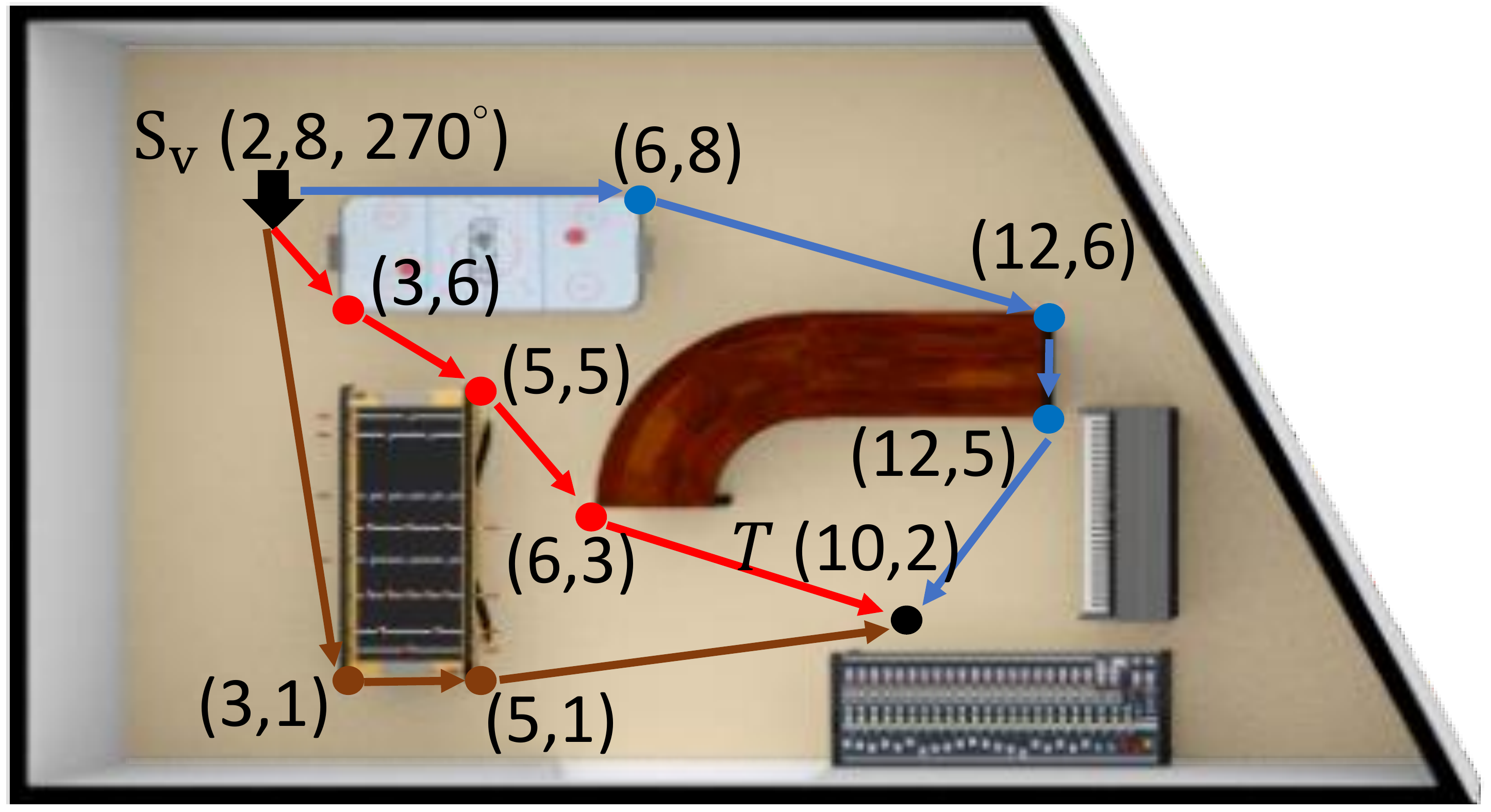}
    \caption{A virtual world.}
    \label{fig:example_v}
  \end{subfigure}
  \begin{subfigure}[b]{0.44\columnwidth}
    \includegraphics[width=\columnwidth]{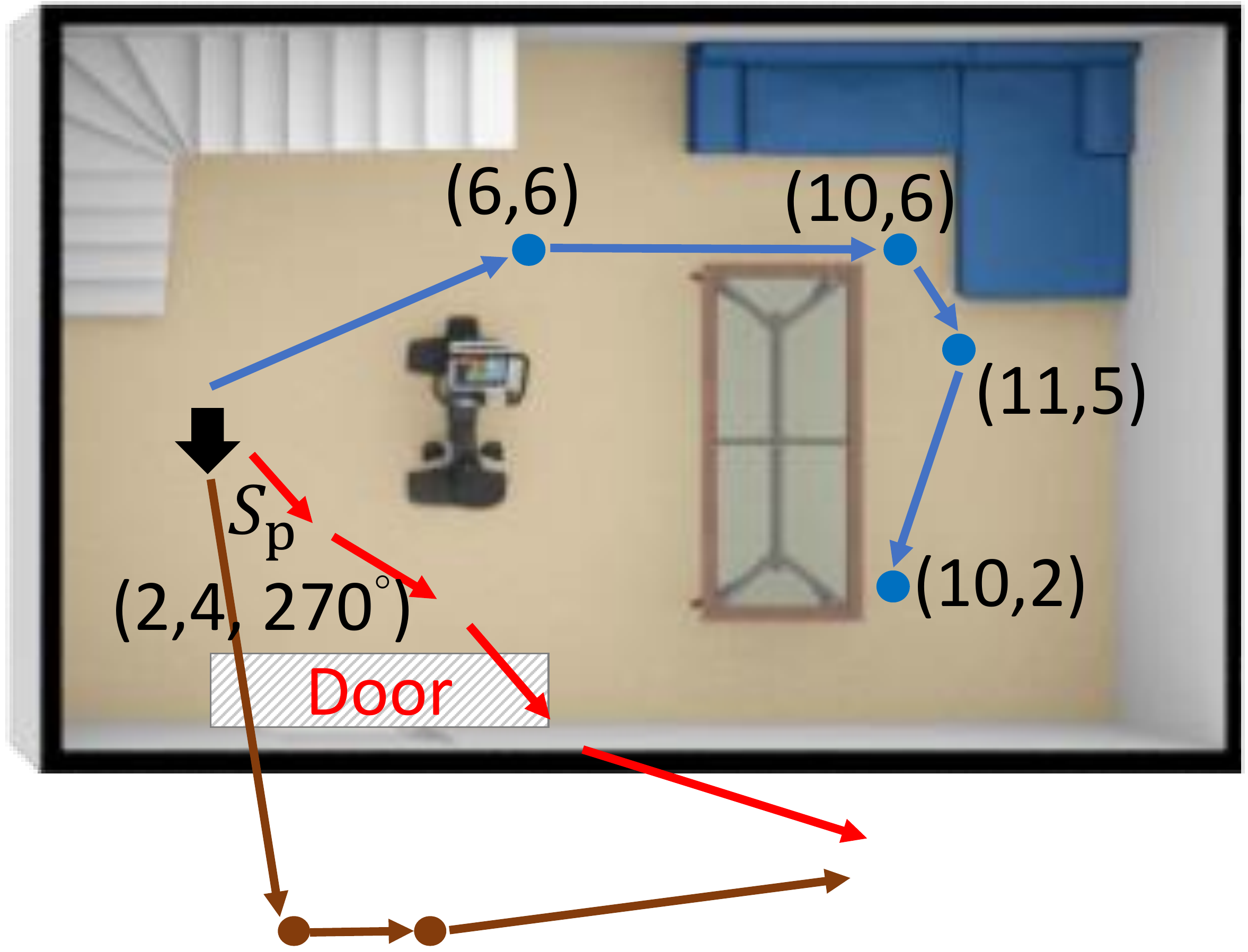}
    \caption{A physical world.}
    \label{fig:example_p}
  \end{subfigure}\\
\caption{An illustrative example for DROP.}
\label{DROPExample}
\end{figure}

\begin{example}(Motivating Example). Figure \ref{DROPExample} lays out an example of virtual and physical worlds to illustrate the notions of v-path and p-path. As shown, $S_{\text{v}}$ and $S_{\text{p}}$ denote the current locations of the user in both worlds, while the thick black arrows indicate the corresponding orientations, i.e., the user faces south in both worlds. The coordinates of some POIs are shown right beside them. The face direction is given (in degrees) for the starting state. Let $T$ be the destination in the virtual world and the preset RW cost threshold is small. In the virtual world, the shortest obstacle-free path, bypassing corners of the obstacles as indicated by the red solid line segments, has a total length of 10.83. However, this path is actually infeasible because the starting location in the physical world is too close to the wall and door area (see the corresponding infeasible p-path shown in red). Similarly, the brown path (which features a length of 14.17 in the virtual world) is not feasible. In contrast, the optimal path of DROP is the blue one with a total length of 14.93. This path, bypassing the upper part of virtual obstacles, incurs only minimal RW operations including a rotation at the beginning to avoid obstacles and prohibited areas in the physical world. \hfill \qedsymbol
\end{example} 

DROP, which actually returns not only the paths in the dual worlds but also the corresponding RW operations, is much more challenging than finding the shortest obstacle-free path in a single world. Some heuristics useful in geographic space, e.g., the triangular inequality, are not applicable here due to the obstacles appearing in both worlds. Moreover, traditional spatial index structures, e.g., R-Tree \cite{RTree}, M-Tree \cite{MTree}, and O-Tree \cite{HZ16} are designed for only one world instead of the entangled dual worlds, and thus do not handle the cost of RW operations. Finally, in a multi-user VR environment, the same path in the virtual world may be walked differently by users in their individual physical worlds. The RW operations carried out for the same virtual path are unlikely to be the same for different users and thus are not precomputable. Indeed, we prove DROP is NP-hard.

To solve DROP, we first present a dynamic programming algorithm, namely \textit{Basic DP}, as a baseline to find the optimal solution which unfortunately requires exponential time. Basic DP is computationally intensive due to the need of maintaining an exponentially large number of intermediate states to ensure the optimal solution. To address the efficiency issue while still ensuring the solution quality, we propose a Fully Polynomial-Time Approximation Scheme, namely \textit{Dual Entangled World Navigation (DEWN)}, to approach the optimal solution in polynomial time. The main idea of DEWN is to quickly obtain a promising feasible solution (called \textit{reference path}) in an early stage. Via the reference path, we explore novel pruning strategies to avoid redundant examinations of states that lead to excessive RW costs or long path lengths. 

However, finding a promising reference path directly from the entangled dual worlds is actually computationally intensive. To address this issue, we precompute the range of RW cost, termed as \textit{Minimum Immersion Loss Range (MIL Range)}, which consists of an \textit{MIL lower bound} and an \textit{MIL upper bound}, for a possible straight-line walk between two POIs in the virtual world. With MIL Ranges for potential path segments in the virtual world, we jointly minimize the weighted sum of v-path length and RW cost by \textit{Lagrangian relaxation} (LR). Accordingly, we derive the optimal weight (i.e., the Lagrange multiplier) to ensure both the feasibility and quality of the reference path. Equipped with DEWN as a building block, we then show that existing spatial query algorithms and index structures can support the counterparts of $k$NN and range queries in VR. The contributions of this work are summarized as follows:
\begin{itemize}
    \item \sloppy We redefine a new shortest path query, namely {\em Dual-world Redirected-walking Obstacle-free Path} (DROP), tailored for the dual entangled obstructed spaces in VR applications. We introduce a novel notion of \text{MIL Range} that captures the possible range of Redirected Walking cost in state transitions of movements and prove DROP is NP-hard. 
    \item We first tackle DROP by dynamic programming and then design an online query algorithm, DEWN, which exploits efficient ordering and pruning strategies to improve computational efficiency significantly. We prove that DEWN is a Fully Polynomial-Time Approximation Scheme for DROP.
    \item We show that existing spatial query algorithms and index structures can leverage DEWN as a building block to support $k$NN and range queries in VR.
    \item We perform experiments on real datasets and conduct a user study to evaluate the proposed algorithms with various baselines. Experimental results show that DEWN outperforms the baselines in both solution quality and efficiency.
    
\end{itemize}

This paper is organized as follows. Section \ref{sec:related} reviews the related work. Section \ref{sec:problem} introduces the preliminaries and formulates DROP. Section \ref{sec:algo} details DEWN and provides a theoretical analysis. Section \ref{sec:extensions} proposes an enhancement for DROP and extends our ideas for spatial queries. Section \ref{sec:exp} reports the experimental results, and Section \ref{sec:conclusions} concludes this paper.

\section{Related Work}
\label{sec:related}

\para{Shortest Path Query.}
Exact \cite{TA13}, top-$k$ \cite{TA15}, approximate \cite{PM09, MQ14}, constrained \cite{SW16, YMN15}, and adaptive \cite{Adapt11, MSH16} shortest path queries have been studied extensively in the database community. Akiba \etal \cite{TA13} precompute shortest path distances by breadth-first search and store the distances on the vertices. To improve efficiency, a query-dependent local landmark scheme \cite{MQ14} is proposed to provide a more accurate solution than the global landmark approach \cite{PM09} by identifying a landmark close to both query nodes and leveraging the triangular inequality. In continental road networks with length and cost metrics, COLA \cite{SW16} utilizes graph partition to minimize the path length within a cost constraint. Hassan \etal \cite{MSH16} find the adaptive type-specific shortest paths in dynamic graphs with edge types. Nevertheless, the above research is designed for \textit{one network} (i.e., \textit{one world}). None of the existing works incorporates the cost, e.g., Redirected Walking, in dual worlds of different layouts.

\para{Spatial Query.}
Spatial database is a major research area in the database research community \cite{spatial}. Queries on spatial network databases, including range search, nearest neighbors, e-distance joins, and closest pairs \cite{PD03}, have attracted extensive research interests. In recent years, considering the presence of obstacles, the obstructed version of various spatial queries are revisited \cite{JZ04}. Sultana \etal \cite{NS14} study the obstructed group nearest neighbor (OGNN) query to find a rally point with the minimum aggregated distance. Range-based obstructed nearest neighbor search \cite{HZ16} extracts the nearest neighbors within a range for obstructed sequenced routes, where the route distance is minimized \cite{AA17}. However, the above algorithms are designed for one world, instead of the entangled dual worlds, where the physical worlds of users are different from each other. As a result, these existing works are not applicable to the dual world spatial queries tackled in this paper. 

\para{Walking in Virtual Reality.}
To move in the virtual space, Point-and-Teleport \cite{EB16} allows a user to point at and then transport to a target location, but the experience is not immersive due to the abrupt scene change and loss of sense in time \cite{DC09}. Research shows that real walking is more immersive than Point-and-Teleport \cite{MU99}. Redirected Walking (RW) \cite{RWsurvey} exploits the inability of the human vestibular system to detect a subtle difference (in the walking speed and direction) between movements in the dual worlds. It has been demonstrated that RW can support free walking in a large virtual space for a relatively small physical space \cite{RWsurvey, CN12}, and the degradation in user immersion can be quantitatively measured from the acoustic and visual perspectives \cite{FS10,SS13,FM16}. Detailed implementation and performance evaluation of RW have been studied in \cite{EH14} and \cite{RWsurvey}. Recent evaluation \cite{LE18} demonstrates that RW provides the most preferable user experience than joystick-based and teleportation-based systems. However, most existing works on RW focus on creating immersive experience but do not provide systematic approaches for query processing in dual worlds, which inspires our study in this work.

\section{Problem Formulation}
\label{sec:problem}

In this section, we first provide background on the Visibility Graph and Redirected Walking operations. Then we formulate the DROP problem and prove that DROP is NP-hard.

\subsection{Preliminaries} \label{subsec:prelim}

\para{Visibility Graph.} 
The notion of Visibility Graph (VG), widely used in computational geometries and obstructed spatial query processing \cite{HZ16, AA17, MM18}, models obstacles as polygons and regards their corners as VG nodes. Those corners are important as they are usually the turning points for shortest paths in an obstructed space. In VG, two nodes are connected by a weighted edge if and only if there exists a straight line segment between them without crossing any obstacle \cite{NS14, HZ16}. In this paper, we exploit VG to define the DROP problem on dual worlds for the following reasons: 1) VG preserves the unobstructed shortest paths in the obstructed spatial space \cite{JZ04, NS14}, simplifying the distance computation and reducing the computational complexity in processing obstructed spatial queries. 2) Representing the virtual world in VG ensures natural movements of users since the obtained v-paths avoid zigzagging patterns. 3) Whereas DROP depends on both worlds due to the RW operations, VG for both worlds can be constructed separately \cite{EM04, HH04}. While existing works on obstructed spatial queries most consider only corners of obstacles in VG, we also extend VG to include all POIs in the virtual world. We refer the interested readers to \cite{VGBook} for more background on VGs.

\begin{figure}[t]
\centering
\includegraphics[width=0.5\columnwidth]{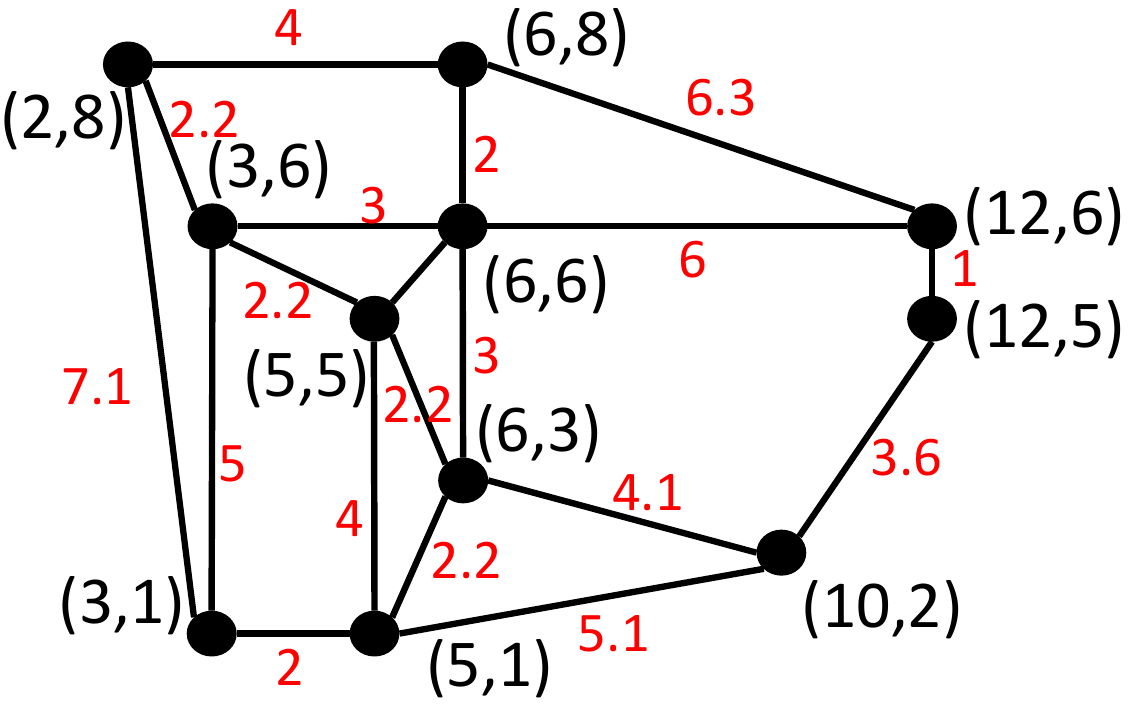}
\caption{VG of the virtual world in Example \ref{DROPExample}.}
\label{fig:VG}
\end{figure}

\begin{example}
Figure \ref{fig:VG} illustrates the VG constructed from the virtual world in Example \ref{DROPExample}, where the nodes represent virtual locations of interests (POIs and obstacle corners) in the application, and the edges (called {\em v-edges}) denote straight-line moving paths between two virtual locations\footnote{We omit a few of VG nodes for brevity and handiness to continue using it as the running example.}. For example, the virtual location $(2,8)$ is a POI (the start location), while $(6,8)$ represents the upper right corner of the white gameboard in Example \ref{DROPExample}. The v-edge between them represents a move along the upper side of the gameboard which has a length of 4 (shown in red). \hfill \qedsymbol
\end{example} 

\para{Redirected Walking Operations.} Redirected Walking (RW) \cite{SR01} introduces minor differences in the walking speed and turning angle to adapt the perception of walking in the dual worlds. Basic RW operations include \textit{Translation} (TO) \cite{WG18}, \textit{Rotation} (RO) \cite{FM16, RO16}, and \textit{Curvature} (CO) \cite{CN12,LE17}. TO introduces a slight scaling factor between the walking speed in the virtual world and the actual walking speed in the physical world. Thus, the distances in the dual worlds are different after a user walks for a period of time. Similarly, RO tailors the mapping between the rotation angular velocity in the virtual world to that in the physical world. When a user intends to move straightly in the virtual world, CO lets the user traverse a slightly bending curve\footnote{\url{https://youtu.be/THk92rev1VA}.} to avoid obstacles in the physical world. 
However, when a user is very close to obstacles and not able to escape from them with the above operations, a \textit{Reset} operation \cite{BW07} may be issued to specifically ask the user to rotate her body to face a different direction in the physical world, whereas the virtual world is \textit{suspended} (remaining the same).\footnote{\url{https://youtu.be/gD1qa0edVA8}.} Note that Reset incurs significantly higher disturbance for users \cite{RWsurvey} and thus introduces a much larger RW cost. An RW cost model of different operations can be constructed based on the usage count or other measures of user experience, e.g., detection thresholds in \cite{SS13,FM16} or immersion thresholds in \cite{PS18}. For example, according to \cite{SS13}, a TO that down-scales the walking distance by 40\% has a roughly 90\% chance to be detected by the users. Thus, applying a TO of such magnitude may incur an RW cost of 0.9 in a detection threshold-based cost model. In Appendix \ref{appen:rw_def}, we provide some definitions of the basic RW operations, as well as briefly discuss some possible RW cost models. For a complete survey on RW, we refer the interested readers to \cite{RWsurvey}.

Given a user's current location and orientation in both worlds (formally introduced later as the \textit{loco-state}), the possible combinations of RW operations to pilot the user to a target loco-state is bounded due to the finite operations.\footnote{For instance, in Example \ref{DROPExample}, to guide the user from the start locations ($S_{\text{v}}$ and $S_{\text{p}}$, in the virtual and physical world, respectively) to the next locations on the blue paths, i.e., (6,8) in the virtual world and (6,6) in the physical world, one possible configuration of RW operations is to first perform an RO that down-scales the rotation angular velocity by 25.0\% to re-orient the user to face the targeted locations, then followed by a TO, which down-scales the walking speed by 10.6\% in the virtual world, to align the walking distances in the dual worlds. Another feasible configuration is a Reset and then a TO, which incurs a larger RW cost since Reset severely downgrades the user experience.} It is also more efficient for the user to move along straight line segments in the VG. Therefore, in this paper, a near-shortest path between two locations with the smallest \textit{RW cost} (i.e., minimum degradation of user experience) can be precomputed by exploring different RW operation sequences. This RW cost is coined as the \textit{Minimum Immersion Loss} (MIL) between the two loco-states. Note that MIL represents the RW costs on small segments of movements. It is independent of the start and destination locations in DROP and thus can be precomputed offline. 

\subsection{Problem Formulation} \label{subsec:prob}

In the following, we introduce the notations used to formulate DROP. We use VG graphs for both virtual and physical worlds to abstract unobstructed movements of users. We also summarize the notations in Tables \ref{tab:symbols_algo1} and \ref{tab:symbols_algo2}.

\begin{definition}
\textit{Location Sets} ($\Gamma^\text{v}, \Gamma^{\text{p}}$). The \textit{virtual location set} $\Gamma^{\text{v}}$ contains all virtual locations $\gamma^{\text{v}} \in \Gamma^{\text{v}}$ corresponding to a VG node in the virtual world. Similarly, the \textit{physical location set} $\Gamma^{\text{p}}$ includes all locations in the physical world, where each \textit{physical location} $\gamma^{\text{p}} \in \Gamma^{\text{p}}$ represents either an unoccupied location or an obstacle in a coarse-grained coordinate of the physical world.\footnote{As the position tracking accuracy in mainstream VR devices varies \cite{MK17}, representing a physical world by a coarse-grid or mesh-based \cite{KM10} graph structure leaves room for errors and may be more suitable than a fine-grained coordinate system.}
\end{definition}

\begin{definition}
\textit{Virtual Graph} ($G^{\text{v}}$) and \textit{Physical Graph} ($G^{\text{p}}$). The virtual graph (\textit{v-graph}) $G^{\text{v}}$ consists of the vertex set $\Gamma^{\text{v}}$ and undirected edge set $E^{\text{v}} : \Gamma^{\text{v}} \rightarrow \Gamma^{\text{v}}$, where a virtual edge (\textit{v-edge}) $e^{\text{v}}$ connects unobstructed virtual locations with a cut-off distance threshold $\ell_{\text{max}}$ \cite{EM04, HH04}. Each v-edge $e^{\text{v}}$ is associated with a positive length $\text{l}(e^{\text{v}})$ that denotes the \textit{Euclidean distance} between the two endpoints in the virtual world. The physical graph (\textit{p-graph}) $G^{\text{p}}$ and the edge set $E^{\text{p}}$ are defined analogously. 
\end{definition}

To determine the appropriate v-path and the corresponding sequence of RW operations, the user's face orientation needs to be considered. In the following, we formally introduce the notion of \textit{loco-state}, which describes the user status in both worlds.

\begin{definition}
\textit{Virtual State} ($st^{\text{v}}$) and \textit{Physical State} ($st^{\text{p}}$). A \textit{v-state} $st^{\text{v}}$ is a tuple $(\gamma^{\text{v}}, \theta^{\text{v}})$ while $\gamma^{\text{v}}$ is the current user location in the virtual world, and $\theta^{\text{v}}$ is her face orientation. The \textit{p-state} $st^{\text{p}} = (\gamma^{\text{p}}, \theta^{\text{p}})$ is defined similarly in the physical world, and $\theta^{\text{v}}, \theta^{\text{p}} \in \Theta$, which is the \textit{Orientation Set} consisting of all legal face directions.
\end{definition}

\begin{definition}
\textit{Locomotion State ($st$) and Loco-state Space ($ST$)}. A loco-state $st = (st^{\text{v}}, st^{\text{p}})$ describes the current user status. The \textit{Euclidean distance} $\text{dist}(st_1, st_2)$ between two loco-states is the straight-line distance between their \textit{virtual} locations. Two loco-states $st_1$ and $st_2$ are \textit{neighboring} if there exists a v-edge $e^{\text{v}}$ between their virtual locations $\gamma^{\text{v}}_1$ and $\gamma^{\text{v}}_2$ with the v-edge length $\text{l}(e^{\text{v}}) = \text{dist}(st_1,st_2)$. The loco-state space $ST$ contains all possible loco-states.
\end{definition}

\begin{example}
In Example \ref{DROPExample}, the starting v-state for the user, denoted as $st^{\text{v}}_{\text{s}}$, is $((2,8), 270^{\circ})$, and the starting p-state is $st^{\text{p}}_{\text{s}} = ((2,4), 270^{\circ})$. The starting loco-state is then $st_{\text{s}} = (((2,8), 270^{\circ}), ((2,4), 270^{\circ}))$. \hfill \qedsymbol
\end{example}

Equipped with the notion of loco-state, user movements in the dual worlds can be regarded as sequences of state transitions between neighboring loco-states. The possible combinations of RW operations to pilot the user to a target loco-state is bounded due to the finite operations. Therefore, a configuration with the smallest \textit{RW cost} (i.e., minimum degradation of user experience) can be precomputed by exploring different RW operation sequences. This RW cost is coined as the \textit{Minimum Immersion Loss} (MIL) between the two loco-states. Note that MIL represents the RW costs on small segments of movements. It is independent of the start and destination locations in DROP and thus can be precomputed offline. It is also generic, i.e., supporting any cost model of RW operations.

\begin{definition}{\textit{Minimum Immersion Loss (MIL).} }\\
$\text{MIL}(st_1, st_2)$ represents the smallest RW cost achievable (i.e., realizable by a set of RW operations) for a VR user to move from a loco-state $st_1$ to a neighboring loco-state $st_2$ with a sequence of RW operations.
\end{definition}

\noindent Next, we introduce \textit{RW path} to describe the RW-realizable v-path and the corresponding RW-realized p-path.

\begin{definition}{ \textit{Redirected Walking Path (RW path).} }
\sloppy An RW path $p = \langle st_1, st_2, \cdots, st_n \rangle $ is a sequence of loco-states, including a v-path $p^\text{v} = \langle st^{\text{v}}_1, st^{\text{v}}_2,$ $\cdots, st^{\text{v}}_n \rangle $ with \textit{v-path length} $\text{l}(p)= \text{l}(p^\text{v}) = \sum_{i=1}^{n-1} \text{l} \big( (\gamma^{\text{v}}_i, \gamma^{\text{v}}_{i+1}) \big)$, and a p-path $p^\text{p} = \langle st^{\text{p}}_1, st^{\text{p}}_2, \cdots, st^{\text{p}}_n \rangle $ with the incurred \textit{RW cost} to realize $p^\text{v}$ with $p^\text{p}$ as $\text{c}(p) = \sum_{i=1}^{n-1} \text{MIL}(st_i, st_{i+1})$. 
\end{definition}

\begin{example}
In Example \ref{DROPExample}, the two blue paths in the dual worlds combine for an RW path $p = \langle st_{\text{1}} = st_{\text{s}} = (((2,8), 270^{\circ}), ((2,4), 270^{\circ}))$, $st_{\text{2}} = (((6,8)$, $0^{\circ})$, $((6,6), 30^{\circ}))$, $st_{\text{3}} = (((12,6), 330^{\circ})$, $((10,6), 0^{\circ}))$, $st_{\text{4}} = (((12,5)$, $270^{\circ})$, $((11,5), 315^{\circ}))$, $st_{\text{5}} = (((10,2), 225^{\circ})$, $((10,2)$, $240^{\circ}))$ $\rangle$. The lengths of the corresponding v-edges are respectively $\text{l} \big( ((2,8),(6,8)) \big) = 4$, $\text{l} \big( ((6,8),(12,6)) \big) = 6.32$, $\text{l} \big( ((12,6),(12,5)) \big) = 1$, and $\text{l} \big( ((12,5),(10,2)) \big) = 3.61$. Thus, the total v-path length is $4+6.32+1+3.61=14.93$. Assume the MIL values between the loco-states are $\text{MIL}(st_1, st_2) = 0.17$, $\text{MIL}(st_2, st_3) = 1$, $\text{MIL}(st_3, st_4) = 1.18$, and $\text{MIL}(st_4, st_5) = 1$ (these values are derived via a detection threshold-based cost model). The total RW cost along $p$ is then $0.17+1+1.18+1 = 3.35.$ \hfill \qedsymbol
\end{example}

\noindent 
Note that $\text{dist}(st_1, st_2)$ is the straight-line distance between their virtual locations. However, the \textit{v-path length} $\text{l}(p)$ of some RW path $p$ from $st_1$ to $st_2$ may not be the same as the \textit{Euclidean distance} $\text{dist}(st_1, st_2)$ or the \textit{obstructed distance} \cite{NS14} between $st_1$ and $st_2$ in the virtual world. For instance, in the above example, the v-path length $p^\text{v}$ is 14.93, while the Euclidean distance between $(2,8)$ and $(10,2)$ is 10.0, and the obstructed shortest distance is 10.83. We formulate DROP as follows. 

\vskip 0.06in
\hrule
\vskip 0.03in

\noindent \textbf{Problem: Dual-world RW Obstacle-free Path (DROP).}

\noindent \textbf{Given:} 
Loco-state space $ST$, MIL cost $\text{MIL}(\cdot,\cdot)$ between neighboring loco-states, start loco-state $st_{\text{s}}$, destination location $\gamma^{\text{v}}_{\text{t}} \in \Gamma^{\text{v}}$, and RW cost constraint $C$.

\noindent \textbf{Find:} An RW path $p^\ast$ from $st_{\text{s}}$ to $\gamma^{\text{v}}_{\text{t}}$ with $\text{c}(p^\ast) \leq C$ such that $\text{l}(p^\ast)$ is minimized.

\vskip 0.03in
\hrule
\vskip 0.06in

\noindent Note that $ST$ depends on $G^{\text{v}}, G^{\text{p}}$ and the orientation set $\Theta$. Moreover, $p^\ast$ may end at any feasible loco-state associated with $\gamma^{\text{v}}_{\text{t}}$. In the following, we prove that DROP is NP-hard.

\begin{theorem} \label{thm:nphard}
DROP is NP-hard.
\end{theorem}

\begin{proof}
We prove this theorem with a reduction from the NP-hard 0-1 Knapsack problem (KP) \cite{Knapsack90}. Given a set of $n$ items with weights $w_1, w_2, ... w_n$, values $v_1, v_2, ... v_n$, and a capacity limit $W$, KP maximizes the total value of the selected items such that the total weight does not exceed $W$. Given a KP instance with $V = \max_i v_i$ as the maximum value, we first create a source $a_0$ and then add two virtual locations $a_i$ and $b_i$ in DROP corresponding to each item $i$ in KP, whereas the destination is $a_n$. For each element $i \leq n-1$ in KP, we construct three edges in DROP: 1)  $e^1_i = (a_i, a_{i+1})$ with length $V + 2$, 2) $e^2_i = (a_i, b_{i+1})$ with length $V-v_{i+1} + 1$, and 3) $e^3_i = (b_{i+1}, a_{i+1})$ with length 1. The p-graph is identical to the v-graph in DROP, and $\text{MIL}(st_1, st_2)$ are set as follows.
\begin{itemize}
    \item $w_{i+1}$, if the transition corresponds to $e^2_i$ for some $i$, i.e., $st_1$ and $st_2$ are $a_i$ and $b_{i+1}$, respectively;
    \item $0$, if the transition corresponds to $e^1_i$ or $e^3_i$ for some $i$;
    \item $2W$, otherwise.
\end{itemize}

The RW constraint $C$ in DROP is identical to $W$ in KP, and $\ell_w = \infty$. Any feasible solution of DROP includes a v-path and a p-path with every $a_i$ and $a_{i+1}$ either 1) connected by a direct edge $e^1_i$ or 2) connected via $b_{i+1}$, i.e., via $e^2_i$ and $e^3_i$, with an RW cost $w_{i+1}$. The above two cases correspond to dropping and selecting item $i+1$ in KP, respectively. The former contributes $V + 2$ to the total v-path length, while the latter contributes $V-v_{i+1} + 2$, or $v_{i+1}$ less than the former. Therefore, any feasible solution in the KP instance with a total value of $v^\ast$ and a total weight of $w^\ast$ is one-to-one correspondent to one feasible solution in DROP with a v-path of length $(V+2) \cdot (n-1) - v^\ast$ and a total RW cost of $w^\ast$ in the DROP instance. The theorem follows.
\end{proof}

\section{Basic dynamic programming algorithm}
\label{sec:dp}

A simple approach for DROP is to first find the shortest v-path in the virtual world via state-of-the-art approaches \cite{TA13, MQ14}, then try to follow the v-path until approaching an obstacle in the physical world, and then adapt by Reset. As this approach does not carefully examine the entangled dual worlds, the solutions are not always feasible, as illustrated in Example \ref{DROPExample}. 

In this section, therefore, we propose a basic dynamic programming algorithm, \textit{Basic DP}, as a baseline to find the optimal solution of DROP. Basic DP cautiously derives the feasible solutions with short lengths by examining the space of \textit{Dynamic Programming States (DP States)} which is defined as follows. For every valid loco-state $st \in ST$ and every possible v-path length $l$, Basic DP creates a DP state $(st, l)$ where $l$ represents the v-path length from source $st_{\text{s}}$ to $st$. Let \textit{DP cost} $\text{c}(st, l)$ represent the minimum RW path cost for $(st, l)$. We construct a \textit{transition edge} from a DP state $(st_1, l)$ to another DP state $(st_2, l+\text{l}(st_1, st_2))$ with a transition cost $\text{MIL}(st_1, st_2)$. Let $\text{N}(st)$ be the set of loco-states neighboring to $st$. We derive $\text{c}(st, l)$ as follows.

\begin{align}
\label{eq:dprelation}
\text{c}(st, l) = \min\limits_{st' \in \text{N}(st)}\text{c}(st', l - \text{l}(st', st)) + \text{MIL}(st', st)
\end{align}

Equation \eqref{eq:dprelation} captures the fact that any RW path should arrive at $st$ via a transition edge from some other neighboring loco-state $st'$. Equipped with Equation \eqref{eq:dprelation}, the DP costs for all DP states can be iteratively derived from DP states with smaller $l$ values to larger ones. Therefore, any DP state $(st, l)$ with $\text{c}(st, l) \leq C$ corresponds to a feasible RW path from $st_{\text{s}}$ to $st$. Let $D$ denote the set of all destination DP states, i.e., $D = \{ (st, l): \gamma^{\text{v}} = \gamma^{\text{v}}_t \}$. The objective of DROP is equivalent to finding $\min\limits_{\text{c}(st,l) \leq C, (st,l) \in D}l$, and the RW path can be generated by backtracking from the destination toward $st_{\text{s}}$. Different from single-world algorithms, Basic DP carefully examines the entangled dual worlds and MIL values to find the optimal solution of DROP in $O(N^2 \cdot 2^{|E^{\text{v}}|})$-time. Below, we prove the optimality and analyze the time complexity of Basic DP. The pseudocode of Basic DP is given in Algorithm \ref{alg:basicdp}.

\begin{algorithm}[t]  
    \caption{Basic Dynamic Programming Algorithm}
    \label{alg:basicdp}
    \begin{algorithmic}[1]
        \Require 
        $ST, st_{\text{s}}, \gamma^{\text{v}}_{\text{t}}, \text{MIL}(\cdot), C$
        \Ensure 
        $p^\ast$: optimal solution for DROP
        \State Construct the set of possible v-path lengths $\mathbb{L}$
        \State Construct the DP space $X^{\text{DP}}$ with $ST, \mathbb{L}$
        \For{$(st, l) \in X^{\text{DP}}$}
          \State $\text{c}(st, l) \gets \infty$
        \EndFor
        \State $\text{c}(st_{\text{s}}, 0) \gets 0$
        \For{$l \in \mathbb{L}$}
          \For{$st \in ST$}
            \For{$st' \in \text{N}(st)$}
              \If{$\text{c}(st',l-\text{l}(st,st'))+\text{MIL}(st,st') < \text{c}(st,l)$}
                 \State $\text{c}(st,l) \gets \text{c}(st',l-\text{l}(st,st'))+\text{MIL}(st,st')$
                 \State $\text{pred}(st) \gets st'$
              \EndIf
            \EndFor
          \If {$\gamma^{\text{v}} = \gamma^{\text{v}}_{\text{t}}$ and $\text{c}(st,l) \leq C$}
            \State $p^\ast \gets \text{Backtrack}(st)$
            \Return $p^\ast$
          \EndIf
          \EndFor
        \EndFor
        \Return Infeasible
    \end{algorithmic}
\end{algorithm}

\begin{algorithm}[t]  
    \caption{Backtrack($st$)}
    \label{alg:backtrack}
    \begin{algorithmic}[1]
        \Require $st$
        \Ensure RW path $p$
        \State $p \gets \emptyset$
        \State ThisState $\gets st$
        \While {ThisState $\neq st_{\text{s}}$}
          \State Add ThisState to $p$
          \State ThisState $\gets$ Predecessor(ThisState)
        \EndWhile
        \Return $p$
    \end{algorithmic}
\end{algorithm}

\para{Optimality.} For the correctness of Equation \eqref{eq:dprelation}, if Equation \eqref{eq:dprelation} does not hold for some DP state $(st,l)$, i.e., there exists an RW path $p^\ast$ from $st_{\text{s}}$ to $st$ with total RW cost $\text{c}(st, l) < \min\limits_{st' \in \text{N}(st)}\text{c}(st', l - \text{l}(st, st')) + \text{MIL}(st, st')$. Let $st'' \in \text{N}(st)$ be the previous one of the last loco-state on $p^\ast$. By definition, the RW cost along the RW path $p^\ast$ from $st_{\text{s}}$ to $st''$ is at least $\text{c}(st'', l - \text{l}(st, st''))$. Therefore, we have
\begin{align*}
    & \text{c}(st'', l - \text{l}(st, st'')) + \text{MIL}(st, st'')\\
    \leq & \text{c}(st, l)\\
    < & \min\limits_{st' \in \text{N}(st)}\text{c}(st', l - \text{l}(st, st')) + \text{MIL}(st, st')\\
    \leq & \text{c}(st'', l - \text{l}(st, st'')) + \text{MIL}(st, st''),
\end{align*}
leading to a contradiction.

\para{Time Complexity.} The number of possible v-path lengths is $O(2^{|E^{\text{v}}|})$. Basic DP generates $O(N \cdot 2^{|E^{\text{v}}|})$ DP states, and finding the total RW cost for one DP state involves $O(N)$-time. Therefore, the total complexity is $O(N^2 \cdot 2^{|E^{\text{v}}|})$.

\section{Dual Entangled World Navigation Algorithm}
\label{sec:algo}

\begin{table}[t!]
\caption{Notations used in Section \ref{subsec:dwsp} and \ref{subsec:rpgp}.}
\begin{center}
\begin{tabular}{|c|l|}\hline
Symbol & Description \\\hline\hline
$\Gamma^{\text{v}}, \Gamma^{\text{p}}$ & virtual and physical location sets\\\hline
$\gamma^{\text{v}}_{\text{s}}$ & start virtual location\\\hline
$\gamma^{\text{v}}_{\text{t}}$ & destination virtual location\\\hline
$G^{\text{v}}, G^{\text{p}}$ & virtual and physical graphs\\\hline
$e^{\text{v}}, e^{\text{p}}$ & virtual and physical edges\\\hline
$\text{l}(e^{\text{v}})$ & virtual edge length\\\hline
$st^{\text{v}} = (\gamma^{\text{v}}, \theta^{\text{v}})$ & virtual state (v-state)\\\hline
$st^{\text{p}} = (\gamma^{\text{p}}, \theta^{\text{p}})$ & physical state (p-state)\\\hline
$\Theta$ & orientation set\\\hline
$st$ & {locomotion state (loco-state)}\\\hline
$ST$ & {loco-state space}\\\hline
$st_{\text{s}}$ & start loco-state in DROP\\\hline
\multirow{2}{*}{$\text{dist}(st_1,st_2)$} & Euclidean distance \\ & between loco-states\\\hline
$\text{MIL}(st_1,st_2)$ & MIL between neighboring loco-states\\\hline
$p$ & Redirected Walking path (RW path)\\\hline
$p^{\text{v}}$, $p^{\text{p}}$ & virtual and physical path (v/p-path)\\\hline
${\text{l}}(p)$, ${\text{c}}(p)$ & RW path length and cost\\\hline
$C$ & RW cost constraint\\\hline
$r$ & Lagrange multiplier in LR-DROP\\\hline
$r^\ast$ & optimal $r$ in LR-DROP\\\hline
$\langle \alpha(l), \beta(l) \rangle$ & MIL Range for v-edge length $l$\\\hline
$\alpha(l)$ & MIL lower bound for v-edge length $l$\\\hline
$\beta(l)$ & MIL upper bound for v-edge length $l$\\\hline
$\alpha(p^{\text{v}})$ & aggregated MIL lower bound for $p^{\text{v}}$\\\hline
$\beta(p^{\text{v}})$ & aggregated MIL upper bound for $p^{\text{v}}$\\\hline
$r_\alpha$ & Lagrange multiplier in COS-LR-DROP\\\hline
$r_\beta$ & Lagrange multiplier in CPS-LR-DROP\\\hline
$r^\ast_\alpha$ & optimal $r_\alpha$ in COS-LR-DROP\\\hline
$r^\ast_\beta$ & optimal $r_\beta$ in CPS-LR-DROP\\\hline
$p_\alpha, p_\beta$ & current shortest feasible paths\\\hline
$q_\alpha, q_\beta$ & current min-cost infeasible paths\\\hline
$p_\alpha^{\text{temp}}, p_\beta^{\text{temp}}$ & temporary paths in CSMS\\\hline
$\mathbf{Q}$ & priority queue\\\hline
$st_{\text{t}}$ & a loco-state with virtual location $\gamma^{\text{v}}_{\text{t}}$\\\hline
$\text{f}(st_{\text{s}},st,\gamma^{\text{v}}_{\text{t}})$ & ordering function in TECO\\\hline
$\text{g}(st_{\text{s}},st)$ & AEC of $st$\\\hline
$\text{h}(st,\gamma^{\text{v}}_{\text{t}})$ & REC of $st$\\\hline
$\text{MRL}(st,\gamma^{\text{v}}_{\text{t}})$ & MRL of $\gamma^{\text{v}}_{\text{t}}$\\\hline
$\text{MRC}(st,\gamma^{\text{v}}_{\text{t}})$ & MRC of $\gamma^{\text{v}}_{\text{t}}$\\\hline
$p_{\text{remain}}$ & remaining v-path\\\hline
${\text{d}_\text{s}}(st^{\text{p}})$ & distance to physical obstacles\\\hline
${\text{d}_\text{a}}(st^{\text{v}})$ & total distance to $\gamma^{\text{v}}_{\text{s}}$ and $\gamma^{\text{v}}_{\text{t}}$\\\hline
$\mathbf{Q}'$ & tie-breaking loco-states\\\hline
\end{tabular}
\label{tab:symbols_algo1}
\end{center}
\end{table}

\begin{table}[t!]
\caption{Notations used in Section \ref{subsec:ppnp} and \ref{subsec:theo}.}
\begin{center}
\begin{tabular}{|c|l|}\hline
Symbol & Description \\\hline\hline
$\text{l}_\text{l}(st_{\text{s}},st)$ & path length in shortest RW path\\\hline
$\text{c}_\text{l}(st_{\text{s}},st)$ & path cost in shortest RW path\\\hline
$\text{pred}_\text{l}(st)$ & predecessor state in shortest RW path\\\hline 
$\text{l}_\text{c}(st_{\text{s}},st)$ & path length in min-cost RW path\\\hline
$\text{c}_\text{c}(st_{\text{s}},st)$ & path cost in min-cost RW path\\\hline
$\text{pred}_\text{c}(st)$ & predecessor state in min-cost RW path\\\hline 
$\text{c}^{\alpha}_{\min}(\gamma^{\text{v}}_1,\gamma^{\text{v}}_2)$ & minimum path cost in COS-DROP\\\hline
$\text{c}^{\beta}_{\min}(\gamma^{\text{v}}_1,\gamma^{\text{v}}_2)$ & minimum path cost in CPS-DROP\\\hline
$\text{l}_{\min}(\gamma^{\text{v}}_1,\gamma^{\text{v}}_2)$ & lower bound of feasible path length\\\hline
$\tilde{L}$ & current best reference path length \\\hline
$S$ & scaling parameter \\\hline
${\text{DROP}_\text{X}}$ & post-rounding DROP problem\\\hline
$X$ & post-rounding loco-state space\\\hline
$\underline{L}$ & lower bound of optimal path length \\\hline
$p^\ast$ & optimal RW path\\\hline
$\text{l}_\text{X}(p)$ & v-path length of $p$ in ${\text{DROP}_\text{X}}$ \\\hline
$\epsilon$ & approximation parameter \\\hline

\end{tabular}
\label{tab:symbols_algo2}
\end{center}
\end{table}

In investigation of Basic DP, we observe three types of loco-states that can be avoided: 1) those with v-states far away from the source and destination in the v-graph (unlikely to create short v-paths); 2) those with p-states near the physical boundaries and obstacles (hard to generate feasible RW paths); 3) intermediate loco-states with insufficient RW budget to find a v-path shorter than the best intermediate feasible solution obtained during processing. Therefore, we propose the \textit{Dual Entangled World Navigation (DEWN)} algorithm, which 1) quickly generates a reference path (i.e., a feasible solution) by problem transformation techniques and a novel ordering strategy; 2) leverages the reference path to filter redundant loco-states via several pruning strategies; 3) adopts dynamic programming on the dramatically trimmed solution space to ensure the approximation guarantee.

\begin{figure}[t]
    \centering
    \includegraphics[width = 0.97\columnwidth]{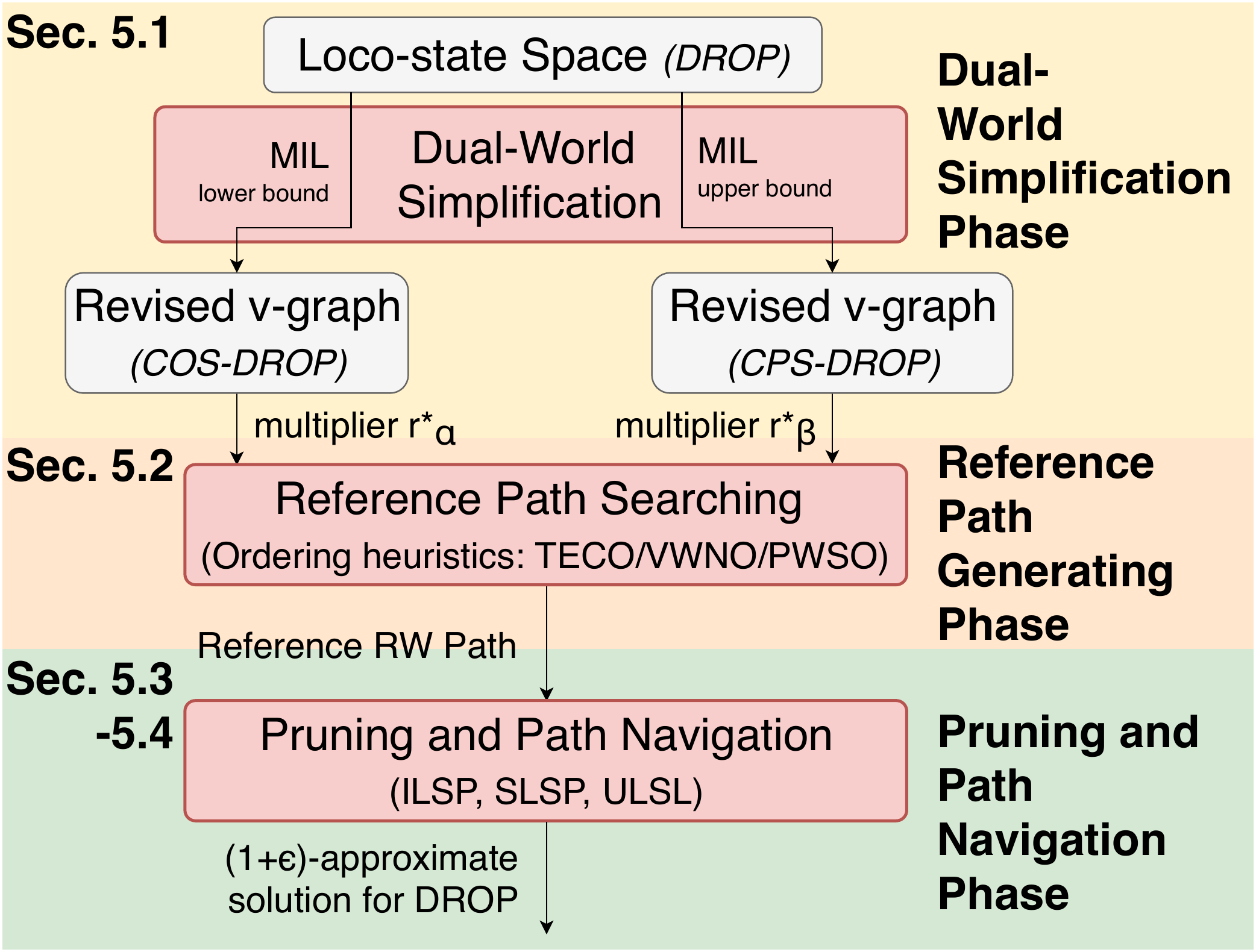}
    \caption{System model of DEWN.}
    \label{fig:flow_chart}
\end{figure}

DEWN consists of three phases as illustrated in Figure \ref{fig:flow_chart}. As it is computationally expensive to find a reference path \textit{directly} from the loco-state space, in \textit{Dual World Simplification Phase} (Section \ref{subsec:dwsp}), we exploit the precomputed \textit{MIL Range} to transform the dual-world DROP problem into two \textit{single-world} problems, COS-DROP and CPS-DROP, respectively, by incorporating the MIL lower and upper bounds as new edge weights of the v-graph to find corresponding v-paths. These problems are then further reduced into Lagrangian relaxed problems where the weighted sum of the v-path length and MIL upper/lower bounds are jointly minimized with Lagrange multipliers $r_\alpha$ and $r_\beta$ as their weights, respectively. We present an efficient algorithm to find the best multipliers $r^\ast_\alpha$ and $r^\ast_\beta$.

Next, \textit{Reference Path Generation Phase} (Section \ref{subsec:rpgp}) exploits $r^\ast_\alpha$ and $r^\ast_\beta$ to find a reference RW path quickly with a new ordering strategy tailored for dual-world path finding that balances the remaining RW cost and v-path distance to the destination. Equipped with the reference RW path, \textit{Pruning and Path Navigation Phase} (Section \ref{subsec:ppnp}) effectively trims off redundant candidate loco-states that incur excessive RW costs and large path distances. DEWN then further applies dynamic programming with the rounding-and-scaling technique on the remaining loco-state space to retrieve an $(1+\epsilon)$-approximate RW path with significantly reduced computational cost. The notations used in this section are summarized in Table \ref{tab:symbols_algo1} and \ref{tab:symbols_algo2}, and the abbreviations are summarized in Table \ref{tab:symbols_abbrv}.

\begin{table}[t]
\caption{Abbreviations used in algorithms.}
\begin{center}
\begin{tabular}{|c|l|}\hline
Abbreviation & Full \\\hline\hline
\multirow{2}{*}{DROP} & Dual-world Redirected-walking \\ & Obstacle-free Path\\\hline
LR-DROP & Lagrange relaxation of DROP\\\hline
COS-DROP & Cost-Optimistic Simplified DROP\\\hline
CPS-DROP & Cost-Pessimistic Simplified DROP\\\hline
COS-LR-DROP & Lagrange relaxation of COS-DROP\\\hline
CPS-LR-DROP & Lagrange relaxation of CPS-DROP\\\hline
CSMS & Cost Simplified Multiplier Searching\\\hline
IDWS & Informed Dual-World Search\\\hline
AEC & Accumulated Estimated Cost\\\hline
REC & Remaining Estimated Cost\\\hline
TECO & Total Estimated Cost Ordering\\\hline
VWNO & Virtual World Naturalness Ordering\\\hline
PWSO & Physical World Safety Ordering\\\hline
ILSP & Infeasible Loco-State Pruning\\\hline
SLSP & Suboptimal Loco-State Pruning\\\hline
ULSL & Unpromising Loco-State Locking\\\hline
\end{tabular}
\label{tab:symbols_abbrv}
\end{center}
\end{table}

\subsection{Dual-World Simplification Phase}
\label{subsec:dwsp}

To strike a good balance between minimizing the v-path length and the RW cost of the reference RW path, the Lagrangian relaxation (LR) problem of DROP, called \textit{LR-DROP}, is defined as follows.

\vskip 0.06in
\hrule
\vskip 0.03in
\noindent \textbf{Problem: LR-DROP.}

\noindent \textbf{Given:} A DROP instance and a Lagrange multiplier $r>0$.

\noindent \textbf{Find:} An RW path $p^\ast$ from $st_{\text{s}}$ to $\gamma^{\text{v}}_{\text{t}}$ to minimize $\text{l}(p^\ast) + r \cdot \text{c}(p^\ast)$.
\vskip 0.03in
\hrule
\vskip 0.06in

This new problem incorporates the constraint on RW cost into the objective via the Lagrange multiplier $r$. Intuitively, with a small $r$, the optimal solution in LR-DROP tends to favor shorter v-paths instead of lower RW costs. In contrast, a feasible solution (in the original problem) is easier to be found by solving LR-DROP with large values of $r$, as manifested in the following property:

\begin{property}
\label{property:lagrange}
Let $p^\ast_1$ and $p^\ast_2$ be the optimal RW paths of LR-DROP with multipliers $0 \leq r_1 < r_2$. Then $\text{l}(p^\ast_1) \leq \text{l}(p^\ast_2)$ and $\text{c}(p^\ast_1) \geq \text{c}(p^\ast_2)$. 
\end{property}

\begin{proof}
Since $p^\ast_1$ and $p^\ast_2$ are optimal, we have
\begin{align}
\text{l}(p^\ast_1) + r_1 \cdot \text{c}(p^\ast_1) &\leq \text{l}(p^\ast_2) + r_1 \cdot \text{c}(p^\ast_2), \label{proof_prop1_eq1}\\
\text{l}(p^\ast_2) + r_2 \cdot \text{c}(p^\ast_2) &\leq \text{l}(p^\ast_1) + r_2 \cdot \text{c}(p^\ast_1).
\end{align}
By summing up the two inequalities,
\begin{align*}
r_1 \cdot \text{c}(p^\ast_1) + r_2 \cdot \text{c}(p^\ast_2) &\leq r_2 \cdot \text{c}(p^\ast_1) + r_1 \cdot \text{c}(p^\ast_2),\\
(r_2 - r_1) \cdot \text{c}(p^\ast_2) &\leq (r_2 - r_1) \cdot \text{c}(p^\ast_1).
\end{align*}
Since $r_1 < r_2$, $\text{c}(p^\ast_1) \geq \text{c}(p^\ast_2)$, and $\text{l}(p^\ast_1) \leq \text{l}(p^\ast_2)$ from Equation \eqref{proof_prop1_eq1}.
\end{proof}

An excellent reference path would be one generated with a small $r$ while complying with the RW cost constraint. Although the LARAC algorithm \cite{AJ01} is effective in approaching the optimal LR-based solution for the constrained shortest path problem, it is too computationally expensive for the dual-world DROP.\footnote{Solving LR-DROP for each $r$ requires $O(N \cdot \log N)$ time, and there are $O(N \cdot \log^3 N)$ iterations to find the optimal $r$, where $N = |ST|$.} Inspired by the fact that traditional LR-based algorithms are only practical in single-world problems, our idea is to first simplify the problem via MIL Range, and then estimate the multiplier through investigating the simplified problems on the much smaller v-graph.

\para{Dual-World Simplification.} We aim to search $r$ in the transformed v-graph, instead of in the loco-state space. For each possible v-edge length $l$, we derive its \textit{MIL Range} $(\alpha(l), \beta(l))$ as follows.

\begin{align*}
    \alpha(l) &= \min\limits_{ \substack{st_1, st_2 \in ST\\ \text{dist}(st_1, st_2) = l} } \text{MIL}(st_1, st_2)\\
    \beta(l) &= \max\limits_{st_1 \in ST} \min\limits_{ \substack{st_2 \in ST\\ \text{dist}(st_1, st_2) = l} } \text{MIL}(st_1, st_2)
\end{align*}

The MIL lower bound $\alpha(l)$ is the smallest possible RW cost to realize a v-edge of length $l$ in the physical world, as it takes the minimum RW cost among all loco-state pairs $(st_1,st_2)$. In contrast, the MIL upper bound $\beta(l)$ is the maximum required RW cost to realize such a v-edge \textit{starting from any fixed loco-state.} More specifically, given $st_1$, the smallest possible RW cost to realize a v-edge of length $l$ would be $\min_{ st_2 \in ST, \text{dist}(st_1, st_2) = l} \text{MIL}(st_1, st_2)$, and $\beta(l)$ takes the maximum value among all $st_1$. For each v-path $p^\text{v} = \langle e_1, e_2, \cdots , e_n \rangle$, MIL Range helps finding the range of the total RW cost along $p^\text{v}$ in the following theorem.

\begin{theorem} \label{thm:bounded_cost}
There exists a p-path $p^\text{p}$ realizing $p^\text{v}$ with a total RW cost bounded by $\alpha(p^\text{v}) = \sum\limits_{i=1}^n \alpha(\text{l}(e_i)) \leq \text{c}(p) \leq \sum\limits_{i=1}^n \beta(\text{l}(e_i)) = \beta(p^\text{v})$.
\end{theorem}

\begin{proof}
We first prove the lower bound. Since every edge $e_i$ in the v-path $p^\text{v}$ incurs at least an RW cost $\alpha(\text{l}(e_i))$, the total RW cost along $p^\text{p}$ is at least $\sum_{i=1}^n \alpha(\text{l}(e_i))$. Thus, $p^\text{v}$ is not feasible when $\alpha(p^\text{v}) > C$. For the upper bound, to build an RW-realized p-path from $p^\text{v}$, a simple approach iteratively selects the next loco-state by choosing the next p-state with the smallest RW cost. Since the cost of $e_i$ does not exceed $\beta(\text{l}(e_i))$, the total RW cost is at most $\sum_{i=1}^n \beta(\text{l}(e_i))$. If it does not exceed $C$, there exists at least one feasible $p^\text{p}$. The theorem follows.
\end{proof}

Note that $\alpha(l)$ for a v-edge length $l$ refers to the MIL lower bound value of $l$, while $\alpha(p^\text{v})$ for a v-path $p^\text{v}$ is the aggregate of MIL lower bound values for the v-edges along $p^\text{v}$. A v-path $p^\text{v}$ is \textit{feasible} if $\beta(p^\text{v}) \leq C$ and is able to act as a reference path in the later phases. In contrast, a v-path $p^\text{v}$ is \textit{infeasible} if $\alpha(p^\text{v}) > C$. Accordingly, we formulate DROP for the transformed v-graph as \textit{Cost-Optimistic} and \textit{Cost-Pessimistic} versions, corresponding to the MIL lower and upper bounds, respectively.

\vskip 0.06in
\hrule
\vskip 0.03in
\noindent \textbf{Problem: Cost-Optimistic Simplified DROP (COS-DROP).}

\noindent \textbf{Given:} A DROP instance.

\noindent \textbf{Find:} A v-path $p^{\text{v}}$ from $\gamma^{\text{v}}_{\text{s}}$ (the virtual location of $st_{\text{s}}$) to $\gamma^{\text{v}}_{\text{t}}$, so that $\text{l}(p^\text{v})$ is minimized, and $\sum_{e \in p^{\text{v}}} \alpha(\text{l}(e)) \leq C$.

\vskip 0.03in
\hrule
\vskip 0.06in

\noindent Analogous to LR-DROP, the LR problem of COS-DROP (called COS-LR-DROP) incorporates a multiplier $r_\alpha>0$.

\vskip 0.06in
\hrule
\vskip 0.03in

\noindent \textbf{Problem: COS-LR-DROP.}

\noindent \textbf{Given:} A DROP instance, and a multiplier $r_\alpha>0$.

\noindent \textbf{Find:} \sloppy A v-path $p^{\text{v}}$ from $\gamma^{\text{v}}_{\text{s}}$ to $\gamma^{\text{v}}_{\text{t}}$ where $\text{l}(p^{\text{v}}) + r_\alpha \cdot \sum_{e \in p^{\text{v}}} \alpha(\text{l}(e))$ is minimized.

\vskip 0.03in
\hrule
\vskip 0.06in

Similarly, \textit{Cost-Pessimistic Simplification} of DROP (CPS-DROP), corresponding to the MIL upper bound, is formulated by replacing $\alpha(\text{l}(e))$ with $\beta(\text{l}(e))$, and its LR problem, CPS-LR-DROP, is associated with multiplier $r_\beta$. 

\begin{figure}[t]
\centering
  \begin{subfigure}[b]{0.45\columnwidth}
    \includegraphics[width=\columnwidth]{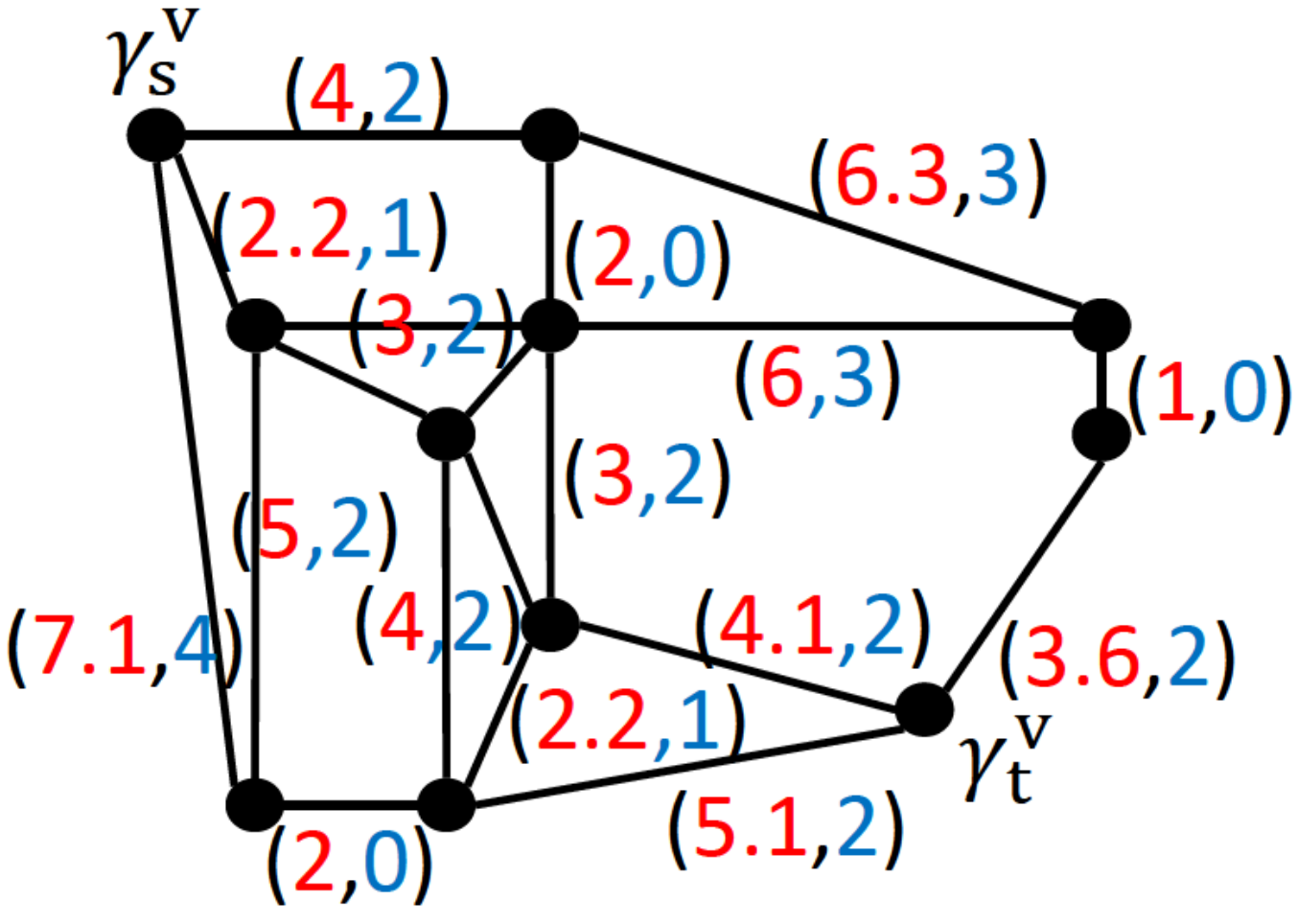}
    \caption{COS-DROP.}
    \label{fig:cos-drop}
  \end{subfigure}
  \hfill 
  \begin{subfigure}[b]{0.45\columnwidth}
    \includegraphics[width=\columnwidth]{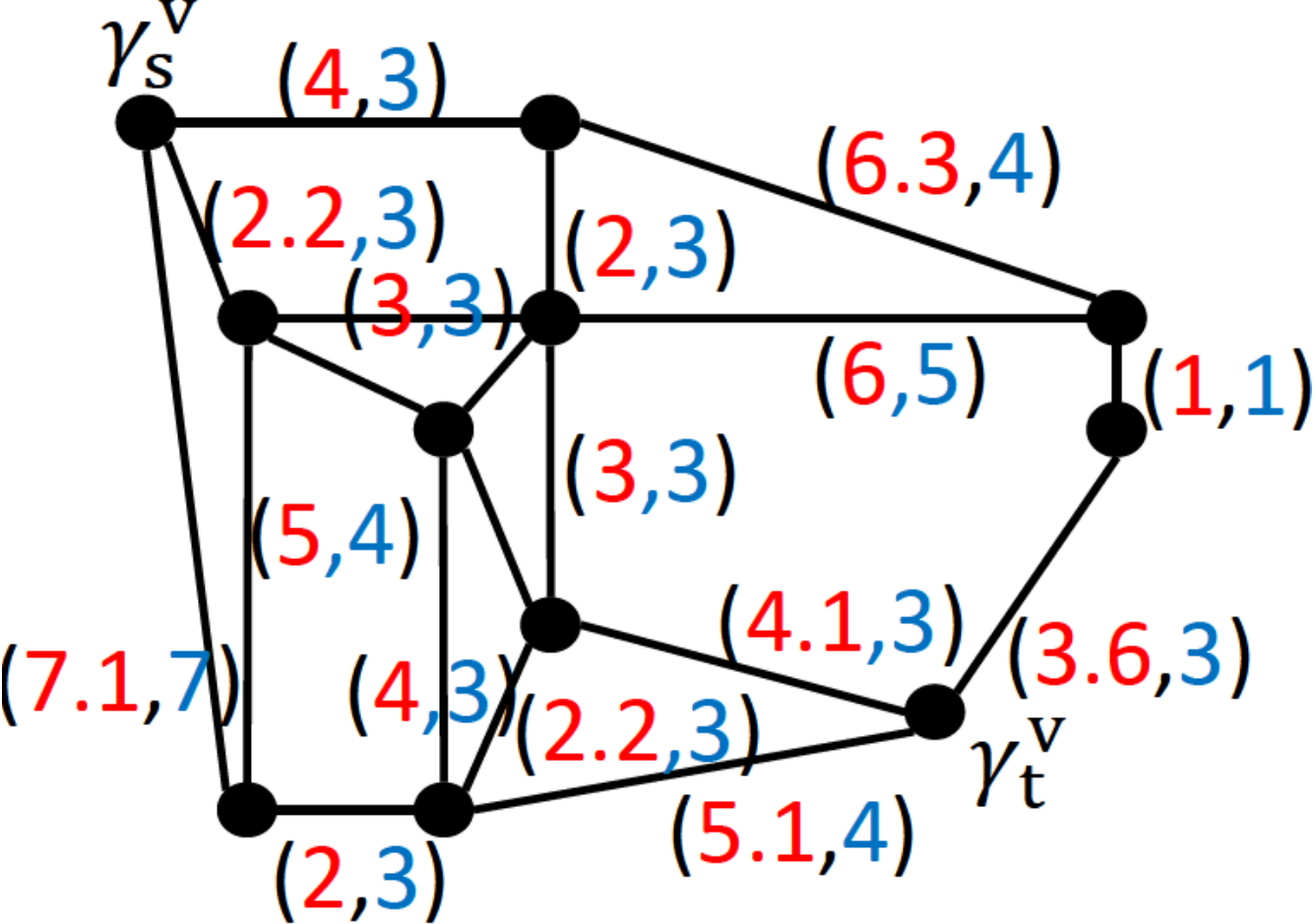}
    \caption{CPS-DROP.}
    \label{fig:cps-drop}
  \end{subfigure}\\
  
  \begin{subfigure}[b]{0.42\columnwidth}
    \includegraphics[width=\columnwidth]{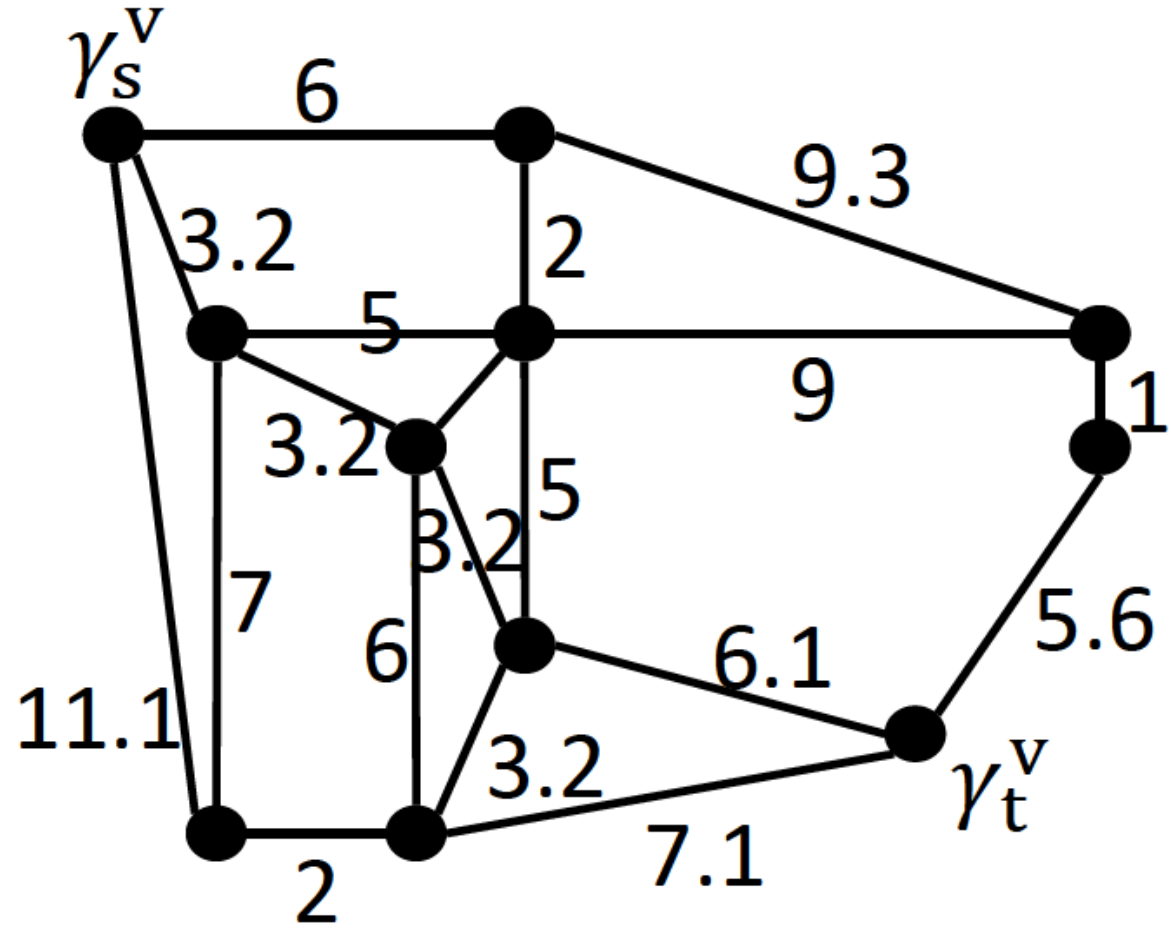}
    \caption{\scriptsize{COS-LR-DROP, $r_\alpha=1$.}}
    \label{fig:cos-lr-drop}
  \end{subfigure}
  \hfill 
  \begin{subfigure}[b]{0.45\columnwidth}
    \includegraphics[width=\columnwidth]{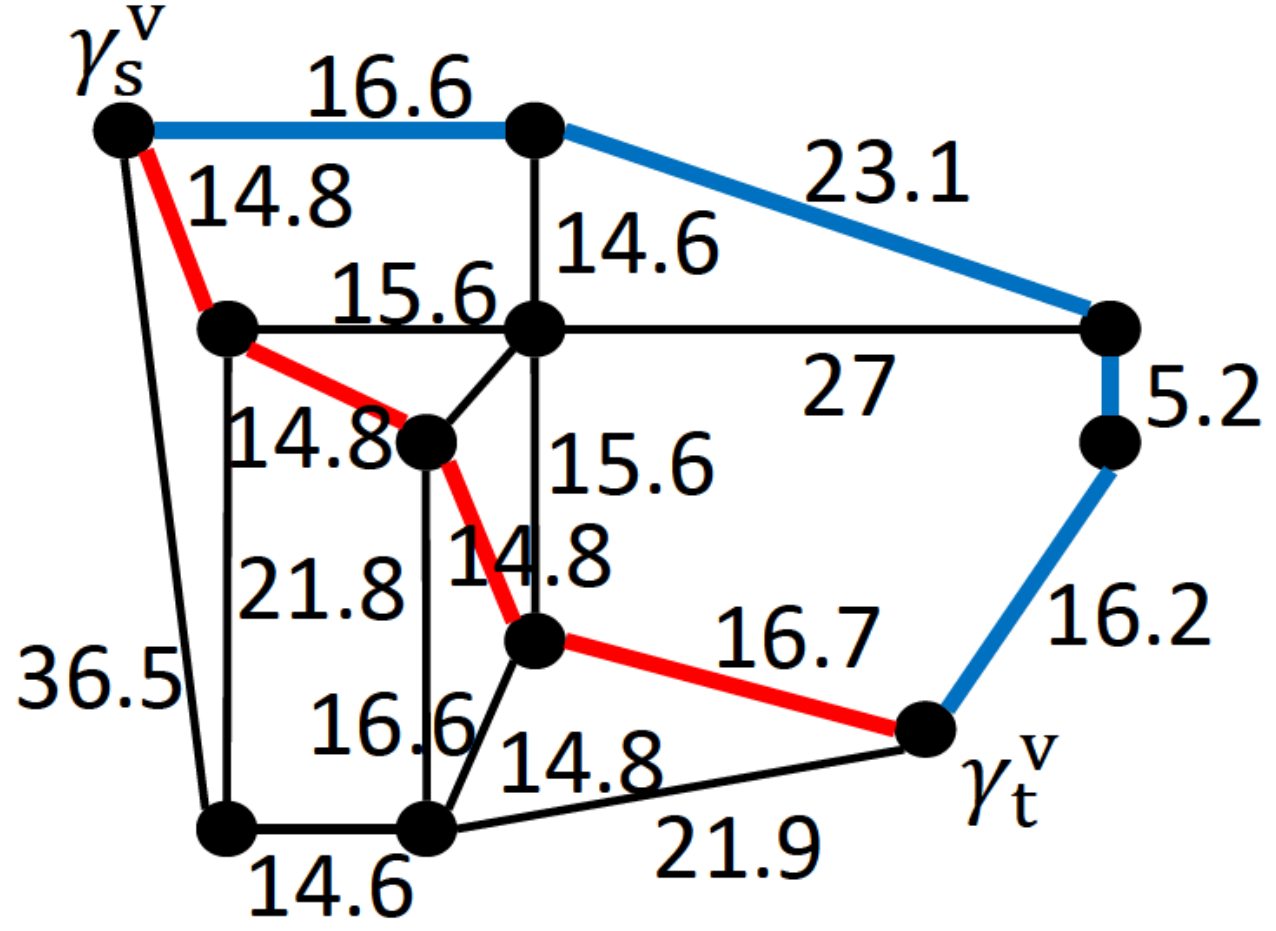}
    \caption{\scriptsize{CPS-LR-DROP, $r_\beta=4.2$.}}
    \label{fig:CSMS-2}
  \end{subfigure}\\
  
  \begin{subfigure}[b]{0.45\columnwidth}
    \includegraphics[width=\columnwidth]{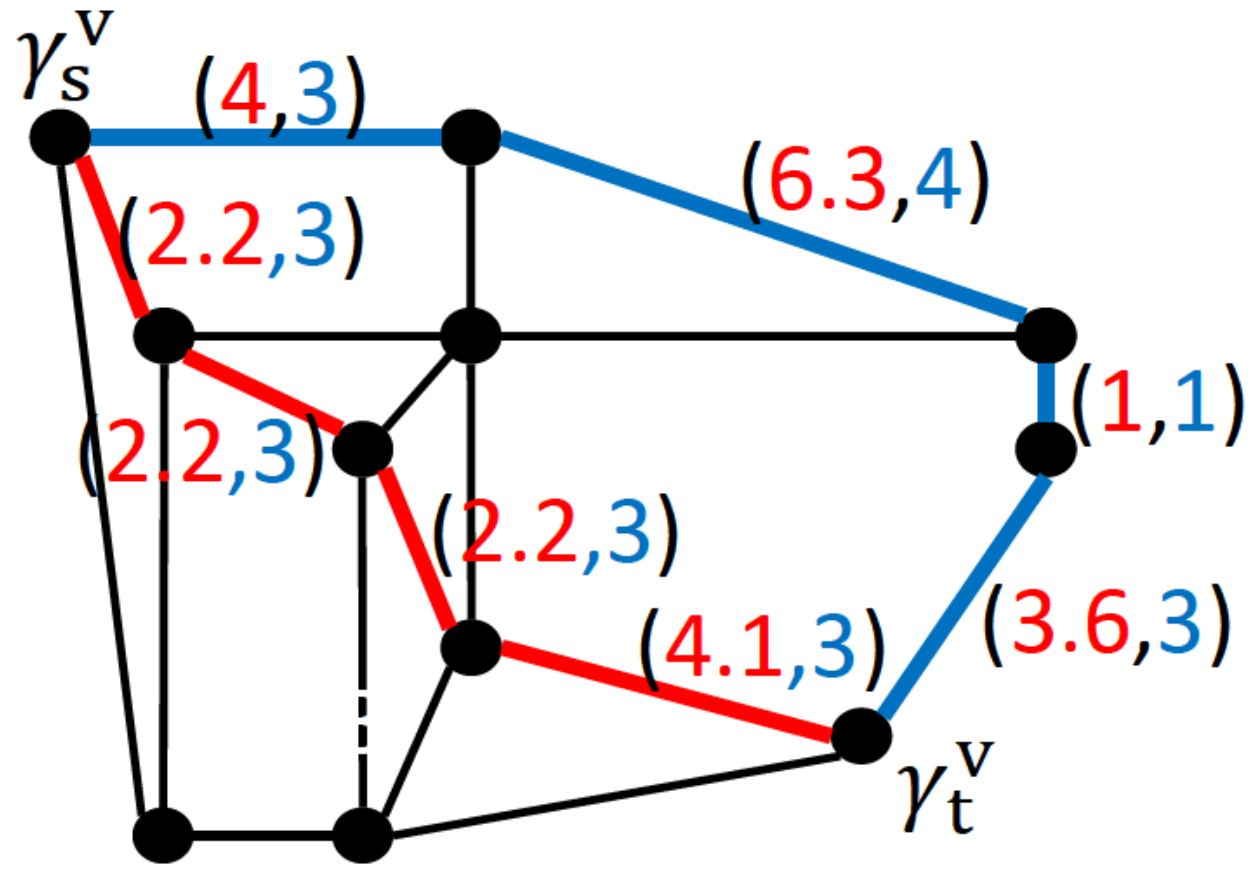}
    \caption{CSMS on CPS-DROP.}
    \label{fig:CSMS-1}
  \end{subfigure}
  \hfill 
  \begin{subfigure}[b]{0.45\columnwidth}
    \includegraphics[width=\columnwidth]{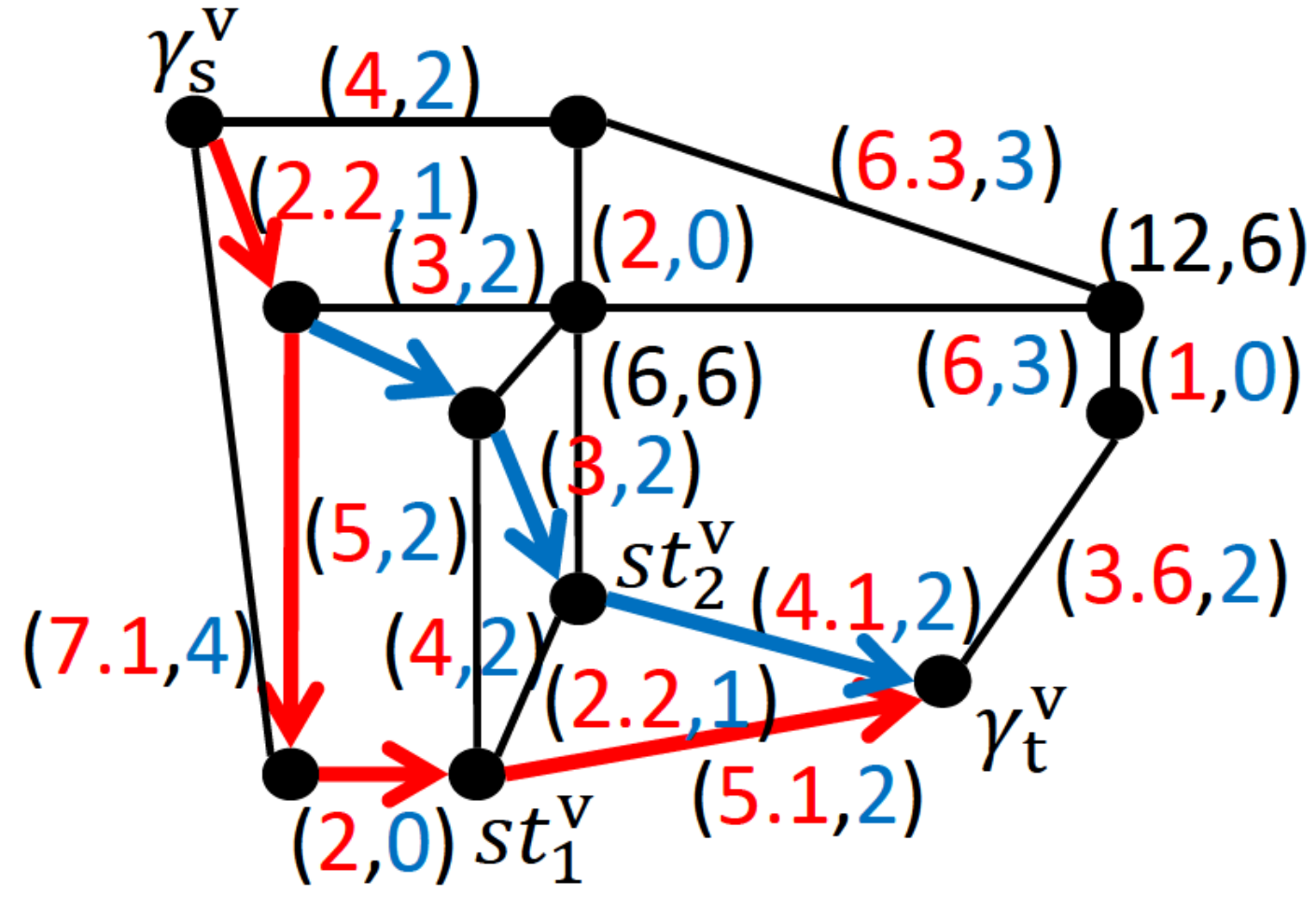}
    \caption{Example of pruning.}
    \label{fig:pruning_example}
  \end{subfigure}
\caption{Running example.}
\label{fig:example}
\end{figure}

\begin{table}[t] 
\caption{Precomputed MIL Range values.}
\label{table:milrange}
\resizebox{\columnwidth}{!}{%
\begin{tabular}{|l|l|l|l|l|l|l|l|l|l|l|l|l|l|}
\hline
$l$    & 1 & 1.4 & 2 & 2.2 & 3 & 3.6 & 4 & 4.1 & 5 & 5.1 & 6 & 6.3 & 8.1 \\ \hline
$\alpha(l)$ & 0 & 0    & 0 & 1    & 2 & 2    & 2 & 2    & 2 & 2   & 3 & 3    & 4    \\ \hline
$\beta(l)$ & 1 & 2    & 3 & 3    & 3 & 3    & 3 & 3    & 4 & 4   & 4 & 4    & 7    \\ \hline
\end{tabular}}
\end{table}

\begin{example}
Figures \ref{fig:cos-drop} and \ref{fig:cps-drop} present the COS/CPS-DROP instances of Example \ref{DROPExample} with MIL Ranges (computed from the MIL between loco-states) listed in Table \ref{table:milrange}. The tuple beside each v-edge describes the edge length (in red) and the MIL lower/upper bound values (in blue). The v-edge lengths are identical in Figures \ref{fig:cos-drop} and \ref{fig:cps-drop}, but the estimated RW cost, i.e., MIL upper/lower bound values, is larger in Figure \ref{fig:cps-drop}. Figure \ref{fig:cos-lr-drop} illustrates the COS-LR-DROP instance obtained from COS-DROP with $r_\alpha=1$. For the top-left v-edge, the weighted sum of the edge length and RW cost in COS-LR-DROP is $4+1 \cdot 2 = 6$. Similarly, Figure \ref{fig:CSMS-2} shows a CPS-LR-DROP instance with $r_\beta=4.2$. \hfill \qedsymbol
\end{example}

\begin{algorithm}[t]  
    \caption{Cost Simplified Multiplier Searching (CSMS)}
    \label{alg:CSMS}
    \begin{algorithmic}[1]
        \Require 
        $G^{\text{v}}, s, t \in \Gamma^{\text{v}}, C, \alpha(\cdot), \beta(\cdot)$
        \Ensure 
        $r^\ast_\alpha, r^\ast_\beta$: Lagrange parameters
        \State $p \gets \text{Dijkstra}(s,t,l)$
        \If{$\beta(p) \leq C$}
          \Return Optimal
        \EndIf
        \State $p_\alpha \gets \text{Dijkstra}(s,t,\alpha(l))$, $p_\beta \gets \text{Dijkstra}(s,t,\beta(l))$
        \State $q_\alpha \gets p, q_\beta \gets p$
        \If{$\alpha(p_\alpha) > C$}
          \Return Infeasible
        \EndIf
        \For{$i \in \{ \alpha, \beta \}$}
          \While{True}
            \State $r_i \gets \frac{\text{l}(q_i) - \text{l}(p_i)}{i(p_i) - i(q_i)}$
            \State $x_i \gets \text{Dijkstra}(s,t,l+r_i \cdot i(l))$
            \If{$x_i = p_i$ or $x_i = q_i$}
              \State $r^\ast_i \gets r_i$
              \State \textbf{break}
            \EndIf
            \If{$i(x_i) \leq C$} 
              \State $p_i \gets x_i$
            \Else
              \State $q_i \gets x_i$ 
            \EndIf
          \EndWhile
        \EndFor
        \Return $r^\ast_\alpha, r^\ast_\beta$ 
    \end{algorithmic}
\end{algorithm}

We then present \textit{Cost Simplified Multiplier Searching (CSMS)} (Algorithm \ref{alg:CSMS}), which can be viewed as generalizing the LARAC algorithm on simplified dual worlds, to find the optimal $r^\ast_\alpha$ for COS-DROP and the optimal $r^\ast_\beta$ for CPS-DROP. CSMS maintains two v-paths $p_\alpha$ and $q_\alpha$. $p_\alpha$ is initialized as the v-path from $\gamma^{\text{v}}_{\text{s}}$ to $\gamma^{\text{v}}_{\text{t}}$ with the minimum RW cost, i.e., the optimal v-path in COS-LR-DROP with $r=\infty$. $q_\alpha$ is initialized as the shortest v-path from $\gamma^{\text{v}}_{\text{s}}$ to $\gamma^{\text{v}}_{\text{t}}$, i.e., the optimal v-path in COS-LR-DROP with $r=0$ (usually not feasible). The above two paths can be found by Dijkstra's algorithm on v-graph (instead of from the large loco-state space $ST$). The initial (and trivial) knowledge is that the optimal multiplier lies in $[0, \infty)$, which is the \textit{possible region} for the best multiplier $\tilde{r_\alpha}$. 

CSMS iteratively 1) updates $r_\alpha = \frac{\text{l}(q_\alpha) - \text{l}(p_\alpha)}{\alpha(p_\alpha) - \alpha(q_\alpha)}$, where $\alpha(p) = \sum_{e \in p} \alpha(\text{l}(e))$, 2) finds the optimal v-path $p^{\text{temp}}_\alpha$ in COS-LR-DROP with $r=r_\alpha$, and 3) examines if $p^{\text{temp}}_\alpha$ is feasible to COS-DROP. If it is feasible, the optimal multiplier leading to the shortest feasible RW path is greater than 0 but smaller than $r_\alpha$. CSMS thereby replaces $p_\alpha$ with $p^{\text{temp}}_\alpha$ to decrease $r_\alpha$ in the next iteration to search for a shorter v-path. Otherwise, $q_\alpha$ is substituted by $p^{\text{temp}}_\alpha$ to increase the multiplier in the next iteration. The above process stops when $p^{\text{temp}}_\alpha$ equals one of $p_\alpha$ or $q_\alpha$, and it returns $r_\alpha$ as the optimal multiplier  $r^\ast_\alpha$ for COS-DROP. $r^\ast_\beta$ for CPS-DROP is optimized analogously, as illustrated below.\footnote{To accelerate the process, an alternative is to adopt an early termination rule: simply terminate when the current value of $r_\alpha$ cannot be increased or decreased by a small ratio $\delta$.} 

As mentioned earlier, DWSP repeats CSMS for COS-LR-DROP and CPS-LR-DROP. Therefore it passes two candidates of multiplier, $r^\ast_\alpha$ and $r^\ast_\beta$, to the next phase RPGP. Note that here $|G^v|$ is tiny compared with the number of loco-states. Hence, finding nice multipliers with CSMS is significantly more efficient than directly applying the existing LARAC algorithm.

\begin{example} \label{CSMS_example}
Figure \ref{fig:CSMS-1} finds $r^\ast_\beta$ for the CPS-DROP instance in Figure \ref{fig:cps-drop}. In the first iteration, v-path $p_\beta$ is the blue one with length 14.9 and estimated RW cost 11. V-path $q_\beta$ is the red path with length 10.7 and estimated RW cost 12. CSMS then updates $r_\beta = \frac{10.7-14.9}{11-12} = 4.2$. Afterwards, since the shortest path is exactly $p_\beta$ and $q_\beta$ (both with aggregated cost 61.1) in Figure \ref{fig:CSMS-2}, $p^{\text{temp}}_\beta$ is either $p_\beta$ or $q_\beta$. Thus CSMS terminates with the optimal multiplier $r^\ast_\beta = 4.2$. \hfill \qedsymbol
\end{example} 

\subsection{Reference Path Generation Phase}
\label{subsec:rpgp}

Since CSMS only finds v-paths, we leverage $r^\ast_\alpha$ and $r^\ast_\beta$ to find the reference path $p^\ast$ in the corresponding LR-DROP instances. Specifically, because CPS-DROP considers the worst-case RW cost for each v-edge, any v-path feasible to CPS-DROP is also feasible to DROP. Consequently, the optimal RW path for LR-DROP with $r = r^\ast_\beta$ is feasible. On the other hand, as $r^\ast_\alpha$ is obtained by an \textit{optimistic} estimate of the RW costs, the the optimal RW-path for LR-DROP with $r^\ast_\alpha$ tends to be shorter but may not be feasible. Thus, we solve LR-DROP for both $r = r^\ast_\alpha$ and $r = r^\ast_\beta$ and return the better (shorter) feasible RW path as the reference path.

To solve LR-DROP with any multiplier $r$, a simple approach is to associate each edge $(st_1, st_2)$ with an \textit{LR cost} $\text{l}((st_1, st_2)) + r \cdot \text{MIL}(st_1, st_2)$ and apply Dijkstra's algorithm in $O(N \cdot \log N)$ time. However, it is again computationally expensive for a large $N = |ST|$. In contrast, we propose \textit{Informed Dual-World Search} (IDWS), which maintains a priority queue $\mathbf{Q}$ to store the loco-states on the boundaries of the visited area. Initially, $\mathbf{Q}$ contains only the start loco-state $st_{\text{s}}$. The algorithm pops one loco-state $st$ from $\mathbf{Q}$ according to the ordering strategies (detailed later) and \textit{expands} $st$ by pushing all unvisited neighboring loco-states of $st$ to $\mathbf{Q}$. 

Moreover, IDWS derives the \textit{Accumulated Estimated Cost} (AEC) $\text{g}(st_{\text{s}},st)$ and \textit{Remaining Estimated Cost} (REC) $\text{h}(st,\gamma^{\text{v}}_{\text{t}})$ upon reaching each loco-state $st$.\footnote{Note that $st_{\text{s}}$ is the start loco-state, and $\gamma^{\text{v}}_{\text{t}}$ is the destination virtual location, i.e., they are fixed variables for comprehensive representations.} $\text{g}(st_{\text{s}},st)$ is the current aggregated LR cost from $st_{\text{s}}$ to $st$ in LR-DROP. Therefore, if a loco-state $st_2$ is reached from expanding $st_1$, $\text{g}(st_{\text{s}},st_2) = \text{g}(st_{\text{s}},st_1) + \text{l}(st_1, st_2) + r \cdot \text{MIL}(st_1, st_2)$. $\text{h}(st,\gamma^{\text{v}}_{\text{t}})$ is the estimated total LR cost from $st$ to the destination virtual location $\gamma^{\text{v}}_{\text{t}}$ (detailed later). The above process repeats until the destination is reached, where IDWS then finds the corresponding RW path by backtracking from $st_{\text{t}}$ to $st_{\text{s}}$. IDWS leverages the idea of \textit{informed search} \cite{RD85} such that if IDWS is \textit{admissible}, i.e., $\text{h}(st,\gamma^{\text{v}}_{\text{t}})$ does not exceed the real total LR cost from $st$ to $\gamma^{\text{v}}_{\text{t}}$, then 1) the returned RW path is optimal to LR-DROP, and 2) the search process visits the fewest states among all algorithms.

\para{Total Estimated Cost Ordering (TECO).} \sloppy Specifically, let \textit{Minimum Remaining Length} $\text{MRL}(st,\gamma^{\text{v}}_{\text{t}})$ and \textit{Minimum Remaining Cost} $\text{MRC}(st,\gamma^{\text{v}}_{\text{t}})$ represent the lower bounds on the v-path length and RW cost from $st$ to $\gamma^{\text{v}}_{\text{t}}$, respectively. They are initialized as the exact v-path length and RW cost obtained from Dijkstra's algorithm on the transformed v-graph (instead of loco-states),\footnote{Traditional index frameworks \cite{TA13} can be incorporated to retrieve the v-path lengths but cannot be directly used for RW cost, since the users' physical worlds vary.} $\text{MRL}(st)$ is initiated as the shortest v-path length from $st$ to the destination, and $\text{MRC}(st)$ is initiated as the least RW cost from $st$ to the destination. Both values can be computed by Dijkstra's algorithm.\footnote{Note here the Dijkstra's algorithm is not computationally intensive since it only runs on the v-graph instead of the whole loco-state space. The MRL and MRC values are also stored, or offline indexed, to avoid repeated calculation. They are reused in the subsequent PPNP phase.} 
where $\text{MRL}(st,\gamma^{\text{v}}_{\text{t}})$ is derived by setting the edge cost between $\gamma^{\text{v}}_1$ and $\gamma^{\text{v}}_2$ as $\text{l}(\gamma^{\text{v}}_1, \gamma^{\text{v}}_2)$, and $\text{MRC}(st,\gamma^{\text{v}}_{\text{t}})$ is obtained by setting the edge cost as the \textit{MIL lower bound} $\alpha(\text{l}(\gamma^{\text{v}}_1, \gamma^{\text{v}}_2))$. Equipped with MRL and MRC, $\text{h}(st,\gamma^{\text{v}}_{\text{t}})$ and the ordering function $\text{f}(st_{\text{s}},st,\gamma^{\text{v}}_{\text{t}})$ in TECO are defined as follows.
\begin{align}
    \text{h}(st,\gamma^{\text{v}}_{\text{t}}) &= \text{MRL}(st,\gamma^{\text{v}}_{\text{t}}) + r \cdot \text{MRC}(st,\gamma^{\text{v}}_{\text{t}})\\ 
    \text{f}(st_{\text{s}},st,\gamma^{\text{v}}_{\text{t}}) &= \text{g}(st_{\text{s}},st) + \text{h}(st,\gamma^{\text{v}}_{\text{t}})
\end{align}
TECO is guided by AEC $\text{g}(st_{\text{s}},st)$ and REC $\text{h}(st,\gamma^{\text{v}}_{\text{t}})$ to extract the next loco-state in $\mathbf{Q}$ with the minimum $\text{f}(st_{\text{s}},st,\gamma^{\text{v}}_{\text{t}})$. Therefore, IDWS features the admissible property $\text{h}(st,\gamma^{\text{v}}_{\text{t}}) \leq \text{l}(p_{\text{remain}}) + r \cdot \text{c}(p_{\text{remain}})$ for any $p_{\text{remain}}$ from $st$ to $\gamma^{\text{v}}_{\text{t}}$, such that it generates an optimal solution to LR-DROP by exploring the fewest loco-states \cite{RD85}. 

\para{Ordering Strategies to improve user experience.} A feasible solution could be found in various orders of visiting candidate loco-states. In the following, we propose Physical World Safety Ordering (PWSO) and Virtual World Naturalness Ordering (VWNO) to generate good reference paths that enhance the user experience. PWSO prioritizes a p-state $st^{\text{p}}$ with the largest distance ${\text{d}_\text{s}}(st^{\text{p}})$ to any physical obstacle in p-graph, and VWNO prefers a v-state $st^{\text{v}}$ with the minimum total straight-line distance ${\text{d}_\text{a}}(st^{\text{v}})$ to the source and destination in v-space. When there are multiple loco-states $\mathbf{Q'} = \{ \argmin\limits_{st' \in \mathbf{Q}} \text{f}(st_{\text{s}},st',\gamma^{\text{v}}_{\text{t}}) \}$, IDWS extracts $st = \argmin\limits_{st \in \mathbf{Q'}} ({\text{d}_\text{a}}(st^{\text{v}}) - {\text{d}_\text{s}}(st^{\text{p}}))$ from $\mathbf{Q'}$ based on PWSO and VWNO, in favor of loco-states with lower ${\text{d}_\text{a}}$ and higher ${\text{d}_\text{s}}$.

\begin{algorithm}[t]
    \caption{Informed Dual-World Search (IDWS)}
    \label{alg:idws}
    \begin{algorithmic}[1]
        \Require LR-DROP instance, multiplier $r$, $\text{d}_{\text{s}}(\cdot)$, $\text{d}_{\text{a}}(\cdot)$
        \Ensure v-path length $\text{l}(p)$ and RW path $p$
        \State $\mathbf{Q} \gets \{ st_{\text{s}} \}$
        \State $\text{Visited} \gets \emptyset$
        \While {$\mathbf{Q} \neq \emptyset$}
          \State $\mathbf{Q'} = \{ \argmin\limits_{st' \in \mathbf{Q}} \text{f}(st_{\text{s}},st',\gamma^{\text{v}}_{\text{t}}) \}$ (TECO)
          \State $st = \argmin\limits_{st \in \mathbf{Q'}} ({\text{d}_\text{a}}(st^{\text{v}}) - {\text{d}_\text{s}}(st^{\text{p}}))$ (PWSO and VWNO)
          \If {$st$ contains $\gamma^{\text{v}}_{\text{t}}$}
            \Return $\text{g}(st_{\text{s}},st)$ and $\text{Backtrack}(st)$ 
          \EndIf
          \For {$st' \in N(st)$}
            \If {$st' \notin \text{Visited}$}
              \State $\text{pred}(st') \gets st$
              \State $\text{MRL}(st',\gamma^{\text{v}}_{\text{t}}) \gets \text{Dijkstra}({st'}^{\text{v}}, \gamma^{\text{v}}_{\text{t}}, \text{l}(\cdot))$
              \State $\text{MRC}(st',\gamma^{\text{v}}_{\text{t}}) \gets \text{Dijkstra}({st'}^{\text{v}}, \gamma^{\text{v}}_{\text{t}}, \alpha(\text{l}(\cdot)))$
              \State $\text{h}(st',\gamma^{\text{v}}_{\text{t}}) \gets \text{MRL}(st',\gamma^{\text{v}}_{\text{t}}) + r \cdot \text{MRC}(st',\gamma^{\text{v}}_{\text{t}})$
              \scriptsize \State $\text{g}(st_{\text{s}},st') \gets \text{g}(st_{\text{s}},st) + (\text{l}(st,st') + r \cdot \text{MIL}(st, st'))$ \normalsize
              \State $\text{f}(st_{\text{s}},st',\gamma^{\text{v}}_{\text{t}}) \gets \text{g}(st_{\text{s}},st') + \text{h}(st',\gamma^{\text{v}}_{\text{t}})$
              \State Add $st'$ to $\mathbf{Q}$
            \EndIf
          \EndFor
        \EndWhile
        \Return Infeasible
    \end{algorithmic}
\end{algorithm}

Since two relaxation parameters $\tilde{r_{\alpha}}, \tilde{r_{\beta}}$ were obtained in DWSP, RPGP repeats IDWS twice with $r = \tilde{r_\alpha}$ and $r = \tilde{r_\beta}$, and return the shorter feasible RW path. From the previous result, at least one RW path would be feasible; in fact, since the $\tilde{r_\alpha}$ and $\tilde{r_\beta}$ are good estimations from DWS, most of the time RPGP returns a \textit{close-to-optimal} RW path $\tilde{p}$, and the subsequent pruning strategies in PPNP are guided by $\text{l}(\tilde{p})$. The detailed steps of IDWS is given in Algorithm \ref{alg:idws}.

\begin{figure}[t]
\centering
  \begin{subfigure}[b]{0.6\columnwidth}
    \includegraphics[width=\columnwidth]{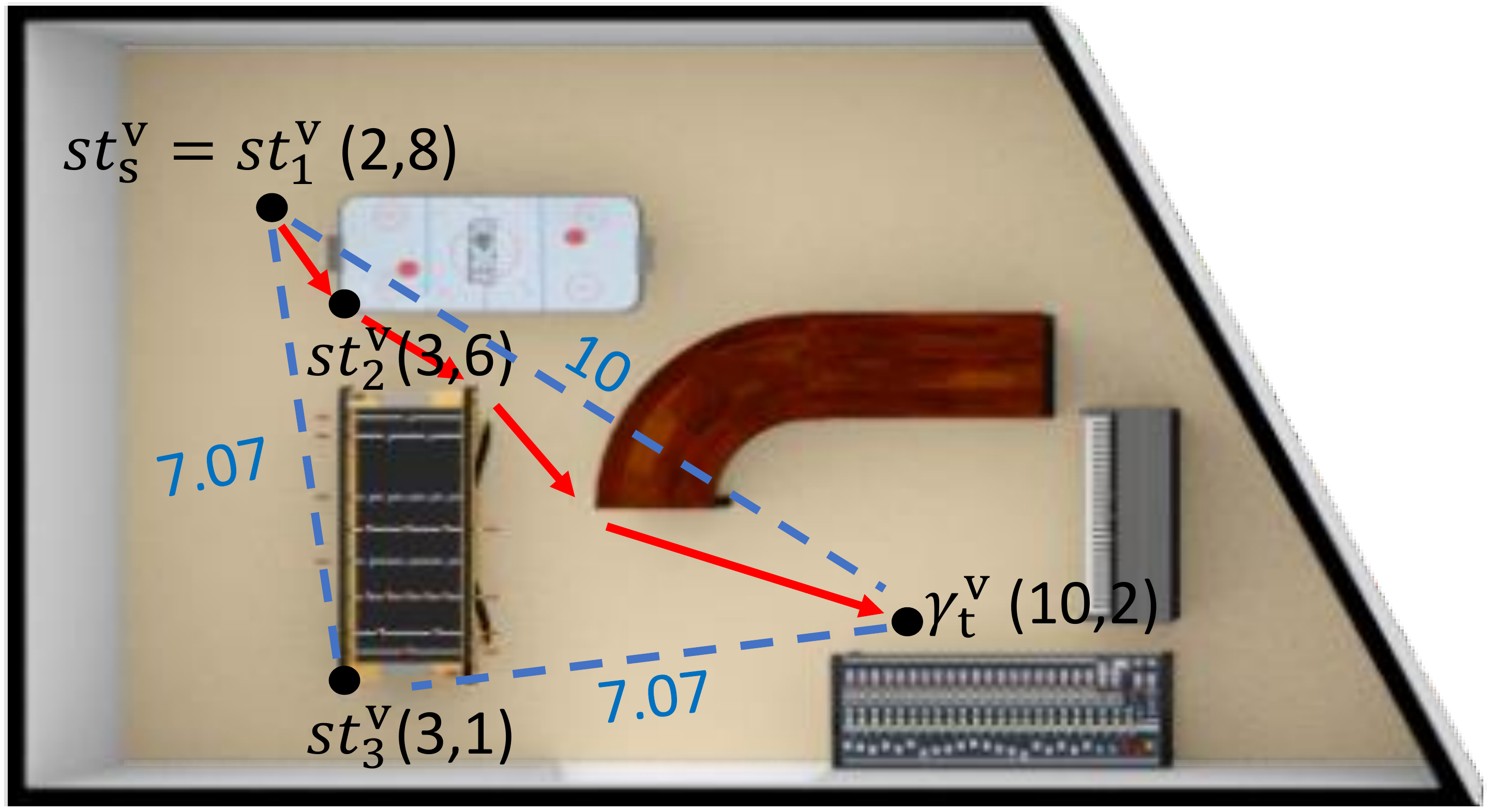}
    \caption{A virtual world.}
    \label{fig:idws-v}
  \end{subfigure}
  \hfill 
  \begin{subfigure}[b]{0.33\columnwidth}
    \includegraphics[width=\columnwidth]{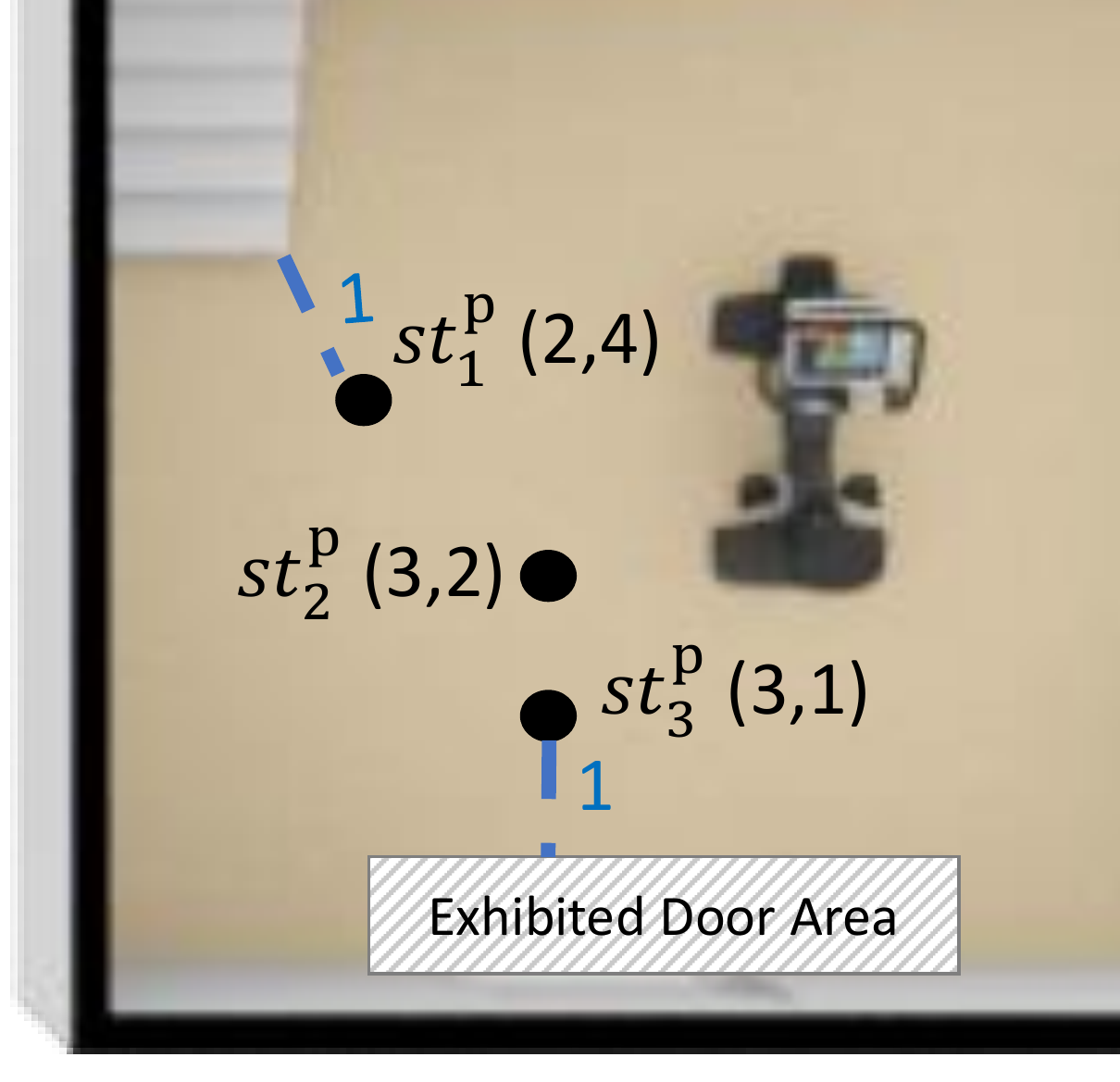}
    \caption{A physical world.}
    \label{fig:idws-p}
  \end{subfigure}
\caption{Three neighboring loco-states in IDWS.}
\label{fig:idws_example}
\end{figure}

\begin{table}[t]
\centering
\caption{An example of TECO in IDWS.}
\begin{tabular}{|c|c|c|c|c|c|}
\hline
    & \begin{tabular}[c]{@{}c@{}}$\text{g}(st_{\text{s}},st)$\\ (AEC)\end{tabular}  & MRL & MRC & \begin{tabular}[c]{@{}c@{}}$\text{h}(st,\gamma^{\text{v}}_{\text{t}})$\\ (REC)\end{tabular}  & \begin{tabular}[c]{@{}c@{}}$\text{f}(st_{\text{s}},st,\gamma^{\text{v}}_{\text{t}})$\\ (TECO)\end{tabular} \\ \hline\hline
$st_1$ & 4.5  & 10.7 & 5   & 31.7 & 36.2                                                   \\ \hline
$st_2$ & 2.2    & 8.5  & 4   & 25.3 & 27.5                                                   \\ \hline
$st_3$ & 20.7 & 7.1  & 2   & 15.5 & 36.2                                                   \\ \hline
\end{tabular}
\label{table:teco_example}
\end{table}

\begin{table}[t]
\centering
\caption{An example of PWSO and VWNO in IDWS.}
\begin{tabular}{|c|c|c|c|}
\hline
    & \begin{tabular}[c]{@{}c@{}}$\text{d}_{\text{s}}$\\ (PWSO)\end{tabular} & \begin{tabular}[c]{@{}c@{}}$\text{d}_{\text{a}}$\\ (VWNO)\end{tabular} & \begin{tabular}[c]{@{}c@{}}Total\\ ($\text{d}_{\text{a}} - \text{d}_{\text{s}}$)\end{tabular} \\ \hline\hline
$st_1$ & 1                                                    & 10                                                    & 9      \\ \hline
$st_3$ & 1                                                    & 14.14                                                    & 13.14     \\ \hline
\end{tabular}
\label{table:tie_breaking}
\end{table}

\begin{example}
\sloppy Recall the state after Example \ref{CSMS_example} where $\mathbf{Q}$ $= \{ st_{\text{s}} \}$ $= \{ (((2,8), 270^{\circ}),$ $((2,4), 270^{\circ}))\}$. IDWS first expands $st_{\text{s}}$ and adds all neighboring loco-states to $\mathbf{Q}$. Figure \ref{fig:idws_example} presents three neighboring loco-states: $st_1 = (((2,8), 270^{\circ}),$  $((2,4), 180^{\circ}))$, which is the result of a Reset operation right at the start loco-state; $st_2 = (((3,6), 315^{\circ}),$ $((3,2), 315^{\circ}))$, which represents a simple straight south-east step without RW operations; and $st_3 = (((3,1), 270^{\circ}),$  $((3,1), 270^{\circ}))$. From $st_{\text{s}}$ to $st_3$, the user walks a long step from $(2,8)$ to $(3,1)$ in the virtual world, while a set of acute RW operations are used to realize the physical transition from $(2,4)$ to $(3,1)$ so that the user does not bump into the boundary. 

Table \ref{table:teco_example} shows the heuristic values of $st_1, st_2$ and $st_3$, where $\text{g}(st_{\text{s}},st)$ is the aggregated LR cost from $st_{\text{s}}$ to $st$ in LR-DROP with $r=4.2$. For $st_2$, moving from $st_{\text{s}}$ to $st_2$ does not incur any RW cost, and the aggregated LR cost is 2.2 (the v-edge length). $\text{MRL}(st_2,\gamma^{\text{v}}_{\text{t}})$ is the shortest v-path length 8.5 from $(3,6)$ to the destination $(10,2)$, following the red \textit{v-path} in Figure \ref{fig:idws-v}. Note that the v-path, while containing $\gamma^{\text{v}}_{\text{s}} = \gamma^{\text{v}}_{\text{1}} = (2,8)$, is obtained on the v-graph instead of from the loco-states. Thus, it is not an RW path and does not passes through the unexplored $st_1$. $\text{MRC}(st_2,\gamma^{\text{v}}_{\text{t}})$ is the lowest estimated cost 4 from $(3,6)$ to $(10,2)$ in COS-DROP (also following the red path). 

Therefore, $\text{h}(st_2,\gamma^{\text{v}}_{\text{t}}) = 8.5 + 4 \cdot 4.2 = 25.3$, and $\text{f}(st_{\text{s}},st_2,\gamma^{\text{v}}_{\text{t}})$ for TECO is $2.2 + 25.3 = 27.5$. Since $st_2$ has the minimum heuristic value, IDWS explores $st_2$ earlier than $st_1$ and $st_3$. For $st_1$ and $st_3$ with the same heuristic value 36.2, Table \ref{table:tie_breaking} shows the PWSO and VWNO values of $st_1$ and $st_3$. For PWSO, their physical locations $(2,4)$ and $(3,1)$ are identically proximal to the nearest obstacle (shown in blue in Figure \ref{fig:idws-p}). In the virtual world, $st^{\text{v}}_1$ is right to the source with the combined straight-line distance 10 (also shown in blue in Figure \ref{fig:idws-v}). However, $st^{\text{v}}_3$ deviates a lot from the straight line and incurs a combined distance $7.07+7.07=14.14.$ Thus, VWNO favors $st_1$. \hfill \qedsymbol
\end{example}

\subsection{Pruning and Path Navigation Phase}
\label{subsec:ppnp}

The main idea behind PPNP is to leverage the reference RW path to remove redundant loco-states. It starts from the source state $st_{\text{s}}$ and iteratively updates the labels of loco-states according to MRL and MRC values. More specifically, for each loco-state $st$, let $\text{l}_\text{l}(st_{\text{s}},st)$ and $\text{c}_\text{l}(st_{\text{s}},st)$ respectively denote the length and RW cost for an RW path from $st_{\text{s}}$ to $st$ with the minimal \textit{length}. Similarly, let $\text{l}_\text{c}(st_{\text{s}},st)$ and $\text{c}_\text{c}(st_{\text{s}},st)$ denote the length and RW cost for an RW path from $st_{\text{s}}$ to $st$ with the minimal RW \textit{cost}. Finally, let $\text{pred}_\text{l}(st)$ and $\text{pred}_\text{c}(st)$ represent the predecessor loco-states of $st$ on the above paths. At the beginning, all four label values of $st_{\text{s}}$ itself is initiated to zero, and $\text{pred}_\text{l}(st_{\text{s}})$ = $\text{pred}_\text{c}(st_{\text{s}}) = st_{\text{s}}$. The label values and predecessors for all other loco-states are initiated upon first visiting (detailed later). 

Moreover, for any location pair $(\gamma^{\text{v}}_1,\gamma^{\text{v}}_2)$ in the v-graph, let $\text{c}^{\alpha}_{\min}(\gamma^{\text{v}}_1,\gamma^{\text{v}}_2)$ and $\text{c}^{\beta}_{\min}(\gamma^{\text{v}}_1,\gamma^{\text{v}}_2)$ denote the minimum RW costs for a v-path from $\gamma^{\text{v}}_1$ to $\gamma^{\text{v}}_2$ in COS-DROP and CPS-DROP, respectively.\footnote{$\text{c}^{\alpha}_{\min}(\gamma^{\text{v}}_1,\gamma^{\text{v}}_2)$ is acquired by finding the shortest-path on the v-graph with the edge weight between $(\gamma^{\text{v}}_1,\gamma^{\text{v}}_2)$ as $\alpha(\text{l}(\gamma^{\text{v}}_1,\gamma^{\text{v}}_2))$. Similarly, $\text{c}^{\beta}_{\min}(\gamma^{\text{v}}_1,\gamma^{\text{v}}_2)$ is found by replacing $\alpha(\text{l}(\gamma^{\text{v}}_1,\gamma^{\text{v}}_2))$ with $\beta(\text{l}(\gamma^{\text{v}}_1,\gamma^{\text{v}}_2))$. According to Theorem \ref{thm:bounded_cost}, there is a p-path incurring an RW cost between $\text{c}^{\alpha}_{\min}(\gamma^{\text{v}}_1,\gamma^{\text{v}}_2)$ and $\text{c}^{\beta}_{\min}(\gamma^{\text{v}}_1,\gamma^{\text{v}}_2)$.}
Meanwhile, let $\text{l}_{\min}(\gamma^{\text{v}}_1,\gamma^{\text{v}}_2)$ denote the lower bound of the length for a feasible v-path from $\gamma^{\text{v}}_1$ to $\gamma^{\text{v}}_2$, i.e., the lower bound of DROP. Note that exactly computing the tightest (largest) $\text{l}_{\min}$ is exactly a DROP query. However, here we only require $\text{l}_{\min}$ to be a lower bound. Therefore, $\text{l}_{\min}$ is given the shortest v-path length between $(\gamma^{\text{v}}_1,\gamma^{\text{v}}_2)$, which can be found with Dijkstra's algorithm again with the edge weight set between $(\gamma^{\text{v}}_1,\gamma^{\text{v}}_2)$ set to their original distnace $\text{l}(\gamma^{\text{v}}_1,\gamma^{\text{v}}_2)$. According to CSMS, if $\text{c}^{\beta}_{\min}(\gamma^{\text{v}}_1,\gamma^{\text{v}}_2) \leq C$, the shortest v-path between $(\gamma^{\text{v}}_1,\gamma^{\text{v}}_2)$ is feasible, and $\text{l}_{\min}$ is tight. Also, let $\gamma^{\text{v}}_{\text{s}}$ and $\gamma^{\text{v}}_{st}$ respectively represent the virtual locations of the source loco-state $st_{\text{s}}$ and the current loco-state $st$. Let $\tilde{L}$ denote the length of the current best feasible reference path.

The search process of PPNP resembles that in IDWS; PPNP here also maintains a priority queue $\mathbf{Q}$, and the search process also contains iterative rounds of loco-state examination. However, PPNP is subtly different from IDWS in RPGP. IDWS explores a tiny fraction of the loco-state space and finds a reference RW path, while PPNP investigates the loco-state space comprehensively and trims off redundant loco-states. While the search order in IDWS follows TECO, PWSO and VWNO, PPNP does not employ them. Instead, the order in PPNP is controlled by pruning strategies, and loco-states may re-enter the priority queue in PPNP. Throughout the search process, 
To skip redundant loco-states, PPNP explores the following pruning strategies. 1) \textit{Infeasible Loco-State Pruning} (ILSP). If $\text{c}^{\alpha}_{\min}(\gamma^{\text{v}}_{\text{s}}, \gamma^{\text{v}}_{st}) + \text{c}^{\alpha}_{\min}(\gamma^{\text{v}}_{st}, \gamma^{\text{v}}_{\text{t}}) > C$, every RW path from $st_{\text{s}}$ to the destination via $st$ is infeasible. Thus, $st$ is removed. 2) \textit{Suboptimal Loco-State Pruning} (SLSP). If $\text{l}_{\min}(\gamma^{\text{v}}_{\text{s}}, \gamma^{\text{v}}_{st}) + \text{l}_{\min}(\gamma^{\text{v}}_{st}, \gamma^{\text{v}}_{\text{t}}) > \tilde{L}$, any RW path from $st_{\text{s}}$ to the destination via $st$ is longer than the reference RW path. $st$ is thereby removed. 3) \textit{Unpromising Loco-State Locking} (ULSL). If $\text{c}_\text{c}(st_{\text{s}},st) + \text{c}^{\alpha}_{\min}(\gamma^{\text{v}}_{st}, \gamma^{\text{v}}_{\text{t}}) > C$, currently it is not likely to find any feasible RW path via $st$. Therefore, PPNP pauses the search expanded from $st$. 
Note that $st$ cannot be removed yet because when the RW path from $st_{\text{s}}$ to $st$ improves later, $\text{c}_\text{c}(st_{\text{s}},st)$ decreases. Hence, $st$ may be expanded accordingly. However, if PPNP now expands $st$, the subsequently visited loco-states always satisfy ULSL and create no feasible solution. Thus, when a loco-state satisfies ULSL, its examination is \textit{postponed} (i.e., not revisited) until all other loco-states in $\mathbf{Q}$ are examined. Similarly, a loco-state $st$ is shelved when $\text{l}_\text{l}(st_{\text{s}},st) + \text{l}_{\min}(\gamma^{\text{v}}_{st}, \gamma^{\text{v}}_{\text{t}}) > \tilde{L}$. If all remaining loco-states are postponed, the reference path cannot be improved, and those states are removed accordingly.

If the current loco-state $st$ passes the above pruning criteria, for each unvisited $st' \in \text{N}(st)$, PPNP assigns $\text{l}_\text{l}(st_{\text{s}},st') = \text{l}_\text{l}(st_{\text{s}},st) + \text{l}(st, st')$, $\text{c}_\text{l}(st_{\text{s}},st') = \text{c}_\text{l}(st_{\text{s}},st) + \text{MIL}(st, st')$, $\text{l}_\text{c}(st_{\text{s}},st') = \text{l}_\text{c}(st_{\text{s}},st) + \text{l}(st, st')$, $\text{c}_\text{c}(st_{\text{s}},st') = \text{c}_\text{c}(st_{\text{s}},st) + \text{MIL}(st, st')$, and also $\text{pred}_\text{l}(st') = \text{pred}_\text{c}(st') = st$. On the other hand, for each visited $st'$, PPNP updates $\text{l}_\text{l}(st_{\text{s}},st'), \text{c}_\text{l}(st_{\text{s}},st')$, and $\text{pred}_\text{l}(st')$ when the new $\text{l}_\text{l}(st_{\text{s}},st')$ is lower. $\text{l}_\text{c}(st_{\text{s}},st'), \text{c}_\text{c}(st_{\text{s}},st')$ and $\text{pred}_\text{c}(st')$ are also updated when $\text{c}_\text{c}(st_{\text{s}},st')$ is better. If the values are updated for a previously locked $st'$, PPNP \textit{unlocks} $st'$ by increasing the priority value from $-\infty$ to $0$.
An improved reference path with length $\text{l}'_\text{l}(st_{\text{s}},st')$ appears when $\text{l}'_\text{l}(st_{\text{s}},st') = \text{l}_\text{l}(st_{\text{s}},st') + \text{l}_{\min}(\gamma^{\text{v}}_{st'}, \gamma^{\text{v}}_{\text{t}}) < \tilde{L}$ and $\text{c}_\text{l}(st_{\text{s}},st') + \text{c}^{\beta}_{\min}(\gamma^{\text{v}}_{st'}, \gamma^{\text{v}}_{\text{t}}) \leq C$. Similarly, an improved reference path with length $\text{l}'_\text{c}(st_{\text{s}},st')$ appears when $\text{l}'_\text{c}(st_{\text{s}},st') = \text{l}_\text{c}(st_{\text{s}},st) + \text{l}_{\min}(\gamma^{\text{v}}_{st}, \gamma^{\text{v}}_{\text{t}}) < \tilde{L}$ and $\text{c}_\text{c}(st_{\text{s}},st) + \text{c}^{\beta}_{\min}(\gamma^{\text{v}}_{st}, \gamma^{\text{v}}_{\text{t}}) \leq C$. The above two cases correspond to the minimum-length and minimum-RW-cost paths to $st$, respectively. 

When all remaining loco-states in $\mathbf{Q}$ are \textit{locked}, i.e., postponed by ULSL before, PPNP computes the exact values of $\text{l}_\text{l}(st_{\text{s}},st)$, $\text{c}_\text{l}(st_{\text{s}},st)$, $\text{l}_\text{c}(st_{\text{s}},st)$, and $\text{c}_\text{c}(st_{\text{s}},st)$ by Dijkstra's algorithm on $ST$ with edge weight assigned to $\text{dist}(st_1, st_2)$ for $\text{l}_\text{l}(st_{\text{s}},st)$ and $\text{c}_\text{l}(st_{\text{s}},st)$, and $\text{MIL}(st_1, st_2)$ for $\text{l}_\text{c}(st_{\text{s}},st)$, and $\text{c}_\text{c}(st_{\text{s}},st)$.\footnote{Note that the exact values are not computed for all loco-states to reduce the computational cost of PPNP. Also, for $\text{l}_\text{l}(st_{\text{s}},st)$ and $\text{c}_\text{l}(st_{\text{s}},st)$, it suffices to apply Dijkstra's algorithm on the v-graph instead of $ST$ since the path length only depends on the v-edge lengths.} PPNP then re-checks the ULSL criteria. If all loco-states still satisfy at least one of the criteria in ULSL, the search process is terminated.

\begin{example}
Figure \ref{fig:pruning_example} illustrates the pruning strategies in Example \ref{DROPExample} in the same DROP query from $st_{\text{s}} = ((2,8), 270^{\circ}, (2,4), 270^{\circ})$ to $\gamma^{\text{v}}_{\text{t}} = (10,2)$ but with the RW with $C=5.5$. The reference RW path has the v-path (in red) length as 14.3 and RW cost as 5. For virtual location $(6,6)$, $\text{c}^{\alpha}_{\min}((2,8), (6,6)) + \text{c}^{\alpha}_{\min}((6,6), (10,2)) = 2 + 4 = 6 > 5.5$. Therefore, any loco-state at $(6,6)$ is pruned by ILSP. For virtual location $(12,6)$, since $\text{l}_{\min}((2,8), (12,6)) + \text{l}_{\min}((12,6), (10,2)) = 10.3 + 4.6 = 14.9 > 14.3$, any loco-state at $(12,6)$ is removed by SLSP. After PPNP expands $st_1 = ((5,1), 0^{\circ}, (5,2), 0^{\circ})$, for $st_1$, $\text{l}_\text{l}(st_{\text{s}},st_1) = 9.1, \text{c}_\text{l}(st_{\text{s}},st_1) = 4, \text{l}_\text{c}(st_{\text{s}},st_1) = 9.2, \text{c}_\text{c}(st_{\text{s}},st_1) = 4$. PPNP then visits the neighboring loco-states $st_2 = ((6,3), 53^{\circ}, (6,4), 53^{\circ})$. The transition from $st_1$ to $st_2$ has a step length of 2.2 with the RW cost as 1. When $st_2$ is examined, PPNP updates $\text{l}_\text{l}(st_{\text{s}},st_2) = 9.1+2.2=11.3, \text{c}_\text{l}(st_{\text{s}},st_2) = 4+1=5, \text{l}_\text{c}(st_{\text{s}},st_2) = 9.2+2.2=11.4, \text{c}_\text{c}(st_{\text{s}},st_2) = 4+1=5.$ Since $\text{c}_\text{c}(st_{\text{s}},st_2) + \text{c}^{\alpha}_{\min}((6,3), (10,2)) = 5+2 = 7 > 5.5$, $st_2$ is postponed by ULSL. However, after the optimal RW path with the v-path in blue is explored, $\text{c}_\text{c}(st_{\text{s}},st_2)$ is lowered to 3, and PPNP then finds a better RW path corresponding to the blue v-path with length 10.7. \hfill \qedsymbol
\end{example}

\begin{algorithm}[t]  
    \caption{Search Process in PPNP}
    \label{alg:ppnp-search}
    \begin{algorithmic}[1]
        \scriptsize 
        \Require DROP instance, reference RW path $\tilde{p}$
        \Ensure Trimmed loco-state space $X$
        \State $\mathbf{Q} \gets \{ st_{\text{s}} \}$
        \State $\text{Visited} \gets \emptyset$
        \State $\tilde{L} \gets \text{l}(\tilde{p})$
        \State Compute or query $\text{c}^\alpha_{\min}(\cdot)$, $\text{c}^\beta_{\min}(\cdot)$, $\text{l}_{\min}(\cdot)$
        \While {$\mathbf{Q} \neq \emptyset$}
          \If {$\text{top}(\mathbf{Q})$ is locked}
            \For {$st \in \mathbf{Q}$}
              \State Compute exact labels and check ULSL
              \If{Some labels of $st$ are updated}
                \State Add $st$ to $\mathbf{Q}$ with priority 0 and \textbf{break}
              \EndIf
              \Return $X = \text{Visited}$
            \EndFor
          \Else
              \State $st \gets \text{top}(\mathbf{Q})$
              \If {$\text{c}^{\alpha}_{\min}(\gamma^{\text{v}}_{\text{s}}, \gamma^{\text{v}}_{st}) + \text{c}^{\alpha}_{\min}(\gamma^{\text{v}}_{st}, \gamma^{\text{v}}_{\text{t}}) > C$ (ILSP)}
                \State Discard $st$ and \textbf{break}
              \EndIf
              \If {$\text{l}_{\min}(\gamma^{\text{v}}_{\text{s}}, \gamma^{\text{v}}_{st}) + \text{l}_{\min}(\gamma^{\text{v}}_{st}, \gamma^{\text{v}}_{\text{t}}) > \tilde{L}$ (SLSP)}
                \State Discard $st$ and \textbf{break}
              \EndIf
                \If {$\text{c}_\text{c}(st_{\text{s}},st) + \text{c}^{\alpha}_{\min}(\gamma^{\text{v}}_{st}, \gamma^{\text{v}}_{\text{t}}) > C$ or $\text{l}_\text{l}(st_{\text{s}},st) + \text{l}_{\min}(\gamma^{\text{v}}_{st}, \gamma^{\text{v}}_{\text{t}}) > \tilde{L}$ (ULSL)}
                \State Add $st$ back to $\mathbf{Q}$ with priority $-\infty$ and \textbf{break}
              \EndIf
              \For {$st' \in N(st)$}
                \State Compute $\text{l}_\text{l}(st_{\text{s}},st'), \text{c}_\text{l}(st_{\text{s}},st'), \text{l}_\text{c}(st_{\text{s}},st'), \text{c}_\text{c}(st_{\text{s}},st')$
                \State Update $\tilde{L}$
                \If {$st \notin \mathbf{Q}$ or $\text{Visited}$}
                  \State Add $st$ to $\mathbf{Q}$ with priority 0
                  \State $\text{pred}_\text{l}(st'), \text{pred}_\text{c}(st') \gets st$
                \Else
                  \State Update the labels and predecessor tag
                \EndIf
              \EndFor
           \EndIf
        \EndWhile
        \Return $X = \text{Visited}$
        \normalsize
    \end{algorithmic}
\end{algorithm}

\begin{algorithm}[t]  
    \caption{Round-and-Scaled DEWN DP}
    \label{alg:ppnp-dp}
    \begin{algorithmic}[1]
        \Require 
        $X$: the remaining loco-state space, $ \epsilon$: the desired approximation ratio, $\underline{L}, \overline{L}$: a lower bound and an upper bound of the optimal objective
        \Ensure 
        $\tilde{p}$: a $(1+\epsilon)$-approximation solution
        \State $S \gets \frac{\epsilon \cdot \underline{L}}{|X|}$
        \For{$(st_1, st_2) \in X \times X$}
          \State $\text{l}'(st_1, st_2) \gets S \cdot \lceil \frac{\text{l}(\gamma^{\text{v}}_1, \gamma^{\text{v}}_2)}{S} \rceil$
        \EndFor
        \For{$\frac{l}{S} = 1$ to $\lceil \frac{\overline{L}}{S} \rceil+ |X|$}
          \For{$st \in X$}
            \State $\text{c}(st,l) \gets \infty$
            \For{$st' \in N(st)$}
              \scriptsize \If{$\text{c}(st',l-\text{l}(st,st'))+\text{MIL}(st,st') < \text{c}(st,l)$} \normalsize
                 \State $\text{c}(st,l) \gets \text{c}(st',l-\text{l}(st,st'))+\text{MIL}(st,st')$
                 \State $\text{Predecessor}(st) \gets st'$
              \EndIf
            \EndFor
            \If {$\gamma^{\text{v}} = \gamma^{\text{v}}_{\text{t}}$ and $\text{c}(st,l) \leq C$}
              \State $\tilde{p} \gets \text{Backtrack}(st)$
              \Return $\tilde{p}$
            \EndIf
          \EndFor
        \EndFor
        \Return ``Infeasible''
    \end{algorithmic}
\end{algorithm}

\subsection{Approximate Solution}
\label{subsec:theo}
With substantial loco-states removed, DP states are then generated from the remaining loco-states, where each DP state is a combination of a loco-state and a \textit{rounded} v-path length (detailed later). Note that the number of DP states is much smaller compared with Basic DP because 1) the pruning process effectively trims off the loco-states, and 2) the number of possible v-path lengths is reduced by a rounding strategy that discretizes the length of v-edges with a scale parameter $S$. The length $\text{l}(e)$ of each edge $e$ is rounded to $S \cdot \lceil \frac{\text{l}(e)}{S} \rceil$. The degradation of solution quality is limited because the rounding error in each edge is at most $S$ (correlated to the approximation ratio). Let ${\text{DROP}_\text{X}}$ denote the \textit{post-rounding} problem instance of DROP with (i) the edge length in v-graph rounded by $S$ and (ii) redundant loco-states removed, and let $X$ be the loco-state space in ${\text{DROP}_\text{X}}$. We set $S$ to $\frac{\epsilon \cdot \underline{L}}{|X|}$, where $\underline{L}$ is a lower bound of $\text{l}(p^\ast)$, the length of the optimal RW path.\footnote{One approach here is to leverage the length of the v-path obtained by CSMS in CPS-LR-DROP; one can prove that it is indeed a lower bound of $\text{l}(p^\ast)$ via Property \ref{property:lagrange}. We revisit this issue in the time complexity discussions.} The detailed procedure is shown in Algorithm \ref{alg:ppnp-dp}, and Algorithm \ref{alg:dewn} presents the framework of DEWN. 

\begin{algorithm}[t]  
    \caption{Dual Entangled World Navigation (DEWN)}
    \label{alg:dewn}
    \begin{algorithmic}[1]
        \Require 
        $ST, st_{\text{s}}, \gamma^{\text{v}}_{\text{t}}, \text{MIL}(\cdot), C$
        \Ensure 
        RW path $p$
        \State Construct or query $\alpha(\cdot)$ and $\beta(\cdot)$ (MIL Range)
        \State $(r^\ast_{\alpha}, r^\ast_{\beta}) \gets \text{CSMS}(G^{\text{v}}, \gamma^{\text{v}}_{\text{s}}, \gamma^{\text{v}}_{\text{t}}, C, \alpha(\cdot), \beta(\cdot))$
        \If {Infeasible}
          \Return Infeasible
        \EndIf
        \If{Optimal (the shortest v-path $p^{\text{v}}$ is feasible)}
          \State $p^{\text{p}} \gets \text{IMCA}(p^{\text{v}})$
          \Return $p$
        \EndIf
        \State $\tilde{p}_\alpha \gets \text{IDWS}(ST, st_{\text{s}}, \gamma^{\text{v}}_{\text{t}}, r^\ast_{\alpha})$
        \State $\tilde{p}_\beta \gets \text{IDWS}(ST, st_{\text{s}}, \gamma^{\text{v}}_{\text{t}}, r^\ast_{\beta})$
        \State $\tilde{p} \gets$ the shorter feasible RW path between $\tilde{p}_\alpha$ and $\tilde{p}_\beta$
        \State $X \gets \text{Search and Pruning in PPNP}(ST, C, \tilde{p})$
        \State $p \gets \text{Rounded-and-Scaled DEWN DP}(X, \epsilon, \underline{L}, \overline{L})$
        \Return $p$
    \end{algorithmic}
\end{algorithm}

\begin{lemma} \label{pruning_optimality}
There exists an RW path $p^\ast$ such that $p^\ast$ is \textit{optimal} in DROP and \textit{feasible} in ${\text{DROP}_\text{X}}$.
\end{lemma}
\begin{proof}

If the lemma does not hold, \textit{every} optimal path $p^\ast$ in DROP is infeasible in ${\text{DROP}_\text{X}}$. Since the MIL values are identical in DROP and ${\text{DROP}_\text{X}}$, the only possibility that some feasible path $p$ in DROP is infeasible in ${\text{DROP}_\text{X}}$ is that $p$ consists of non-existing loco-states, i.e., loco-states not in $X$. For every $p^\ast$, the above implies there exists at least one loco-state in $p^\ast$ that is not in $X$. Given an arbitrary $p^\ast$, let $st$ represent \textit{the first loco-state} $st \in p^\ast$ such that $st \notin X$, i.e., $st$ precedes all other loco-states in $p^\ast$ that are not in $X$. Let $\text{pred}_{st}$ represent the predecessor of $st$ in $p^\ast$. Thus, we have $\text{pred}_{st} \in X$. Furthermore, let $\gamma^{\text{v}}_{\text{pred}_{st}} $ denote the virtual location of $\text{pred}_{st}$.

Since $st \notin X$, $st$ is either 1) removed by ILSP, 2) removed by SLSP, or 3) never visited in PPNP. Suppose $st$ is removed by ILSP. It follows from the definition of ILSP that $c(p^\ast) \geq \text{c}^{\alpha}_{\min}(\gamma^{\text{v}}_{\text{s}}, \gamma^{\text{v}}_{st}) + \text{c}^{\alpha}_{\min}(\gamma^{\text{v}}_{st}, \gamma^{\text{v}}_{\text{t}}) > C$, which contradicts with the fact that $p^\ast$ is feasible in DROP. Next, if $st$ is removed by SLSP, it implies $\text{l}_{\min}(\gamma^{\text{v}}_{\text{s}}, \gamma^{\text{v}}_{st}) + \text{l}_{\min}(\gamma^{\text{v}}_{st}, \gamma^{\text{v}}_{\text{t}}) > \tilde{L}$ for some reference path length $\tilde{L}$ during PPNP. However, $\text{l}_{\min}(\gamma^{\text{v}}_{\text{s}}, \gamma^{\text{v}}_{st})$ is a lower bound of any v-path length between $\gamma^{\text{v}}_{\text{s}}$ and $\gamma^{\text{v}}_{st}$, and $\text{l}_{\min}(\gamma^{\text{v}}_{st}, \gamma^{\text{v}}_{\text{t}})$ is a lower bound of any v-path length between $\gamma^{\text{v}}_{st}$ and $\gamma^{\text{v}}_{\text{t}}$. Thus, we have $l(p^\ast) > \text{l}_{\min}(\gamma^{\text{v}}_{\text{s}}, \gamma^{\text{v}}_{st}) + \text{l}_{\min}(\gamma^{\text{v}}_{st}, \gamma^{\text{v}}_{\text{t}}) > \tilde{L} > l(p^\ast)$, which is a contradiction, where the last inequality comes from the optimality of $p^\ast$ in DROP.

Finally, suppose $st$ is never visited in PPNP. For this to happen, $\text{pred}_{st}$ satisfies ULSL, and also remains satisfying ULSL in the final round of checking in PPNP, since otherwise PPNP either discards $\text{pred}_{st}$ later (not consistent with the definition of $st$) or visits the neighborhood of $\text{pred}_{st}$ (which will visit $st$). In the final checking, PPNP either finds that $\text{c}_\text{c}(st_{\text{s}},\text{pred}_{st}) + \text{c}^{\alpha}_{\min}(\gamma^{\text{v}}_{\text{pred}_{st}} , \gamma^{\text{v}}_{\text{t}}) > C$ or $\text{l}_\text{l}(st_{\text{s}},\text{pred}_{st}) + \text{l}_{\min}(\gamma^{\text{v}}_{\text{pred}_{st}}, \gamma^{\text{v}}_{\text{t}}) > \tilde{L}$. For the former case, since $\text{c}_\text{c}(st_{\text{s}},\text{pred}_{st})$ was given the exact minimum RW cost from $st_{\text{s}}$ to $\text{pred}_{st}$, the total RW cost incurred by the subpath from $st_{\text{s}}$ to $\text{pred}_{st}$ in $p^\ast$ is at least $\text{c}_\text{c}(st_{\text{s}},\text{pred}_{st})$, which implies that $c(p^\ast) \geq \text{c}_\text{c}(st_{\text{s}},\text{pred}_{st}) + \text{c}^{\alpha}_{\min}(\gamma^{\text{v}}_{\text{pred}_{st}} , \gamma^{\text{v}}_{\text{t}}) > C$, which contradicts with the fact that $p^\ast$ is feasible. For the latter case, since $\text{l}_\text{l}(st_{\text{s}},\text{pred}_{st})$ was given the exact shortest RW path length from $st_{\text{s}}$ to $\text{pred}_{st}$, the path length of $p^\ast$ is $l(p^\ast) \geq \text{l}_\text{l}(st_{\text{s}},\text{pred}_{st}) + \text{l}_{\min}(\gamma^{\text{v}}_{\text{pred}_{st}}, \gamma^{\text{v}}_{\text{t}}) > \tilde{L} \geq l(p^\ast)$, which is a contradiction, where the last inequality again comes from the optimality of $p^\ast$ in DROP. Since all cases lead to contradictions, it implies $st$ does not exist, which completes the proof.

\end{proof}

\begin{theorem} \label{thm:fptas}
DEWN returns a $(1+\epsilon)$-approximation for DROP.
\end{theorem}

\begin{proof}
Let the solution returned by DEWN be $p$. By Lemma \ref{pruning_optimality}, let $p^\ast$ denote an optimal solution in DROP which is a feasible solution (not necessarily optimal) for ${\text{DROP}_\text{X}}$. The MIL values are identical in DROP and ${\text{DROP}_\text{X}}$. Let $\text{l}(p^\ast)$ and $\text{l}(p)$ denote the v-path lengths of $p^\ast$ and $p$ in DROP, respectively. Similarly, $\text{l}_\text{X}(p^\ast)$ and $\text{l}_\text{X}(p)$ are their v-path lengths in ${\text{DROP}_\text{X}}$. Due to the possible rounding error, we have $\text{l}_\text{X}(p) \leq \text{l}_\text{X}(p^\ast)$, where $p^\ast$ may not be the optimal solution in ${\text{DROP}_\text{X}}$.

Since $p^\ast$ is optimal and it does not pass through more than $|X|$ loco-states, $\text{l}(p^\ast)$ is the sum of at most $|X|$ v-edge lengths, and the rounding error of each edge is at most $S$. The total rounding error along $p^\ast$ thereby does not surpass $S \cdot |X|$. Thus, we have $\text{l}_\text{X}(p) \leq \text{l}_\text{X}(p^\ast) \leq \text{l}(p^\ast) + S \cdot |X| = \text{l}(p^\ast) + \epsilon \cdot \underline{L} \leq (1+\epsilon)\text{l}(p^\ast)$. The final inequality holds because $\underline{L}$ is a lower bound of $\text{l}(p^\ast)$.
\end{proof}

\para{Time Complexity.}
First, the MIL range can be obtained offline in $O(N^2)$-time since every pair of loco-states needs to be examined once. First, CSMS involves at most $O(|E^{\text{v}}| \cdot \log^3 |E^{\text{v}}|)$ iterations \cite{AJ01}. On the other hand, if CSMS chooses to binary partition the remaining possible region for $r$, and early termination is applied with parameter $\delta$, then the number of iterations is $O(\log(\frac{1}{\delta})).$ Specifically, each iteration of CSMS needs $O(|E^{\text{v}}| + |\Gamma^{\text{v}}| \cdot \log|\Gamma^{\text{v}}|)$ time. The total complexity of CSMS is $O(|E^{\text{v}}|^2 \cdot \log^3|E^{\text{v}}| + |\Gamma^{\text{v}}|^3 \cdot \log^4|\Gamma^{\text{v}}|)$. Moreover, for RPGP, all MRL and MRC values can be derived in $O(|\Gamma^{\text{v}}|)$ iterations, whereas each iteration invokes one Dijkstra's algorithm in the v-graph. The total complexity to find MRL and MRC is therefore $O(|\Gamma^{\text{v}}|) \cdot O(|E^{\text{v}}| + |\Gamma^{\text{v}}| \cdot \log{|\Gamma^{\text{v}}|}) =
O(|\Gamma^{\text{v}}|^3 + |\Gamma^{\text{v}}|^2 \cdot \log{|\Gamma^{\text{v}}|})$, and the complexity of IDWS is $O(N^2 + N \cdot \log{N}) = O(N^2)$. Note that the first two phases collaborate to generate a promising reference RW path in $O(N^2)$ time, while the actual computing effort is effectively reduced by TECO. In contrast, directly finding the optimal $r$ in LR-DROP needs $O(N^2 \cdot \log^4{N})$.

For PPNP, three times of Dijkstra's algorithm on the v-graph is involved to find $c^\alpha_{\min}(\gamma^{\text{v}}_1, \gamma^{\text{v}}_2)$, $c^\beta_{\min}(\gamma^{\text{v}}_1, \gamma^{\text{v}}_2)$, and $\text{l}_{\min}(\gamma^{\text{v}}_1, \gamma^{\text{v}}_2)$ 
for each pair $(\gamma^{\text{v}}_1, \gamma^{\text{v}}_2)$, and the time complexity is $O(|\Gamma^{\text{v}}|^2) \cdot O(|E^{\text{v}}| + |\Gamma^{\text{v}}| \cdot \log{|\Gamma^{\text{v}}|}) = O(|\Gamma^{\text{v}}|^4 + |\Gamma^{\text{v}}|^3 \cdot \log{|\Gamma^{\text{v}}|})$. The complexity of the main search process in PPNP is $O(N^2)$ time because each loco-state has at most $N$ neighboring loco-states, and PPNP updates all labels in $O(1)$ time. After the pruning, there are $|X|$ remaining loco-states. 
Let $\overline{L}$ be the length of the reference RW path. For each loco-state, PPNP creates up to $\lceil \frac{\overline{L}}{S} \rceil + |X| = O(\frac{\overline{L}}{\underline{L}} \cdot \frac{|X|}{\epsilon})$ DP states. Each DP state is examined at most once in $O(|X|)$-time to determine the best predecessor DP state. Thus the total time complexity of the last step is $O(|X|^2 \cdot (\frac{\overline{L}}{\underline{L}} \cdot \frac{|X|}{\epsilon})) = O(\frac{|X|^3}{\epsilon} \cdot \frac{\overline{L}}{\underline{L}})$ time to traverse those states to find the approximate solution.\footnote{Here we point out that the ratio $\frac{\overline{L}}{\underline{L}}$ is not guaranteed to be in $O(1)$. To lower this ratio to $O(1)$, from a theoretical view, it may be necessary to invoke polynomial-time parameter testing techniques \cite{FE02}; practically, $\overline{L}$ is usually already very close to $\underline{L}$, as most of the time the ratio is less than 2 in our experiments.}

To sum up, the time complexity of DEWN is $O(N^2 + \frac{N^3}{\epsilon})$ since $|\Gamma^{\text{v}}|$, $|E^{\text{v}}|$, and $|X|$ are all smaller than $N$. Therefore, DEWN is an FPTAS of DROP. In contrast, the time complexity of Basic DP is $O(N^2 \cdot 2^{|E^{\text{v}}|})$. The running time of DEWN is significantly lower than Basic DP due to the following reasons: 1) DEWN processes mainly on v-graphs (of size $|\Gamma^{\text{v}}|$ and $|E^{\text{v}}|$), whereas Basic DP examines the whole DP space. 2) ILSP and SLSP effectively reduce loco-state space to $X$; 3) The rounding strategy scales the total possible v-path lengths from $O(2^{|E^{\text{v}}|})$ to $O(\frac{|X|}{\epsilon})$.

\section{Enhancements and Extensions} \label{sec:extensions}
In this section, we propose an enhancement for DROP, then show that DEWN can serve as a building block to support other spatial queries. 

\subsection{Critical Orientation Simplification}
While DEWN is efficient and effective in solving DROP, the practical efficiency is still a concern in some applications. In Simp-DEWN (see also Section \ref{sec:exp}) as a simplified implementation, we introduce the COS strategy to efficiently simplify the loco-states by ignoring the orientations for the computationally intensive parts of DEWN. The idea of COS is as follows. On solving DROP, DEWN always examine the user orientations in the dual worlds carefully to leverage RO to steer users away from obstacles, incurring the side-effect of a large time complexity, since a single pair of locations $(\gamma^{\text{v}}, \gamma^{\text{p}})$ corresponds to $|\Theta|^2$ valid loco-states. The final time complexity is then dependent on $O(\Theta^4)$. However, through careful investigation of the RW paths, we found that the complexity can be significantly reduced by allowing an additional RW operation at each loco-state along the RW path. More specifically, if an additional RO (or Reset, but more costly) is allowed at each loco-state along the RW path, then the orientations of the user need not to be considered in DEWN, which trades solution quality for computational efficiency.

Concretely, this is achieved by setting $\Theta = \{ 0^{\circ} \}$, or equivalently speaking, merging all the loco-states with identically locations $(\gamma^{\text{v}}, \gamma^{\text{p}})$. The additional RW cost incurred can be bounded since this incurs at most the largest cost of an RO or Reset at each loco-state. Denote $C_\theta$ as the maximum RW cost incurred, depending on the cost model, to correct the user orientation. The following theorem shows that this simplification gives an error-bounded FPTAS.

\begin{theorem} \label{thm:simplification}
DEWN with COS finds an RW path $\tilde{p}$ that (i) $\text{l}(\tilde{p}) \leq (1+\epsilon) \cdot \text{l}(p^\ast)$, and (ii) $\text{c}(\tilde{p}) \leq C + C_\theta \cdot \text{D}(G^{\text{v}})$, where $\text{D}(G^{\text{v}})$ is the diameter of the v-graph.
\end{theorem}

\begin{proof}
The first part of the theorem is trivial since the simplification does not affect the path length. Let $p_{\text{sim}}$ be the returned path by DEWN with the simplification, and $p^{\text{v}}_{\text{sim}}$ be its v-path. Note that $p^{\text{v}}_{\text{sim}}$ may not necessarily be a simple path, since the optimal RW path may traverse the same location in the virtual world multiple times while realizing at different physical locations. Nevertheless, consider the following steps: 1) at each v-state (a turning point on the v-path), apply an additional Reset operation to find the optimal face orientation for the next edge (so that the real RW cost for the next step is the lowest). 2) if the v-path goes through the same v-state multiple times, remove all the intermediate part of the RW path between the first and last times, and apply an additional Reset operation to connect the loco-states (to relocate the user in the physical world). 3) merge the Reset operations if there exists multiple ones at a turning point due to the above.
Denote $p_{\text{fin}}$ to be the resulted RW path. From the definitions, it is clear that the above modifications invoke at most $\text{D}(G^{\text{v}})$ seperated Reset operations that cannot be merged, each incurring an RW cost at most $C_\theta$. Therefore, it must hold that $\text{c}(\tilde{p}) \leq C + C_\theta \cdot \text{D}(G^{\text{v}})$, which is the second part of the theorem.
\end{proof}

\subsection{Extension to Spatial Queries in Dual Worlds}
\label{sec:spatial}

Here we show that existing spatial query algorithms and index structures can exploit DEWN as a building block to support \textit{Dual-world $k$-Nearest Neighbors} (D$k$NN) and \textit{Dual-world Range} (DR) queries. Similar to DROP, D$k$NN and DR incorporate the RW cost constraint for dual-world VR applications (e.g., virtual touring and navigation). Given a start loco-state $st_{\text{s}}$, D$k$NN finds the POIs in the virtual world so that the v-paths from $st_{\text{s}}$ to them are the top-$k$ shortest ones, and and they comply with the RW cost constraint. Similarly, DR query returns all virtual POIs within a specified range with the RW paths from $st_{\text{s}}$ following the RW cost constraint. 

For D$k$NN, a computationally intensive approach is to find the RW path for every POI with DEWN and then extract the top-$k$ solution. In contrast, the resurging Incremental Euclidean Restriction (IER) algorithm \cite{kNN16} for $k$NN can solve D$k$NN more efficiently, by exploiting DEWN to find promising RW paths. More specifically, upon retrieval of the next Euclidean NN $\gamma^{\text{v}}$ in the virtual world, instead of using Dijkstra's algorithm to evaluate the path distance from $\gamma^{\text{v}}_{\text{s}}$ (the v-state of $st_{\text{s}}$) to $\gamma^{\text{v}}$ in the v-graph, DEWN can be invoked to find the v-path length of the RW path from $st_{\text{s}}$ to $\gamma^{\text{v}}$ following the RW constraint $C$. Moreover, the results for previous DEWN queries can be reused as pruning criteria in determining v-path lengths for subsequent candidate NN's, similar to the idea of Pruned Landmark Labeling \cite{TA15}. 

Moreover, ROAD \cite{ROAD} for $k$NN and range queries can incorporate MIL Range in DEWN to support DR. After ROAD partitions the network into multiple \textit{Regional Subnetworks} (Rnets), in addition to precomputing the path distances between each pair of border nodes of an Rnet, potential RW costs between pairs of border nodes can also be obtained by aggregating MIL Ranges on the v-edges. Furthermore, DEWN queries with different RW cost constraints can be issued to find multiple v-path lengths and the corresponding RW costs between border nodes for constructing multiple shortcut RW subpaths to bypass the Rnets. These shortcuts enable the traversal algorithm to bypass sparse areas containing few POIs without examining the detailed paths inside the Rnets, thereby achieving significant speedups in D$k$NN and DR queries. 

\section{Experiments}
\label{sec:exp}

In this section, we evaluate DEWN against several state-of-the-art algorithms on real datasets for various VR application scenarios. 

\subsection{Experiment Setup and Evaluation Plan}
We collect virtual maps for VR traveling and gaming scenarios, and physical maps from real indoor spatial layouts. For VR traveling, real spatial datasets (POIs and their spatial information) are extracted from OpenStreetMap\footnote{\url{https://www.openstreetmap.org/}.} where convex-hull corners of objects (e.g., buildings, lakes) are added to the location sets according to \cite{EM04} to build the visibility graphs. The numbers of \textit{virtual locations} are 40k in \textit{Seattle}, 79k in \textit{Boston}, 110k in \textit{Taipei}, and 564k in \textit{Yellow Stone}. For VR gaming, maze-puzzle layouts are collected from a maze generator project\footnote{\url{https://github.com/boppreh/maze}.} in which all turning corners are regarded as the locations in visibility graphs, and the number of locations ranges from 625 to 2025. Physical layouts are real indoor layouts\footnote{\url{https://bit.ly/2DZXQLv}, \url{https://bit.ly/2DZXS65}.} divided into up to 672 grid cells of 0.3m $\times$ 0.3m (body-sized areas), where a cell is either an empty cell (free space) or a part of an obstacle. For each combination of virtual and physical maps, 100 samples are generated with random start and destination locations. RW costs are derived according to the detection thresholds in \cite{FS10,SS13,FM16}. 

We compare DEWN with five baselines: Basic DP (DP), Minimum Cost Path (MCP), Constrained Labeling (COLA) \cite{SW16}, k-Shortest-Path (kSP) \cite{TA15}, and Simplified DEWN (S-DEWN). DP is the dynamic programming baseline proposed in Section \ref{sec:dp}. MCP focuses on finding the RW path with the minimum RW cost (instead of minimum length of v-path) via Dijkstra's algorithm on the loco-state space. COLA exploits only the v-graph with each v-edge associated with an \textit{estimated} RW cost\footnote{The estimated RW cost for a v-edge with length $l$ is set to $\beta(l)$ (the MIL upper bound) as a safe estimate, since COLA does not process the physical world.} then finds the shortest v-path such that the total \textit{estimated} RW cost along the v-path follows the RW cost constraint. 4) kSP first finds the top-$k$ minimum-length v-paths. Reset is adopted when obstacles in the physical world are reached. The v-path with the minimum RW cost among the $k$ candidates is returned. We also implement a more scalable variation of DEWN, namely S-DEWN, that 1) removes the orientation information from loco-states and 2) directly returns the reference path (without orientation information) as the solution. Note that only DEWN and DP have theoretical guarantees on both feasibility and solution quality, while MCP only ensures the feasibility, and COLA, kSP and S-DEWN have none. 

We first evaluate all algorithms in Section \ref{subsec:exp_diff_alg} with the following metrics: 1) v-path length, 2) incurred RW cost, 3) average feasibility (proportion of solutions satisfying the RW cost constraint), and 4) running time (in seconds). Afterward, Section \ref{subsec:exp_scenario} evaluates all methods in various VR scenarios: urban traveling (\textit{Seattle}), natural traveling (\textit{Yellow Stone}) and a maze gaming map. Section \ref{subsec:exp_on_off} examines the efficacy of various pruning and ordering strategies proposed in DEWN. Sections \ref{subsec:exp_sensitivity} and \ref{subsec:large_exp} conducts sensitivity and scalability tests on various query parameters. Finally, to understand users' behaviors in real VR applications, a user study is detailed in Section \ref{subsec:userstudy}. The default parameters are $k=5$ for kSP, and $\epsilon = 0.1$ for all algorithms. All algorithms are implemented on an HP DL580 Gen 9 server with an Intel 2.10GHz CPU and 1TB RAM.

\subsection{Comparison of Different Algorithms} \label{subsec:exp_diff_alg}

\begin{figure}[t]
\centering
  \begin{subfigure}[t]{0.49\columnwidth}
    \includegraphics[width=\linewidth, height=2.35cm]{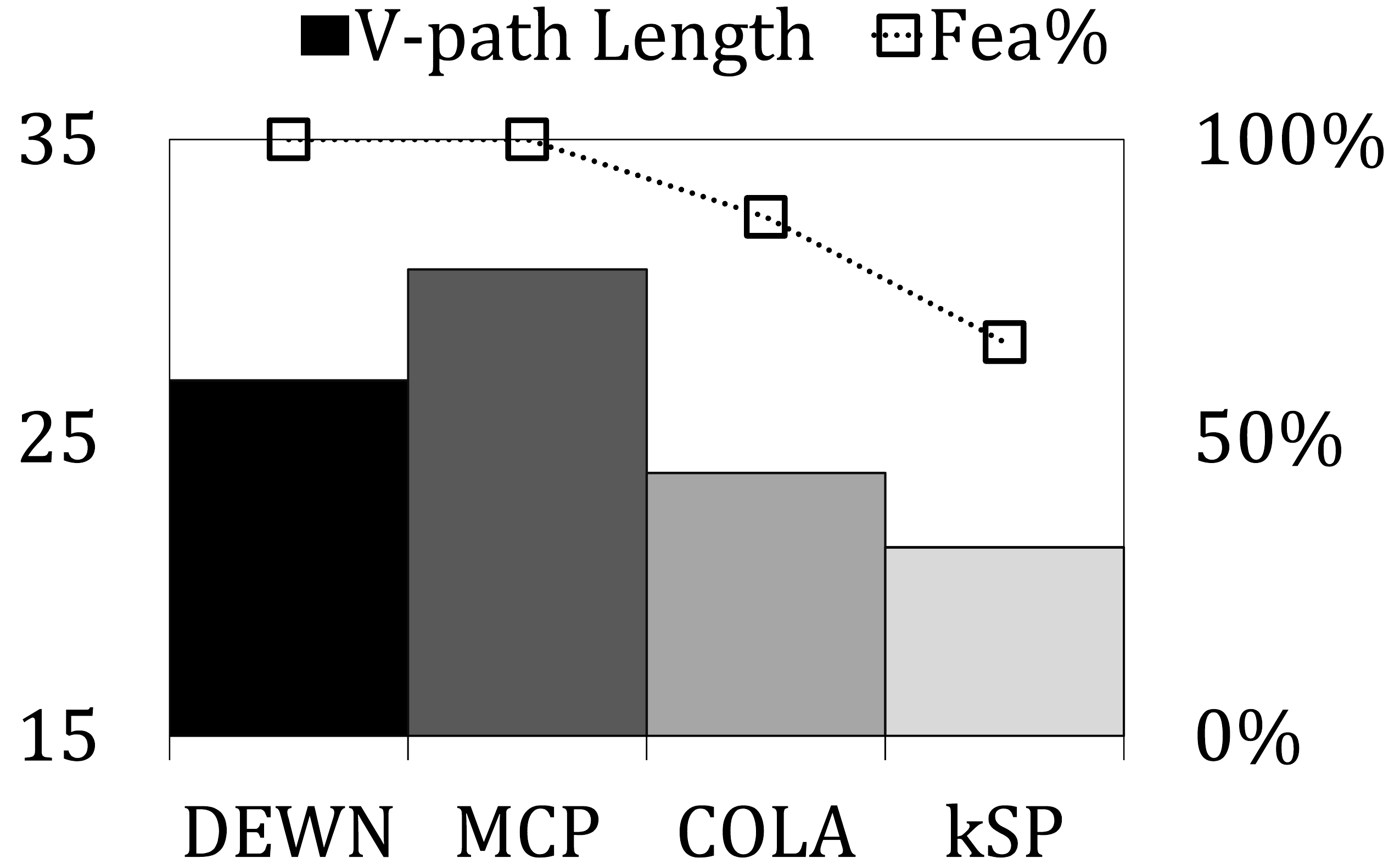}
    \caption{ \centering  V-path length and fea. ratio (\textit{Boston}).}
    \label{fig:bos_fea_obj}
  \end{subfigure}
  \hfill
  \begin{subfigure}[t]{0.49\columnwidth}
    \includegraphics[width=\linewidth, height=2.35cm]{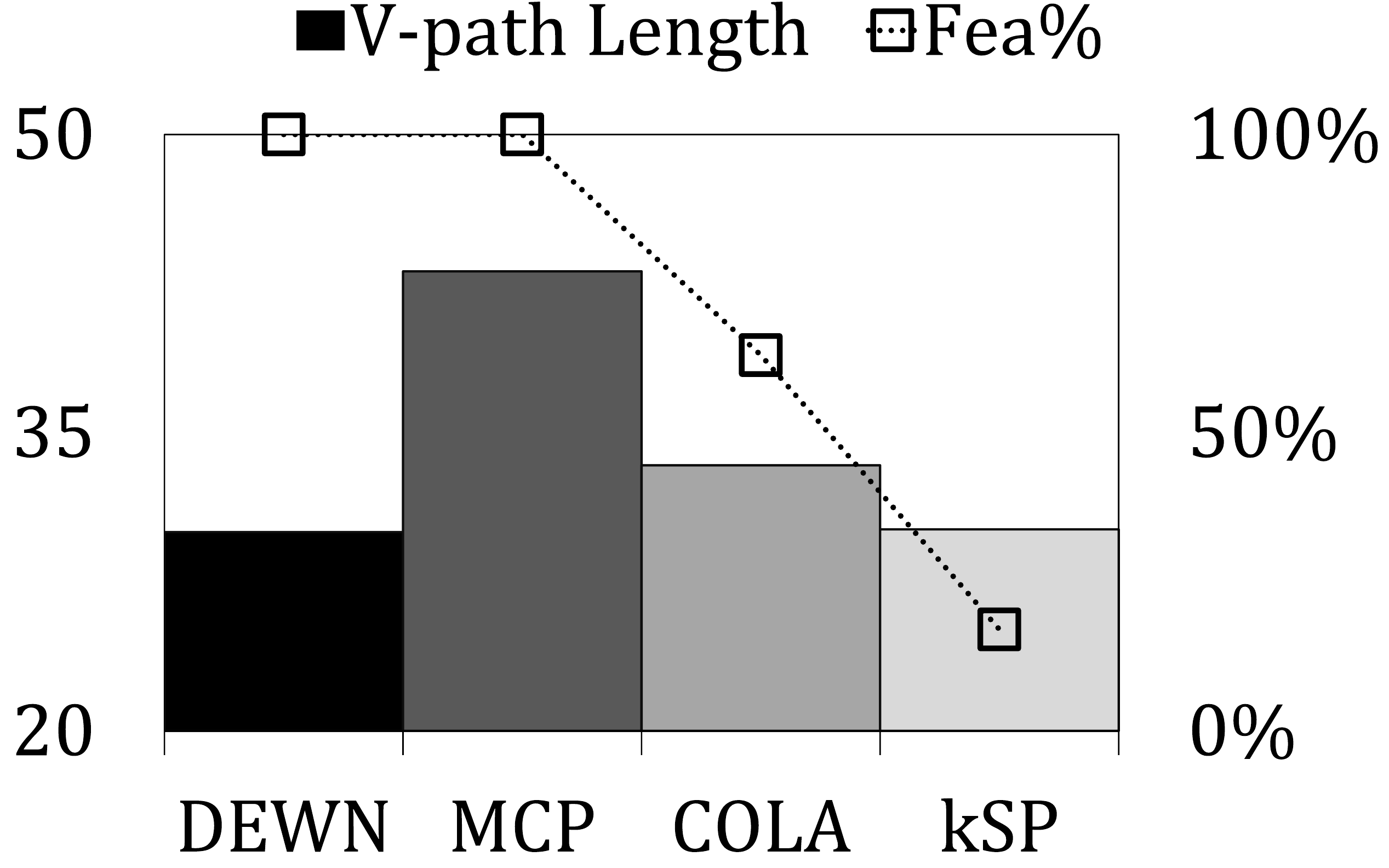}
    \caption{ \centering V-path length and fea. ratio (\textit{Taipei}).}
    \label{fig:tp_fea_obj}
  \end{subfigure}\\
  
  \begin{subfigure}[t]{0.49\columnwidth}
    \includegraphics[width=\linewidth, height=2.35cm]{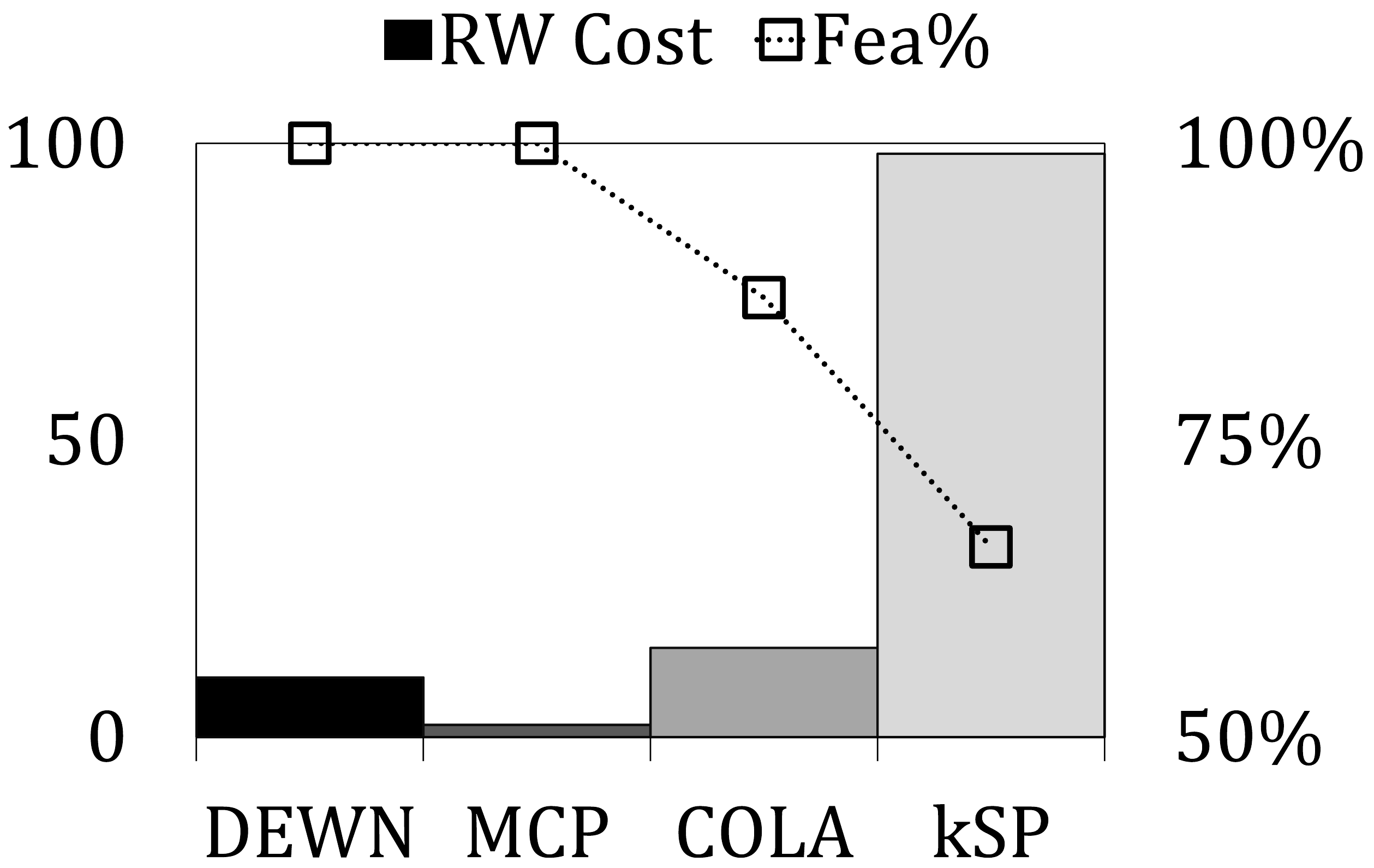}
    \caption{ \centering RW cost and fea. ratio (\textit{Boston}).}
    \label{fig:bos_fea_cost}
  \end{subfigure}
  \hfill
  \begin{subfigure}[t]{0.49\columnwidth}
    \includegraphics[width=\linewidth, height=2.35cm]{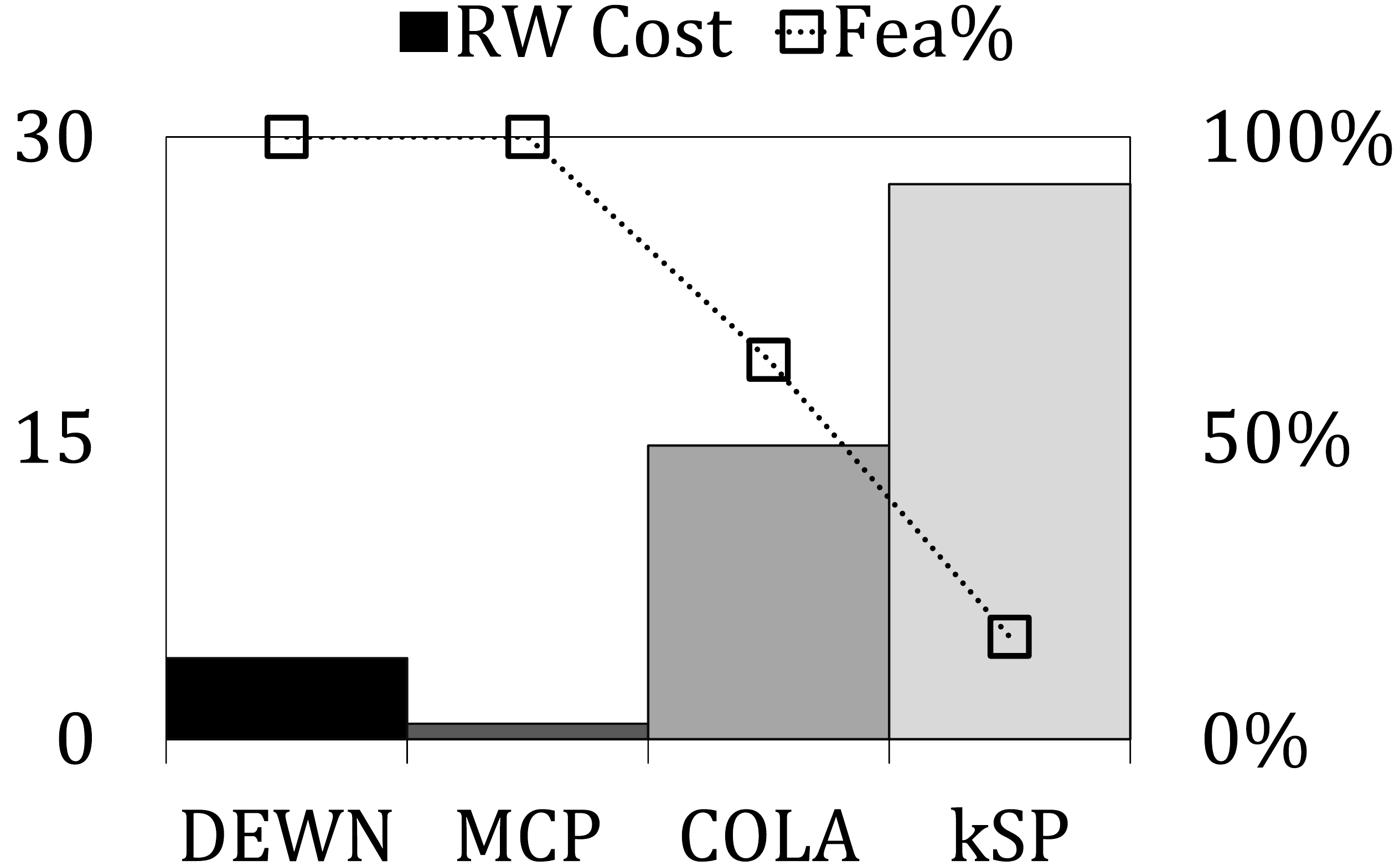}
    \caption{ \centering RW cost and fea. ratio (\textit{Taipei}).}
    \label{fig:tp_fea_cost}
  \end{subfigure}\\

  \begin{subfigure}[t]{0.49\columnwidth}
    \includegraphics[width=\linewidth, height=2.35cm]{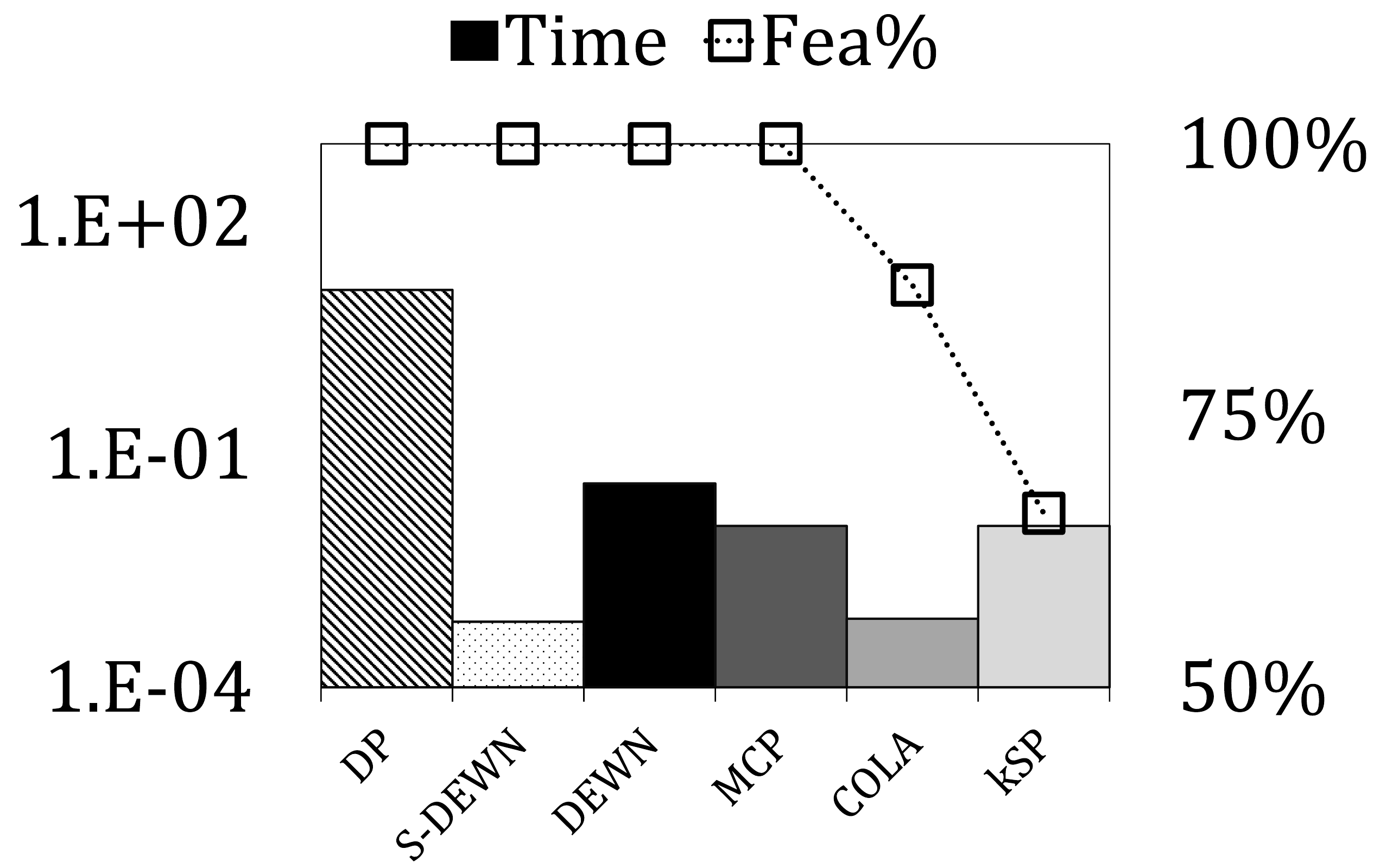}
    \caption{ \centering Running time (\textit{Boston}).}
    \label{fig:bos_fea_time}
  \end{subfigure}
  \hfill
  \begin{subfigure}[t]{0.49\columnwidth}
    \includegraphics[width=\linewidth, height=2.35cm]{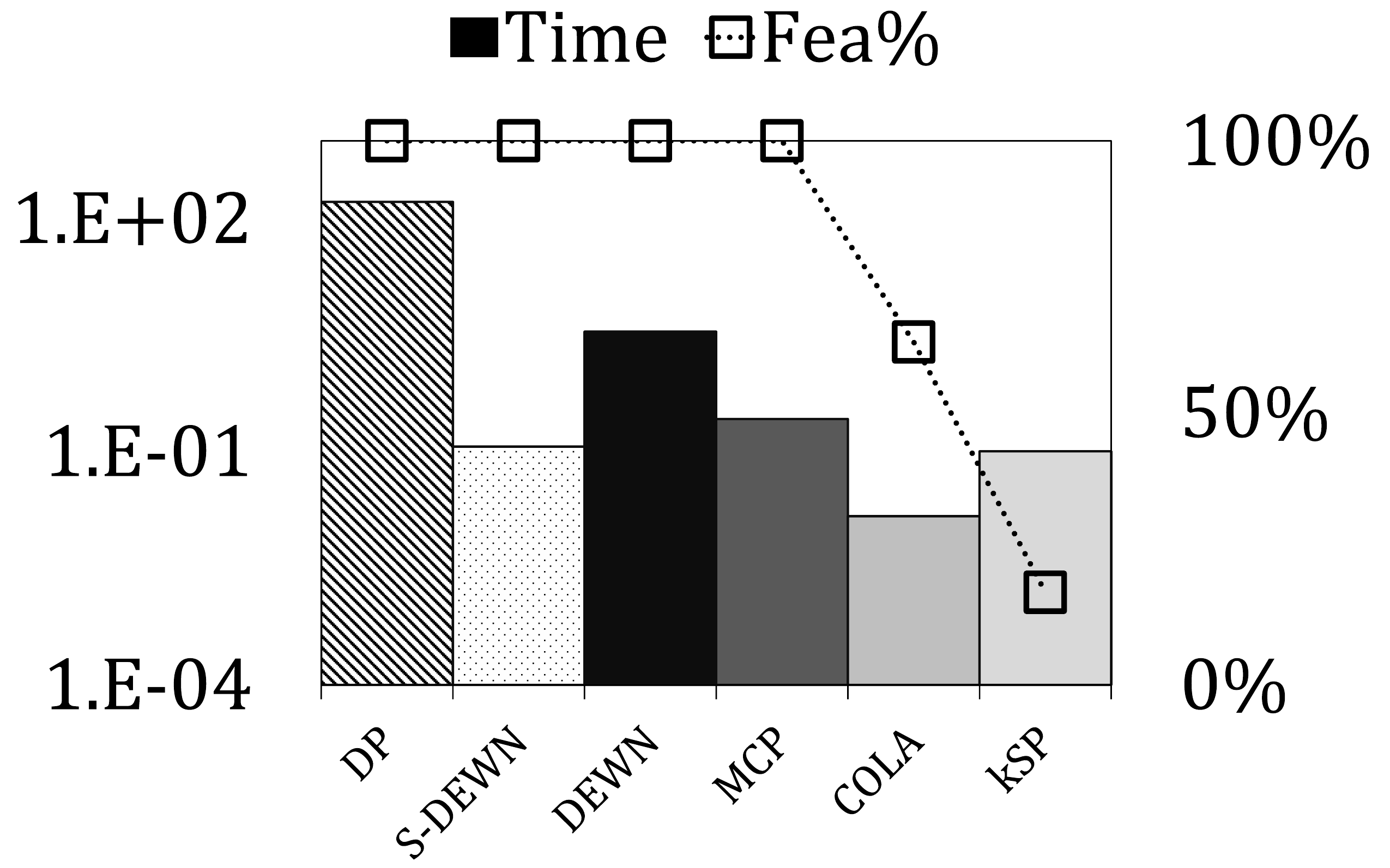}
    \caption{ \centering Running time (\textit{Taipei}).}
    \label{fig:tp_fea_time}
  \end{subfigure}
  \hfill 
  \centering \caption{Experimental results in city maps.}
  \label{fig:general_behavior}
\end{figure}

Figure \ref{fig:bos_fea_obj} and \ref{fig:tp_fea_obj} compare the v-path length and feasibility of all algorithms.\footnote{V-path lengths are averaged only from feasible solutions. S-DEWN and DP share similar results with DEWN and are not shown here.} DEWN, DP, and MCP achieve 100\% feasibility by carefully examining both worlds. Compared with MCP, DEWN generates shorter v-paths by leveraging LR to properly allocate the RW cost budget. The feasibility of kSP is poor because the RW cost is not carefully reduced during path search, as shown in Figure \ref{fig:bos_fea_cost} and \ref{fig:tp_fea_cost}. Figure \ref{fig:bos_fea_time} and \ref{fig:tp_fea_time} show the running time of all algorithms. kSP and COLA are efficient since they are designed for a single world (v-graph here) without ensuring the feasibility, whereas DP, DEWN, and MCP explore loco-states on both worlds. Compared with DP, DEWN generates feasible solutions with much smaller time because the pruning strategies effectively trim off redundant loco-states. S-DEWN is even faster as it only invokes dual-world simplification and reference path generating phases on v-graph, but it does not provide any theoretical guarantee.

\subsection{Comparisons on Different VR Scenarios} \label{subsec:exp_scenario}
\begin{figure}[t]
\centering
  \begin{subfigure}[t]{0.49\columnwidth}
    \includegraphics[width=\linewidth, height=2.35cm]{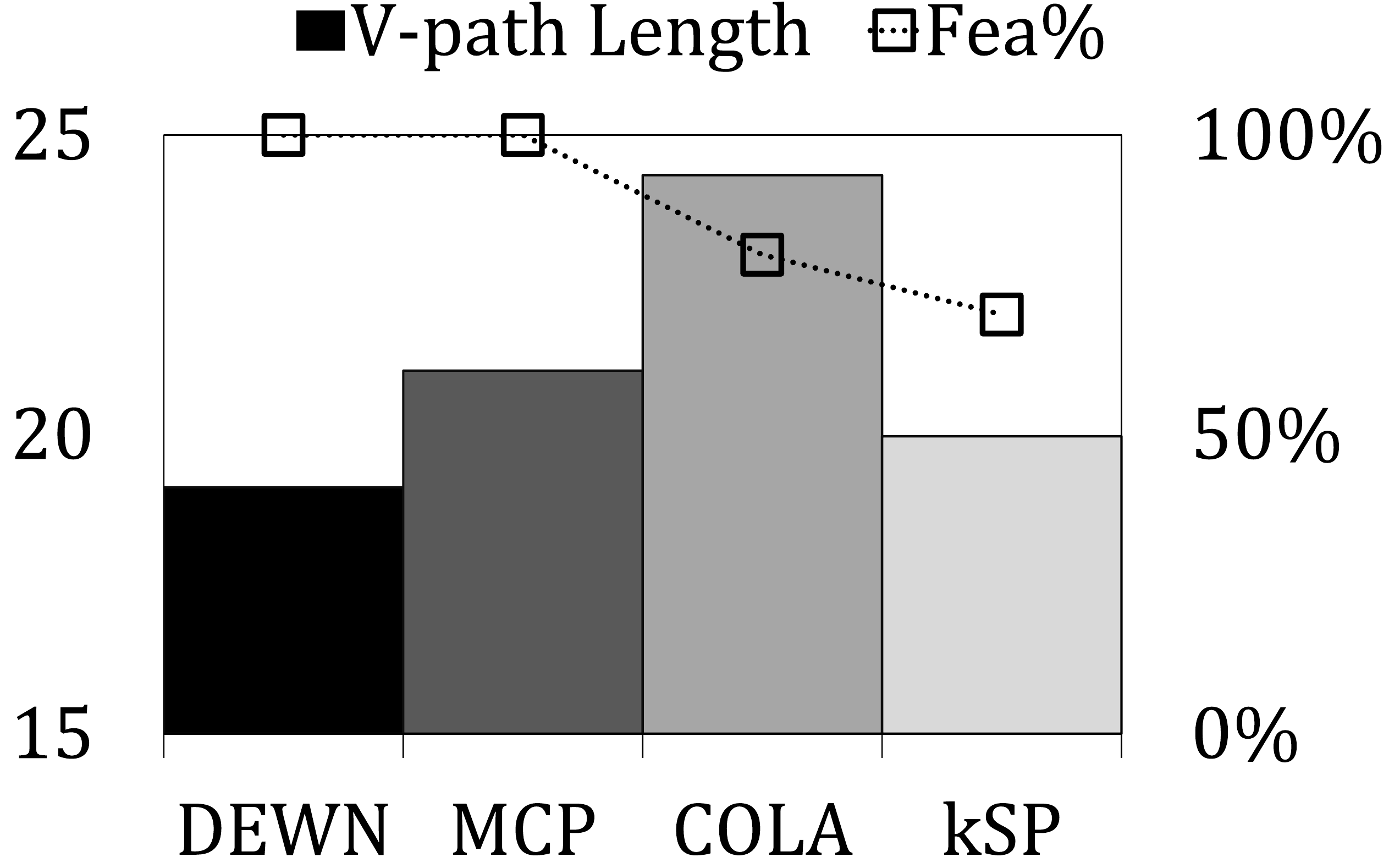}
    \caption{ \centering  V-path length and fea. ratio (\textit{Seattle}).}
    \label{fig:urban_fea_obj}
  \end{subfigure}
  \hfill
  \begin{subfigure}[t]{0.49\columnwidth}
    \includegraphics[width=\linewidth, height=2.35cm]{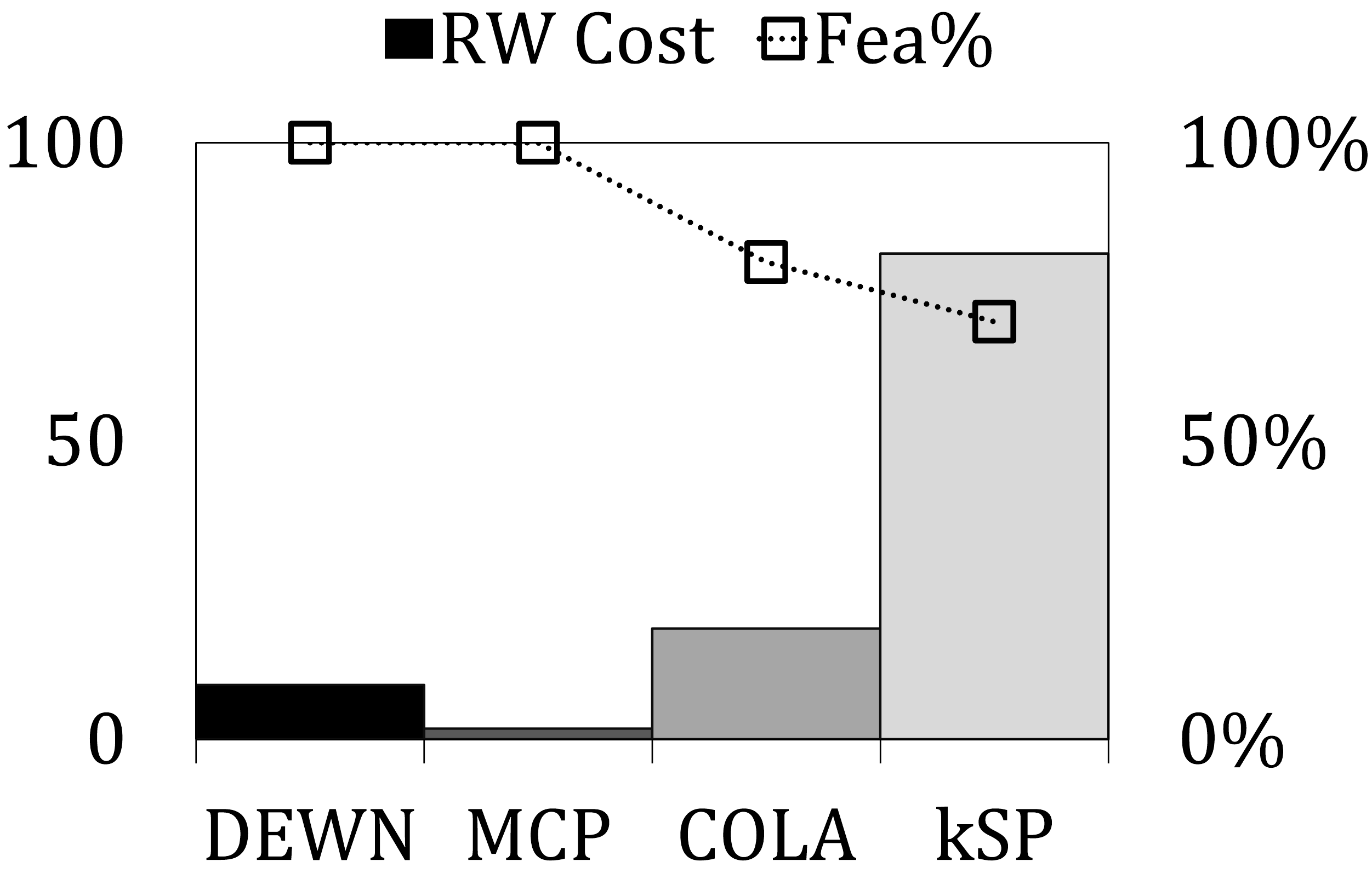}
    \caption{ \centering RW cost and fea. ratio (\textit{Seattle}).}
    \label{fig:urban_fea_cost}
  \end{subfigure}\\
  
  \begin{subfigure}[t]{0.49\columnwidth}
    \includegraphics[width=\linewidth, height=2.35cm]{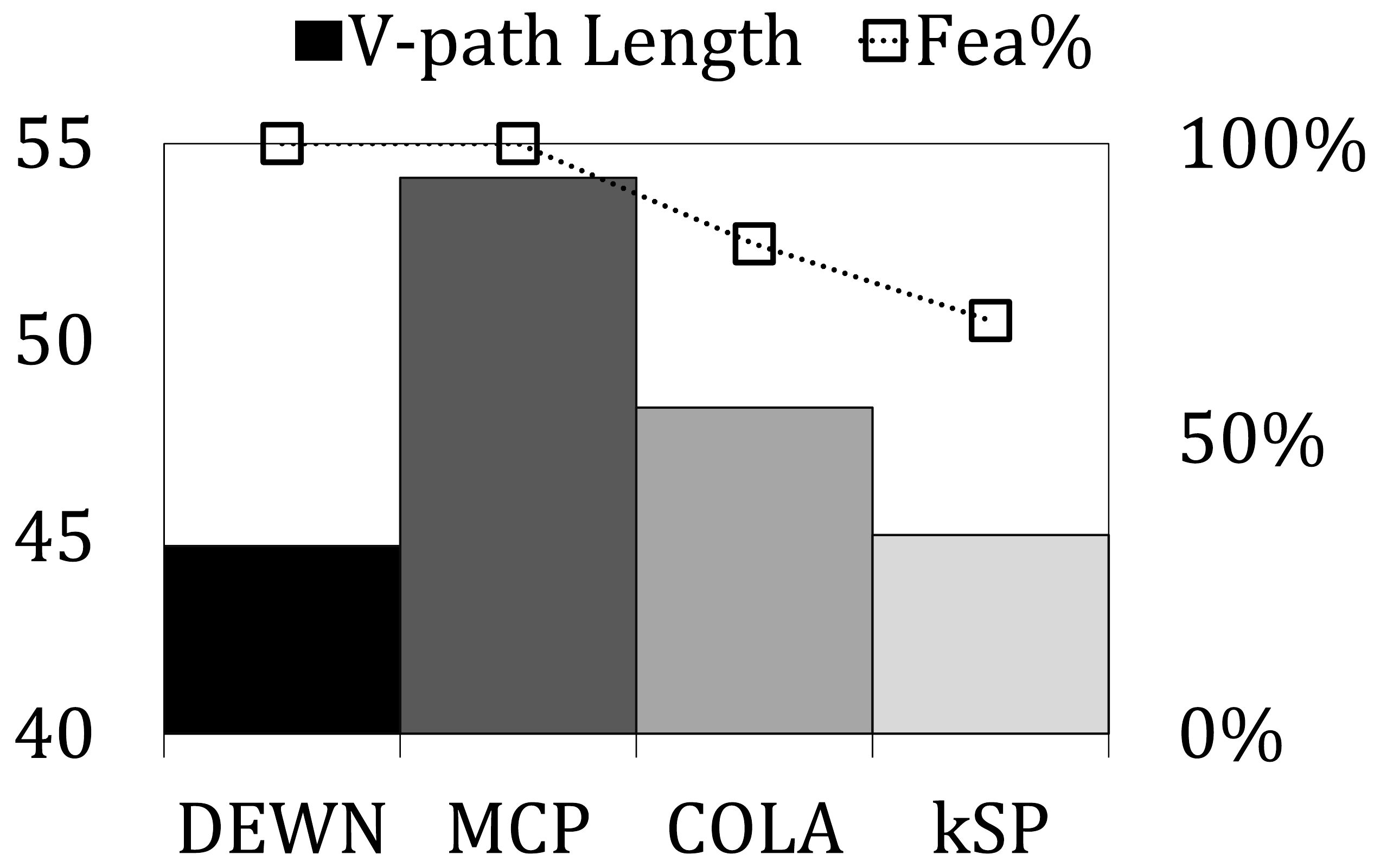}
    \caption{ \centering V-path length and fea. ratio (\textit{Yellow Stone}).}
    \label{fig:nature_fea_obj}
  \end{subfigure}
  \hfill
  \begin{subfigure}[t]{0.49\columnwidth}
    \includegraphics[width=\linewidth, height=2.35cm]{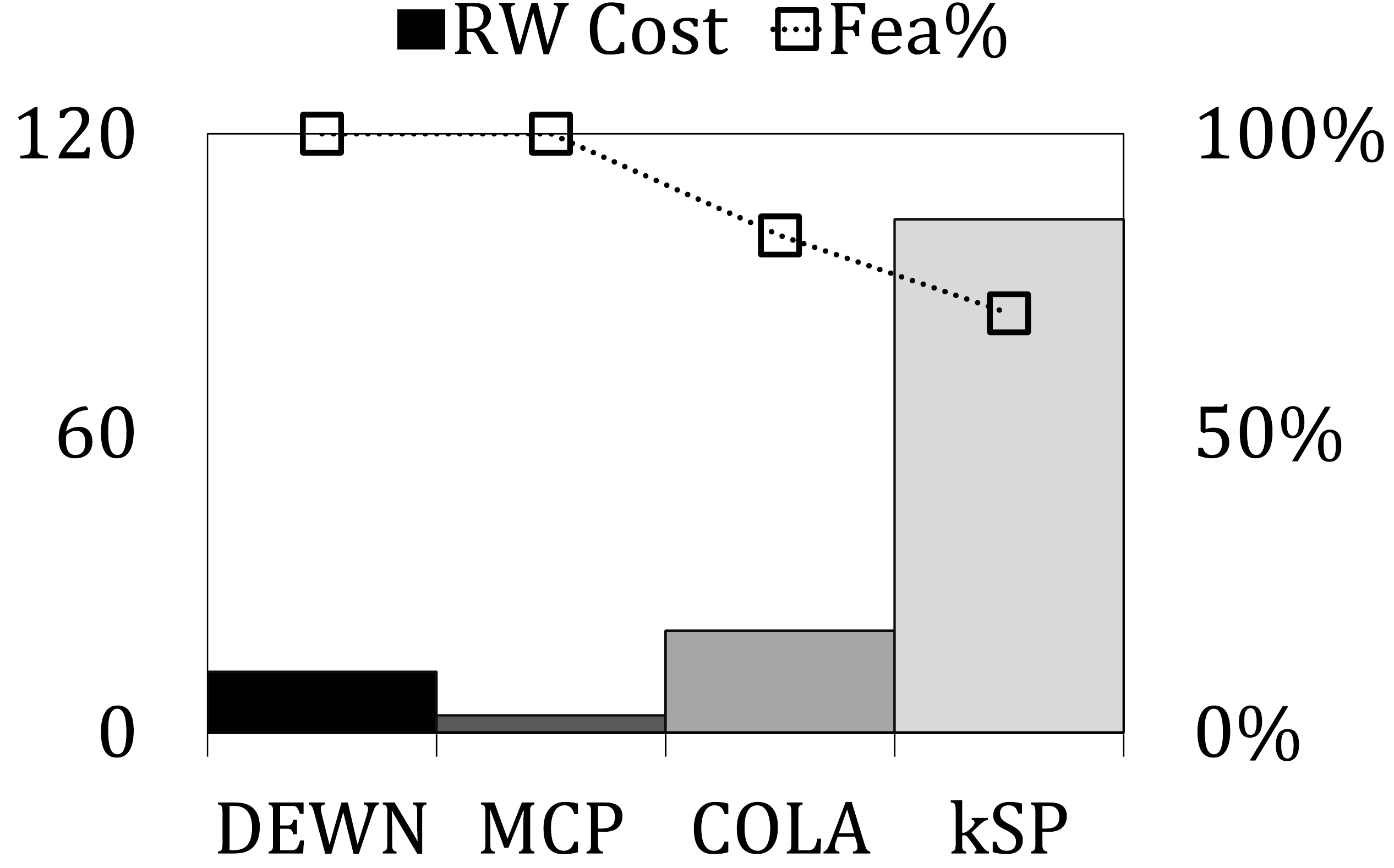}
    \caption{ \centering RW cost and fea. ratio (\textit{Yellow Stone}).}
    \label{fig:nature_fea_cost}
  \end{subfigure}\\

  \begin{subfigure}[t]{0.49\columnwidth}
    \includegraphics[width=\linewidth, height=2.35cm]{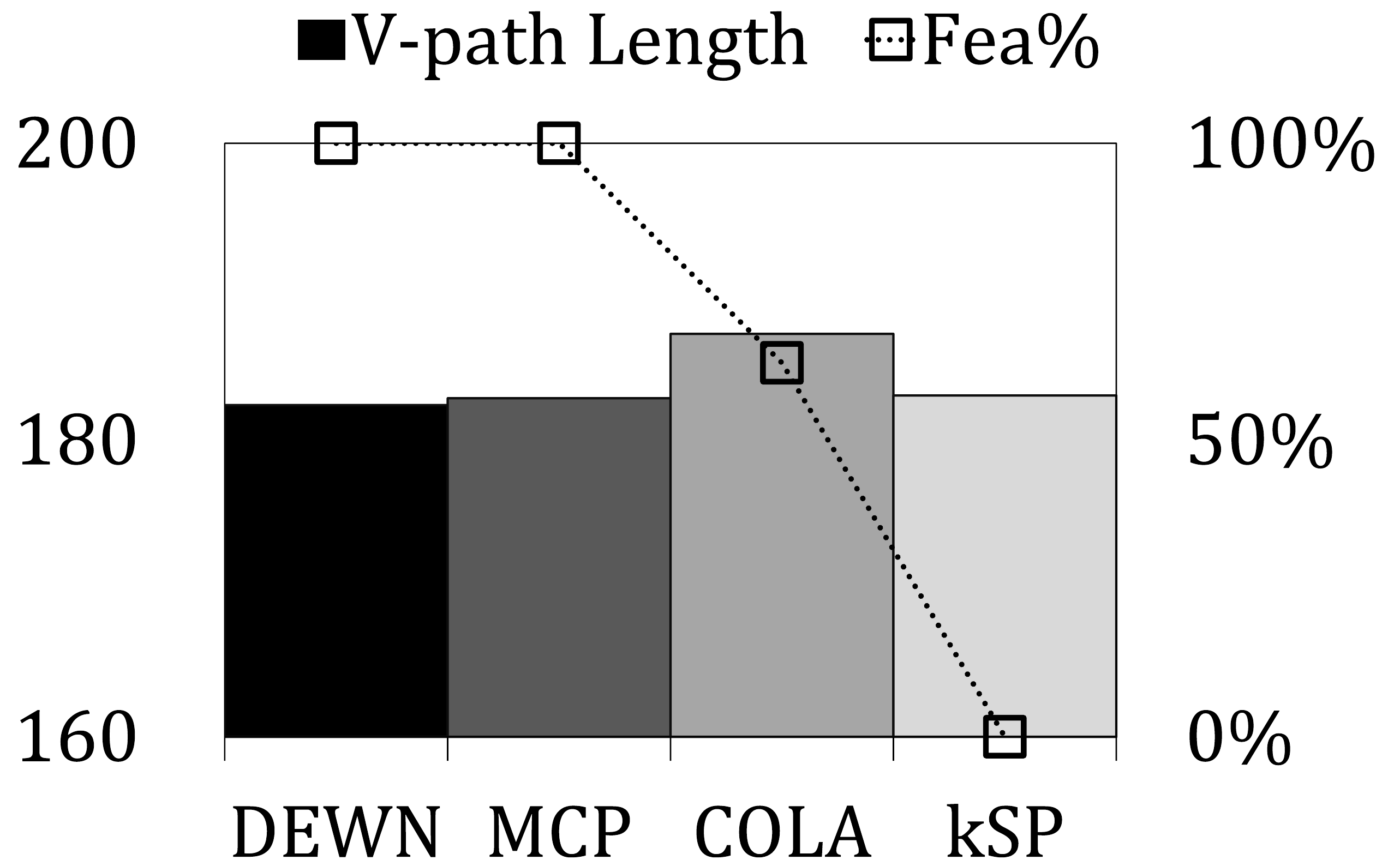}
    \caption{ \centering V-path length and fea. ratio (\textit{Maze}).}
    \label{fig:maze_fea_obj}
  \end{subfigure}
  \hfill
  \begin{subfigure}[t]{0.49\columnwidth}
    \includegraphics[width=\linewidth, height=2.35cm]{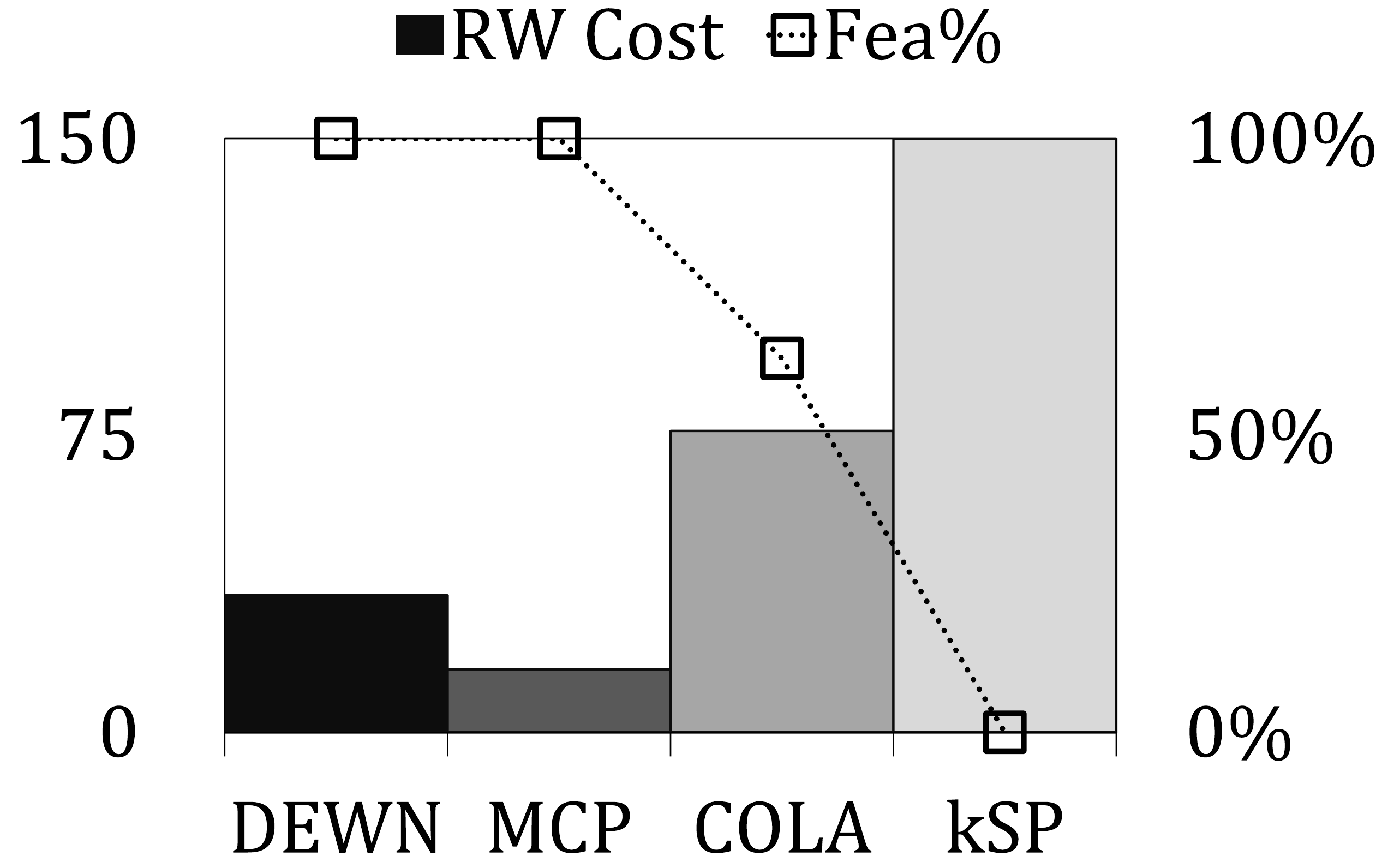}
    \caption{ \centering RW cost and fea. ratio (\textit{Maze}).}
    \label{fig:maze_fea_cost}
  \end{subfigure}
  \centering \caption{Comparisons on different VR scenarios.}
  \label{fig:diff_scenarios}
\end{figure}

Figure \ref{fig:diff_scenarios} compares \textit{Seattle}, \textit{Yellow Stone} and \textit{Maze} (scaled to similar sizes for a fair comparison) with open space ratios of 50.1\%, 98.2\%, and 47.7\%, respectively.\footnote{The result of S-DEWN is similar to DEWN and thereby not shown, and DP does not scale up here.} Note that the v-path lengths in Figures \ref{fig:urban_fea_obj}, \ref{fig:nature_fea_obj} and \ref{fig:maze_fea_obj} are averaged only from feasible solutions. Therefore, under the same RW cost constraint, feasible v-paths in \textit{Seattle} are shorter than those in \textit{Yellow Stone} since \textit{Seattle} involves more buildings and obstacles and thereby requires more RW operations. The feasibility ratios of COLA and kSP are much lower in \textit{Maze} (especially, 0\% for kSP) since there are only a few v-paths connecting the source and destination, and the challenge thus becomes identifying the p-path following the RW cost constraint, as the v-paths returned by all algorithms are similar. Figures \ref{fig:urban_fea_cost}, \ref{fig:nature_fea_cost} and \ref{fig:maze_fea_cost} manifest that single-world COLA and kSP are difficult to meet the RW cost constraint because the physical-world layouts are not investigated. In contrast, MCP focuses on reducing the RW cost, but its v-paths are longer than those of DEWN, especially in \textit{Yellow Stone} where the virtual world contains mainly free spaces and thus easier to be optimized.

\begin{figure}[t]
\centering
  \begin{subfigure}[b]{0.48\columnwidth}
    \includegraphics[width=\linewidth, height=2.35cm]{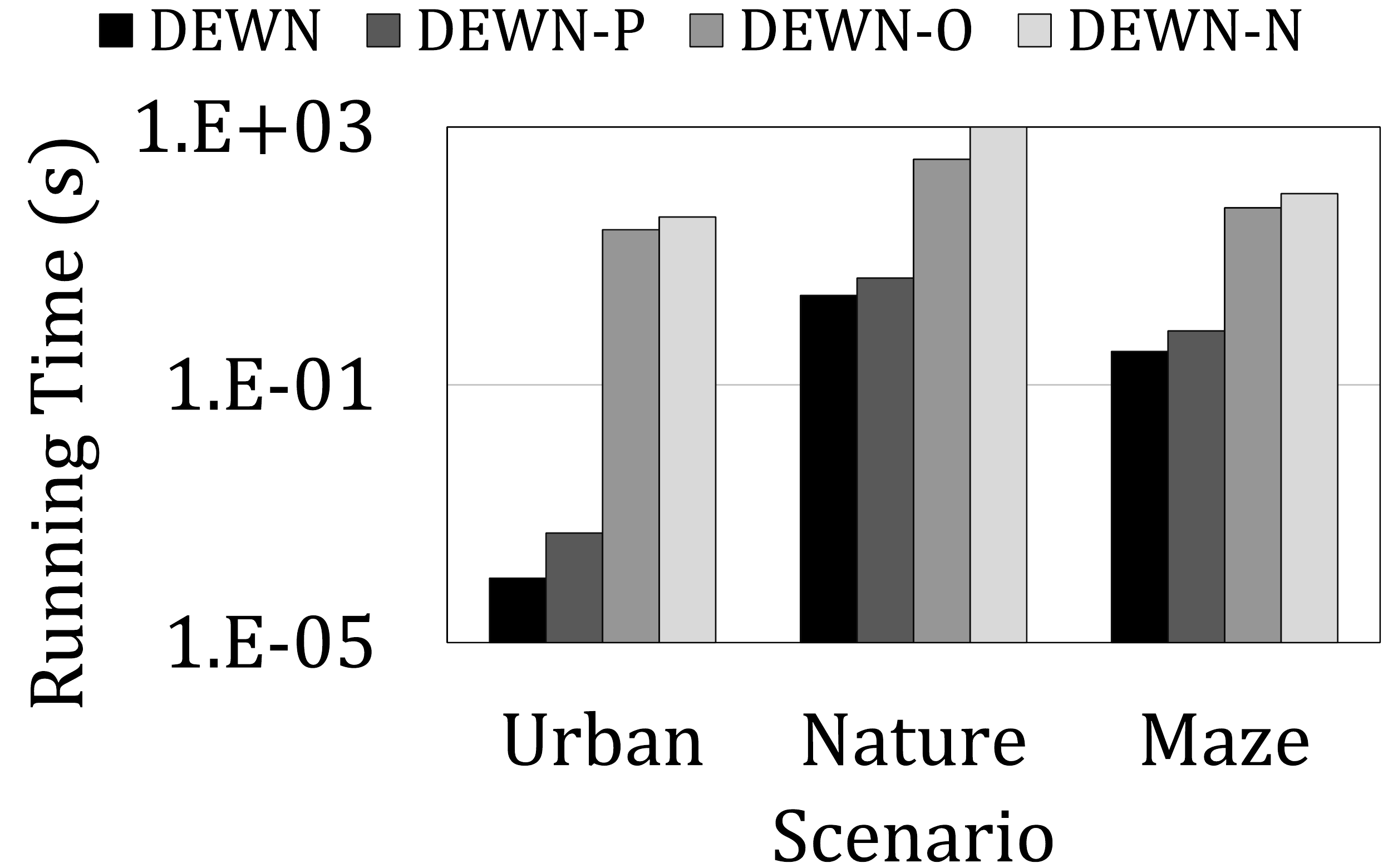}
    \caption{ \centering Effects on different \newline scenarios.}
    \label{fig:switch_scenarios}
  \end{subfigure}
  \hfill 
  \begin{subfigure}[b]{0.48\columnwidth}
    \includegraphics[width=\linewidth, height=2.35cm]{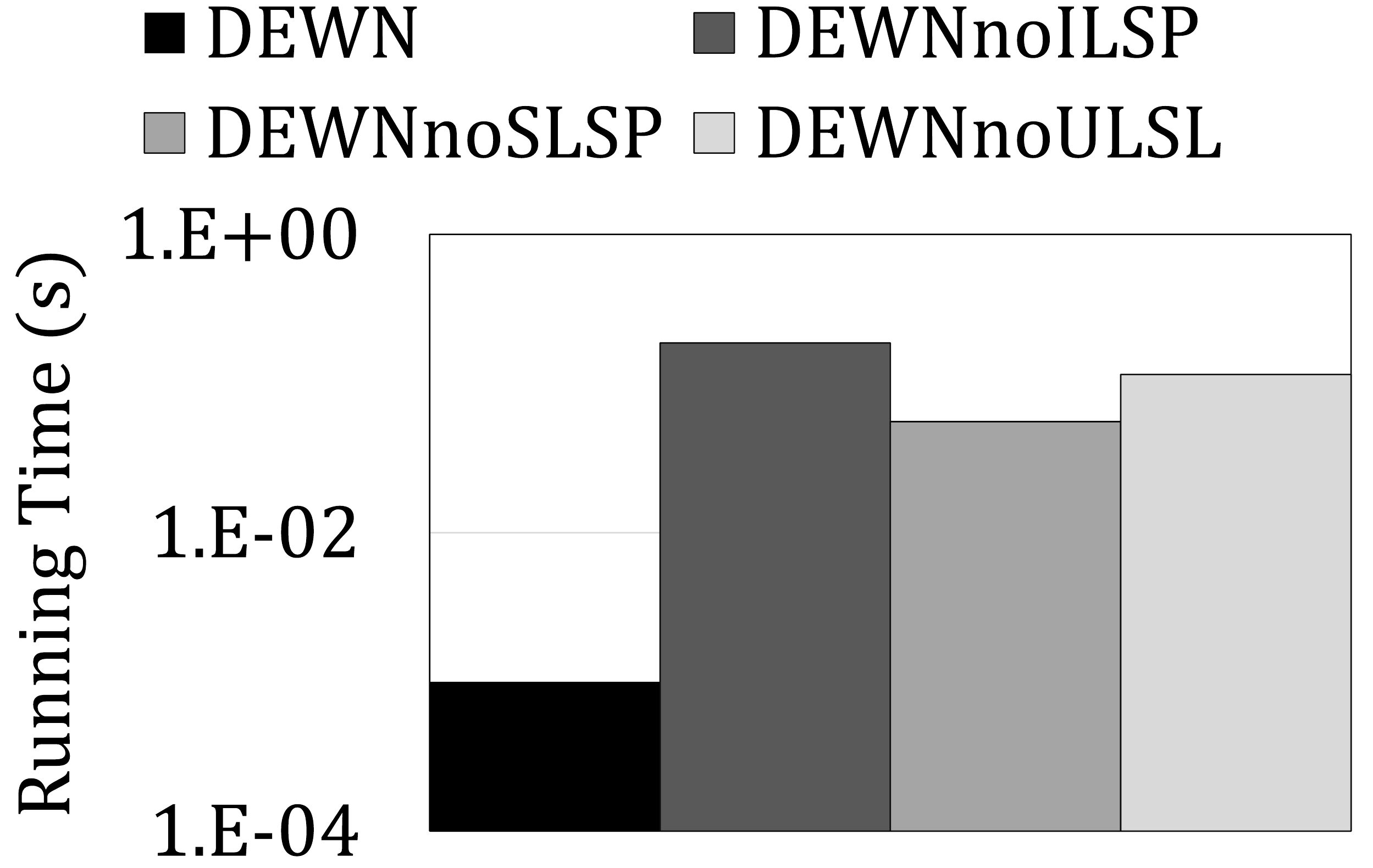}
    \caption{ \centering Effects of pruning \newline (\textit{Seattle}).}
    \label{fig:prune-urban}
  \end{subfigure}
  \hfill
  \begin{subfigure}[b]{0.48\columnwidth}
    \includegraphics[width=\linewidth, height=2.35cm]{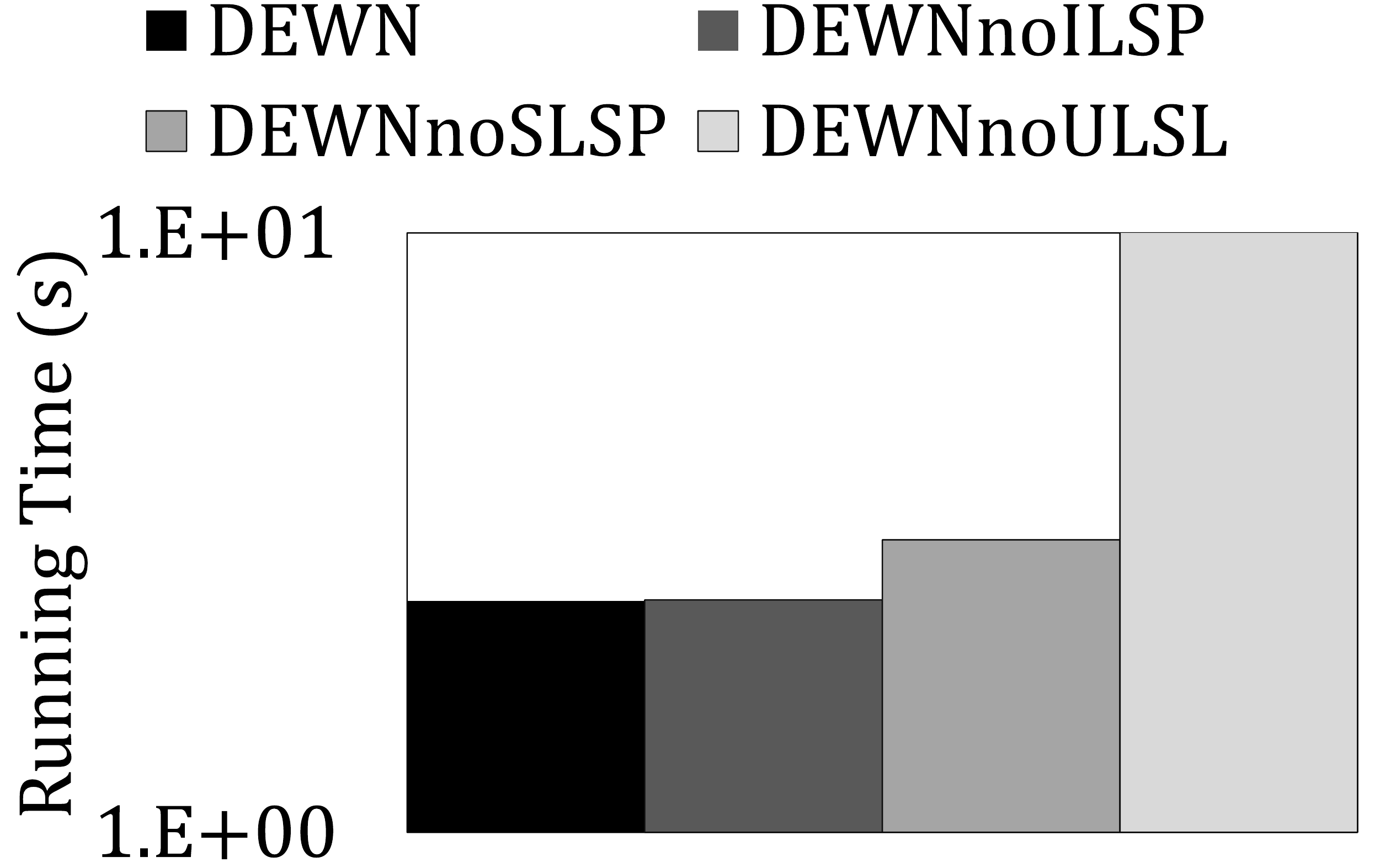}
    \caption{ \centering Effects of pruning \newline (\textit{Yellow Stone}).}
    \label{fig:prune-nature}
  \end{subfigure}
  \hfill
  \begin{subfigure}[b]{0.48\columnwidth}
    \includegraphics[width=\linewidth, height=2.35cm]{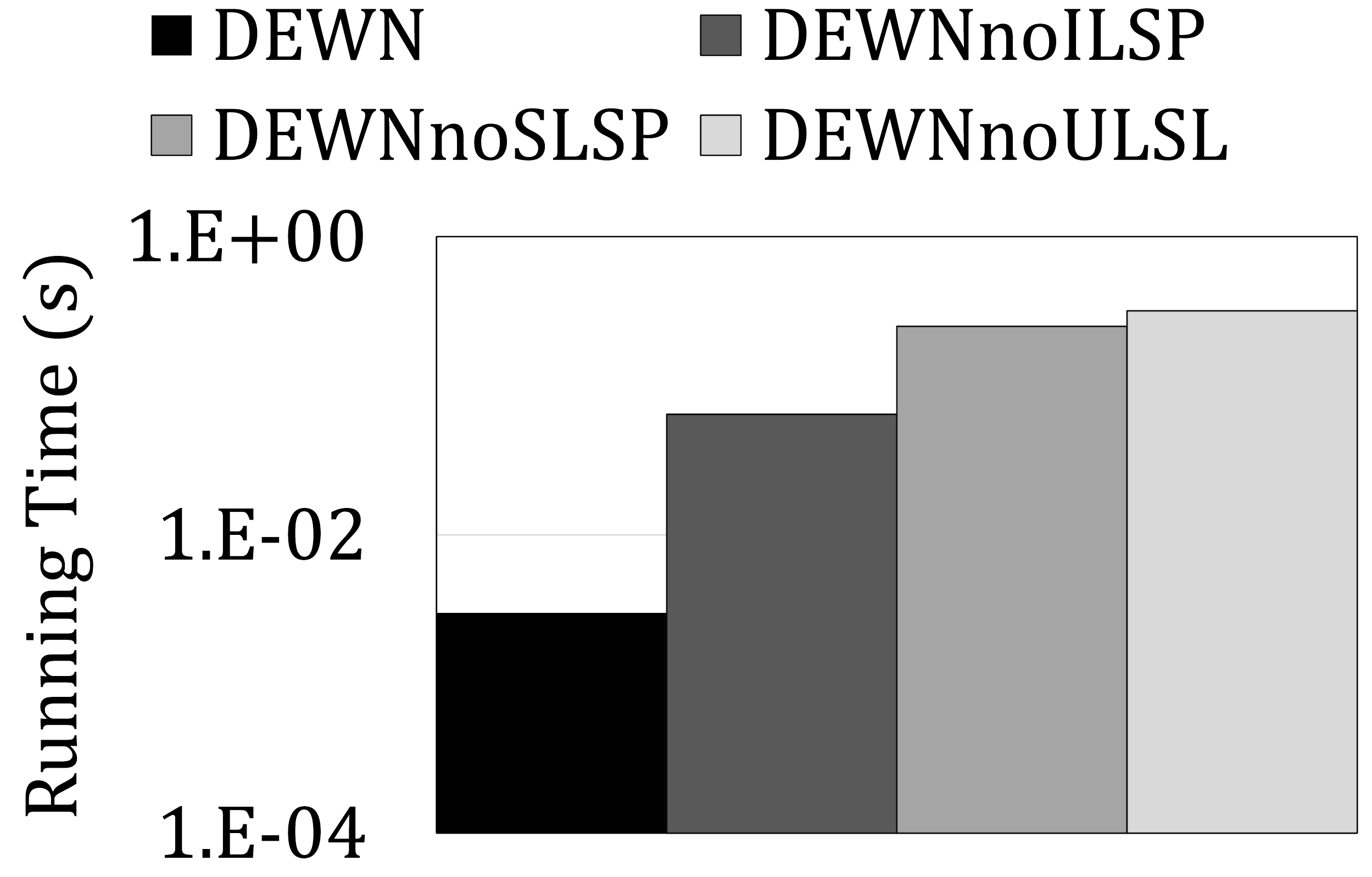}
    \caption{ \centering Effects of pruning \newline (\textit{Maze}).}
    \label{fig:prune-maze}
  \end{subfigure}
  \hfill
  \centering \caption{Effects of pruning and ordering strategies.}
  \label{fig:pruning_ordering}
\end{figure}

\subsection{Evaluation of Ordering and Pruning} \label{subsec:exp_on_off}
Figure \ref{fig:pruning_ordering} shows the efficacy of the pruning strategies (ILSP, SLSP and ULSL) and the ordering strategies (TECO, PWSO and VWNO) where DEWN-P includes only the pruning strategies, DEWN-O incorporates only the ordering strategies, and the na\"ive DEWN-N applies none of them. Figure \ref{fig:switch_scenarios} manifests that all strategies effectively speedup DEWN, and the pruning strategies play more dominant roles in improving efficiency since a massive number of loco-states are effectively removed. In contrast, without pruning, the merits of ordering for DEWN-O are not unveiled because reference paths are not leveraged to truncate redundant search. Figures \ref{fig:prune-urban}, \ref{fig:prune-nature}, and \ref{fig:prune-maze} further show the efficacy of the pruning strategies in different scenarios with the leave-one-out setting. For example, DEWNnoILSP employs only SLSP and ULSL without ILSP. ULSL is very important in \textit{Yellow Stone} since it has more open space, and the v-paths thereby include many long straight segments. Therefore, the loco-state labels in PPNP are acquired in a more straightforward fashion, instead of being iteratively improved. For \textit{Maze}, since there are much fewer v-path candidates, the reference path length tends to be close to the optimal length. Therefore, SLSP can effectively discard many redundant loco-states in most parts of the maze. In contrast, \textit{Seattle} consists of grid-based street layouts with abundant possibilities for v-paths. Therefore, ILSP is more important because it leverages MIL Range to estimate the RW costs in zigzagging RW paths.

\subsection{Sensitivity Test on Query Parameters} \label{subsec:exp_sensitivity}

\begin{figure}[t]
  \begin{subfigure}[b]{0.48\columnwidth}
    \includegraphics[width=\columnwidth, height=2.35cm]{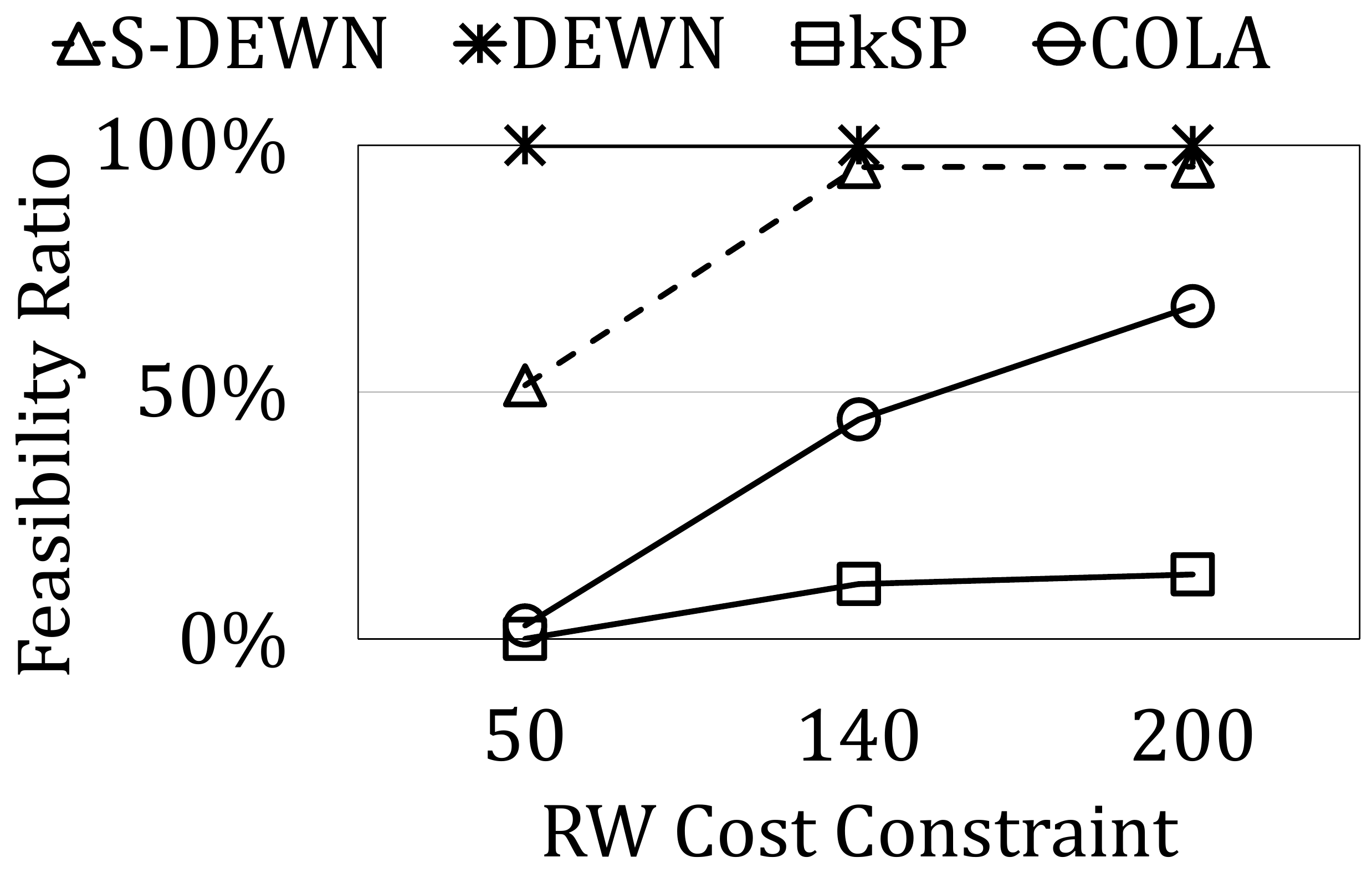}
    \caption{Fea. ratio on diff. $C$.}
    \label{fig:Diff_c}
  \end{subfigure}
  \hfill 
  \begin{subfigure}[b]{0.48\columnwidth}
    \includegraphics[width=\columnwidth, height=2.35cm]{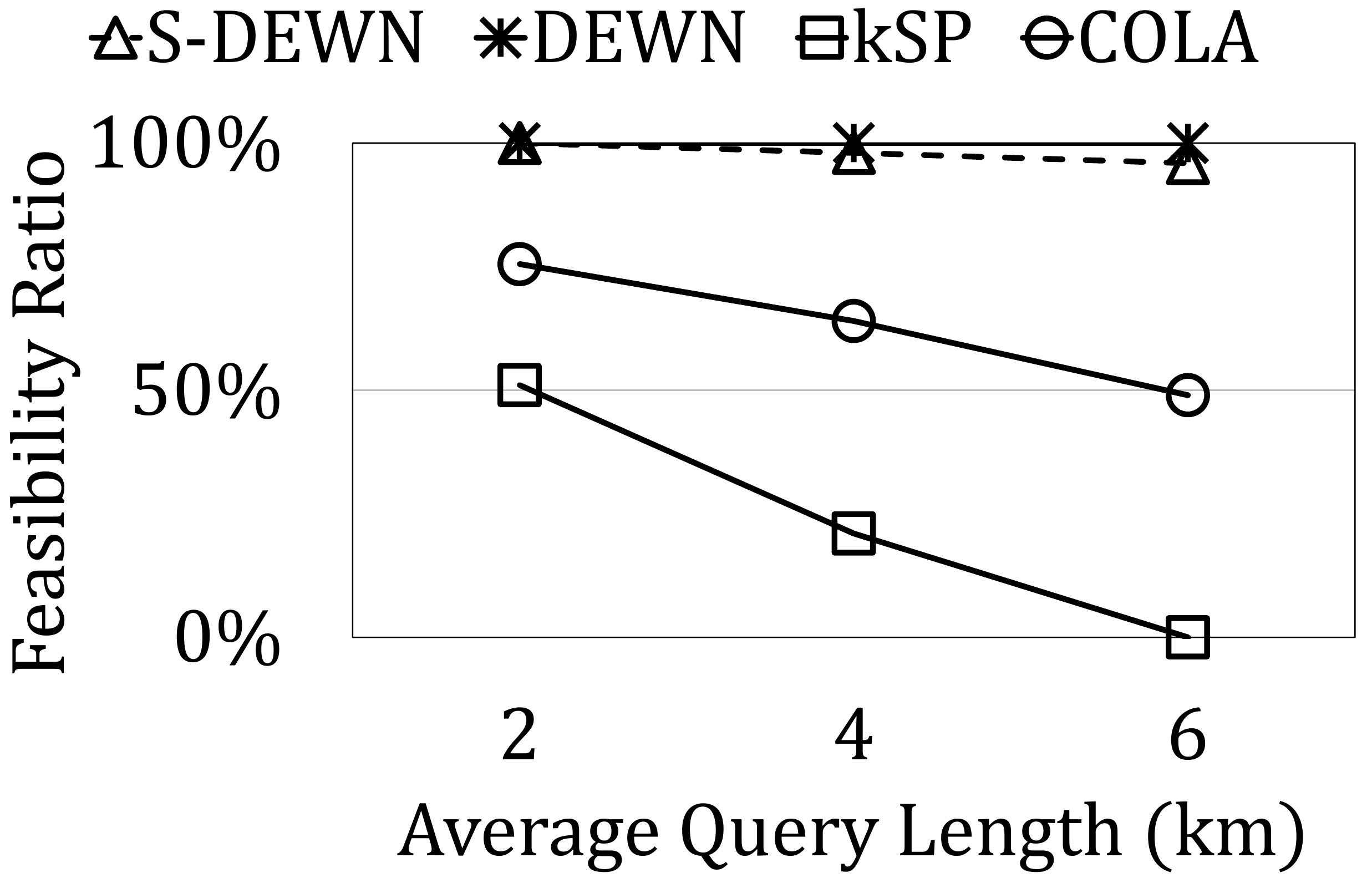}
    \caption{Fea. ratio on diff. $\ell_{\text{q}}$.}
    \label{fig:Diff_ql}
  \end{subfigure}
  \caption{Sensitivity test on query parameters.}
  \label{fig:diff_parameters}
\end{figure}

Figure \ref{fig:diff_parameters} evaluates all algorithms with varied query parameters. Figure \ref{fig:Diff_c} indicates that the feasibility improves with an increasing RW cost constraint $C$, where the results of MCP and DEWN are the same (thereby with only DEWN shown here). DEWN is always feasible since the Pruning and Path Navigation phase guarantees the solution feasibility, but other approaches have difficulty in finding a feasible solution, especially for a small $C$. COLA outperforms kSP because kSP focuses on minimizing the v-path length during the path search. Figure \ref{fig:Diff_ql} compares different algorithms with various straight line distances $\ell_q$ between the start and destination locations in the virtual world. kSP and COLA have difficulty finding feasible solutions when they are far way with more POIs between them, implying less feasible for larger maps. By contrast, DEWN and S-DEWN effectively find feasible RW paths because IDWS derives promising multipliers to balance the v-path length and the RW cost.

\subsection{Scalability Test on Large Virtual Maps} \label{subsec:large_exp}

\begin{figure}[t]
  \begin{subfigure}[b]{0.48\columnwidth}
    \includegraphics[width=\columnwidth, height=2.35cm]{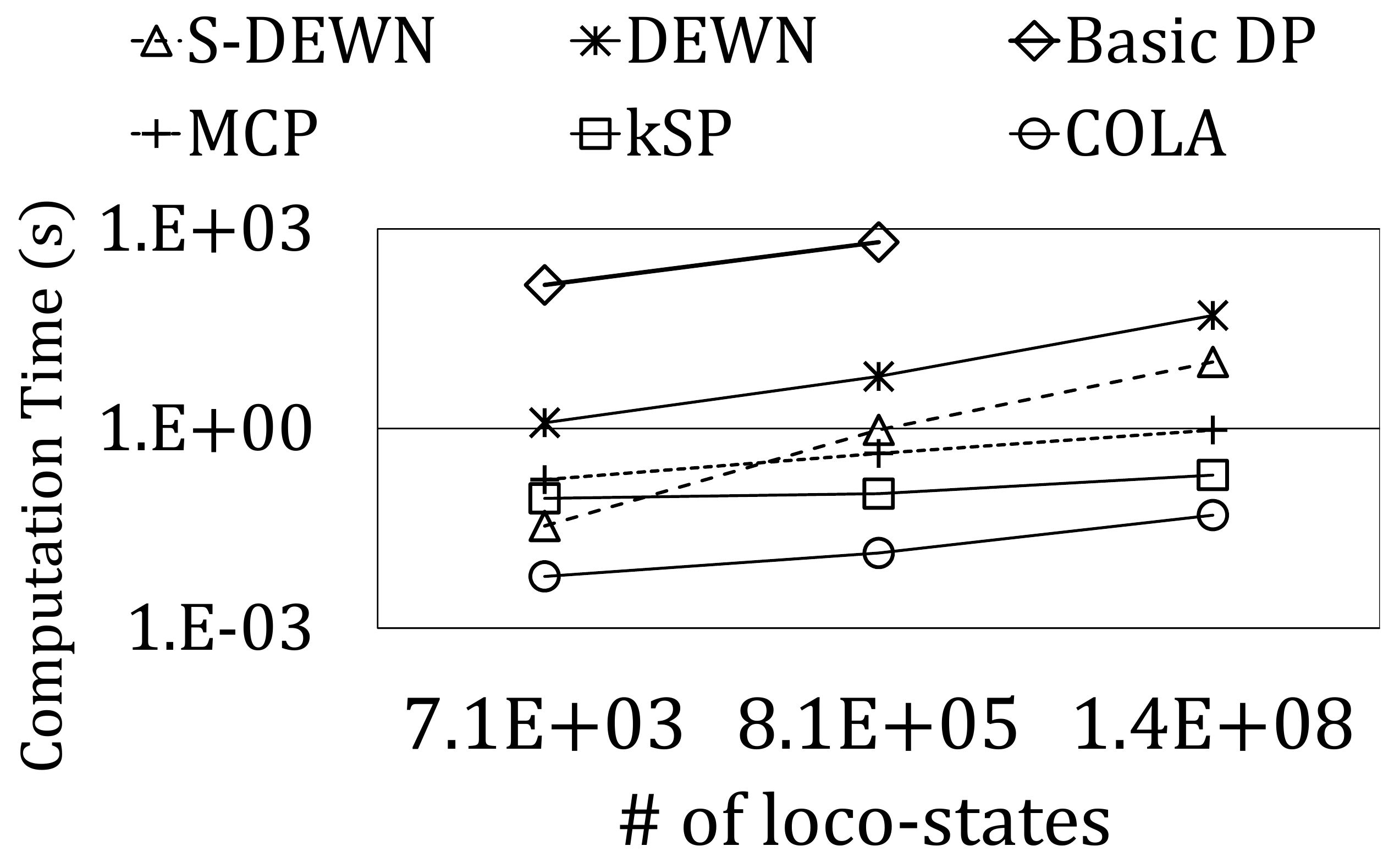}
    \caption{Efficiency on diff. maps.}
    \label{fig:Vmap_time}
  \end{subfigure}
  \hfill 
  \begin{subfigure}[b]{0.48\columnwidth}
    \includegraphics[width=\columnwidth, height=2.35cm]{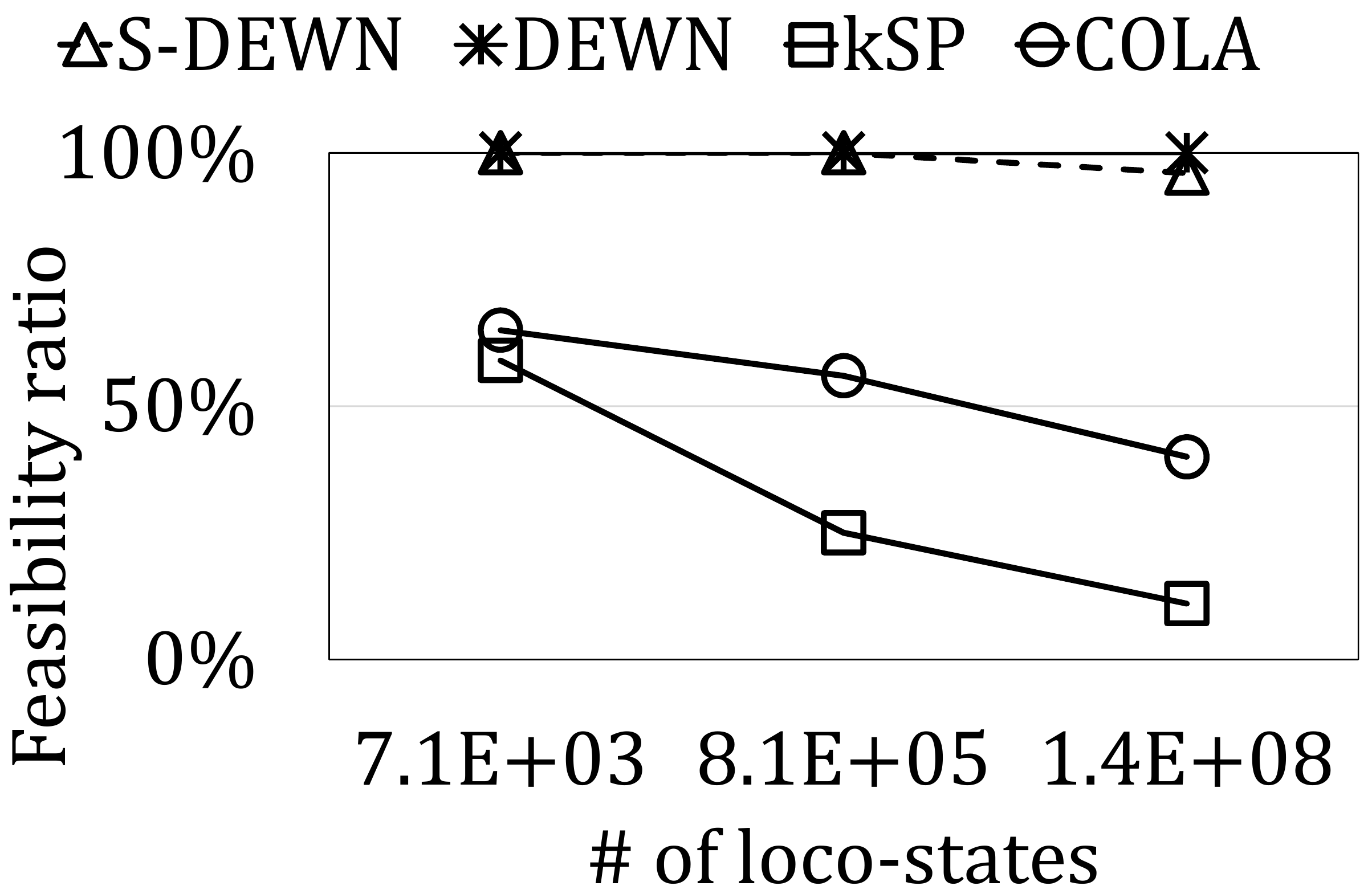}
    \caption{Fea. ratio on diff. maps.}
    \label{fig:Vmap_feas}
  \end{subfigure}
  \caption{Scalability test on large virtual maps.}
  \label{fig:scalability}
\end{figure}

Figure \ref{fig:scalability} compares the scalability of all methods in virtual maps with different sizes (i.e., number of virtual locations), which in turn results in different numbers of loco-states. The feasibilities of single-world methods kSP and COLA drop as the number of locations increases because more corners will appear in the virtual world and thus require more complicated RW paths. Figure \ref{fig:Vmap_time} indicates that DP is extremely unscalable. S-DEWN has comparable efficiency with single world methods kSP and COLA, which are implemented with specialized index structures. In regards to feasibility, S-DEWN and DEWN consistently outperforms single world methods.

\subsection{User Study}
\label{subsec:userstudy}

We conduct a user study to understand users' behaviors while they walk along paths returned by different algorithms in a VR maze, built by Unity 2017.3.1f1 and SteamVR Plugin 1.2.2, for users wearing hTC VIVE HMD. The VE is a 3D Pac-man arcade game implemented like \cite{PacMan16} but incorporated with real walking experience, where users are navigated along precomputed RW paths. Following the user-study setting of VR in Computer Graphics and HCI research \cite{EH14, MA16, HC17}, we recruited 30 users to test our developed system and provide feedback on immersion according to Presence Questionnaire \cite{MA16}, including important questions such as ``How much are your experiences in the virtual environment consistent with your real-world experiences?'' We also measure the dizziness according to Simulator Sickness Questionnaire (SSQ), which evaluates symptoms such as headaches, vertigo, and nausea, and is widely used in measuring motion sickness in VR applications \cite{CN12, EH14}. 

The experiment is described to the users as a single-user VR arcade game similar to the classical Pac-Man. Users are asked to actually walk along predefined paths in the virtual world and touch red reward pellets along the path in order to collect game points. The specified path is shown as a sequence of blue guide wires. The reward pellets actually play the role of anchor points; we update the user's location information when the device detects the user ``touching'' the cube, and apply different RW operation gains accordingly. Every user experiences multiple paths that vary in both total length, total RW cost, kinds of used RW operations, and also the virtual environment. We ask the users to provide feedback after finishing each path. 

\sloppy An example of the virtual and physical worlds is shown in Figure \ref{fig:user_vmap} and \ref{fig:user_pmap}, respectively. The virtual world is a maze environment while the physical world is an office equipped with two VIVE Lighthouse tracking base stations in the corners. The numbers represent the sequence of states in v-path and p-path of the example, where S and T are the start and destination locations, respectively. The Rotation, Reset, and Translation operations are labeled as yellow stars, green triangles, and red lines, respectively. No RW operation is involved for the white circles and blue lines in this example. 
\begin{figure}[t]
\centering
  \begin{subfigure}[b]{0.6\columnwidth}
    \centering
    \includegraphics[width=0.6\columnwidth]{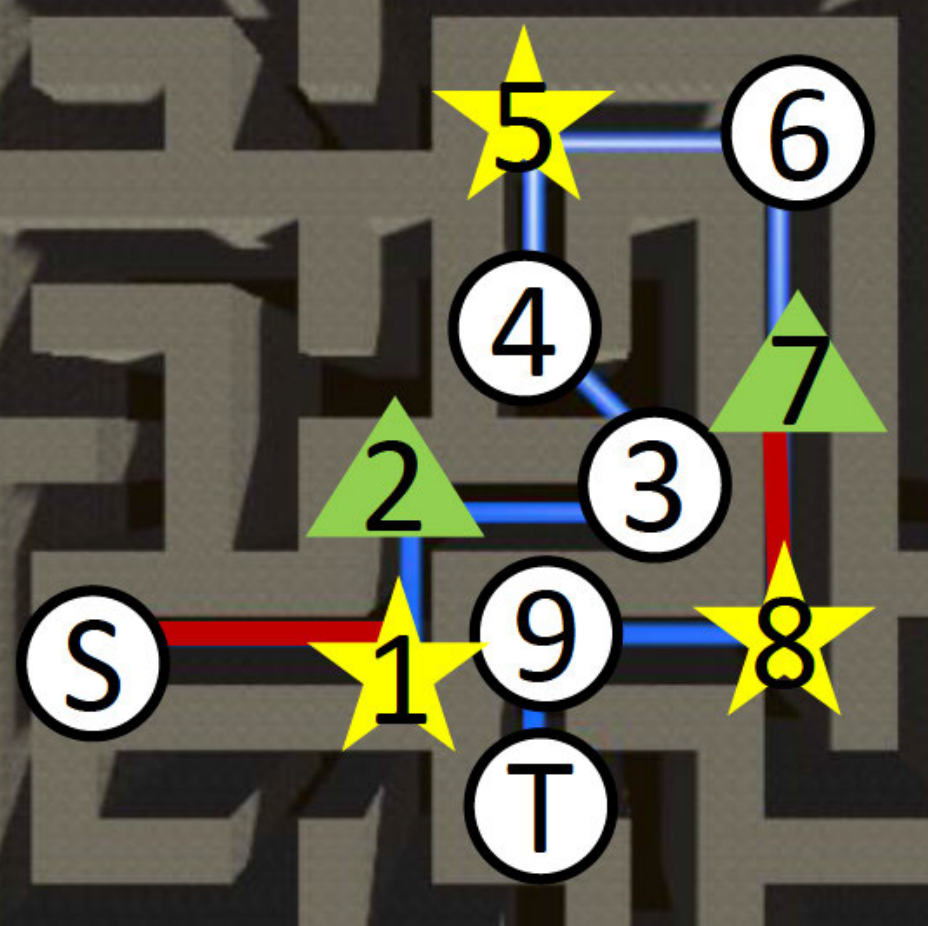}
    \caption{A v-path in the maze environment.}
    \label{fig:user_vmap}
  \end{subfigure}
  \begin{subfigure}[b]{0.36\columnwidth}
    \centering
    \includegraphics[width=0.7\columnwidth]{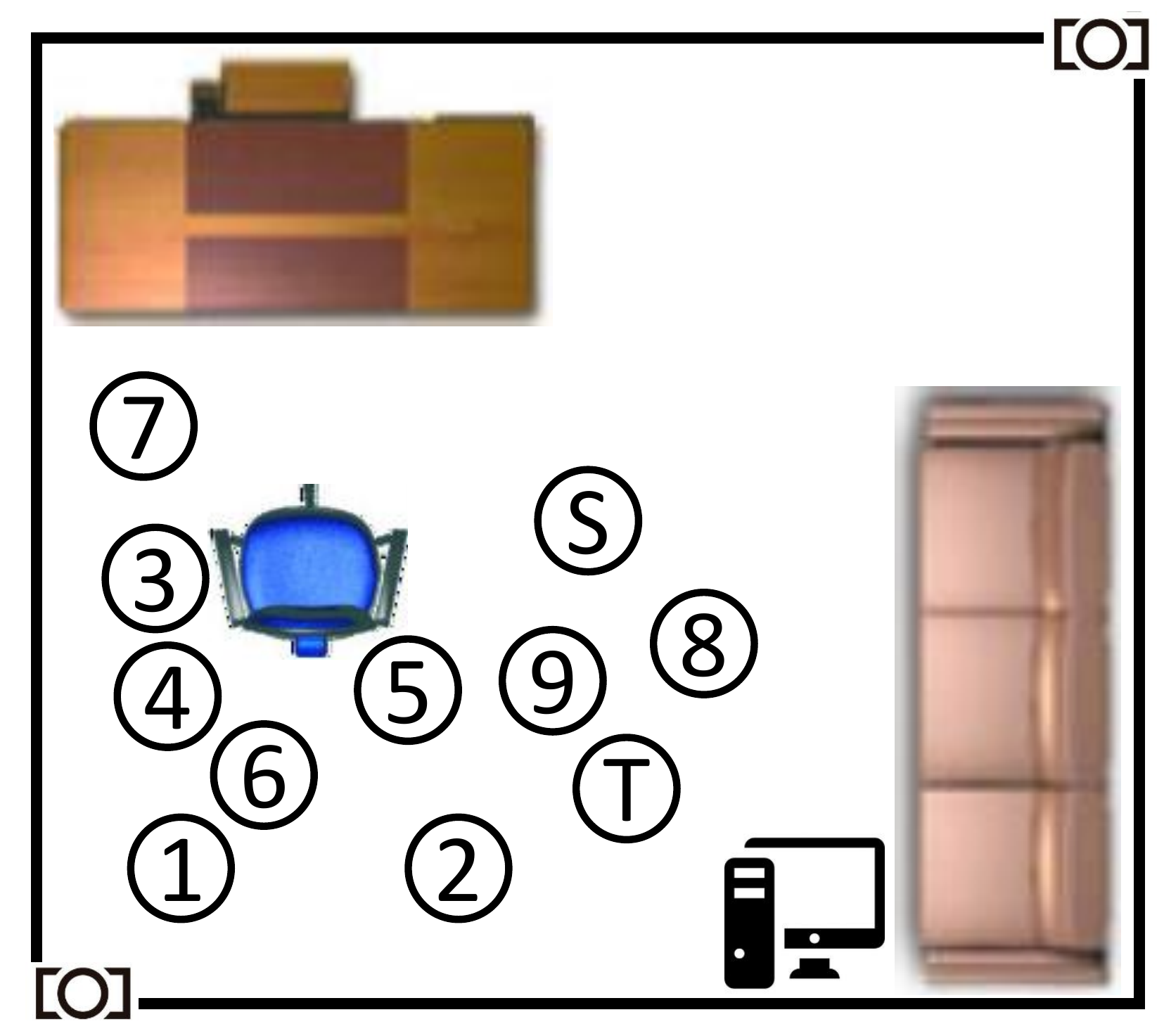}
    \caption{A p-path.}
    \label{fig:user_pmap}
  \end{subfigure}
  \begin{subfigure}[b]{0.49\columnwidth}
    \includegraphics[width=\columnwidth, height=2.35cm]{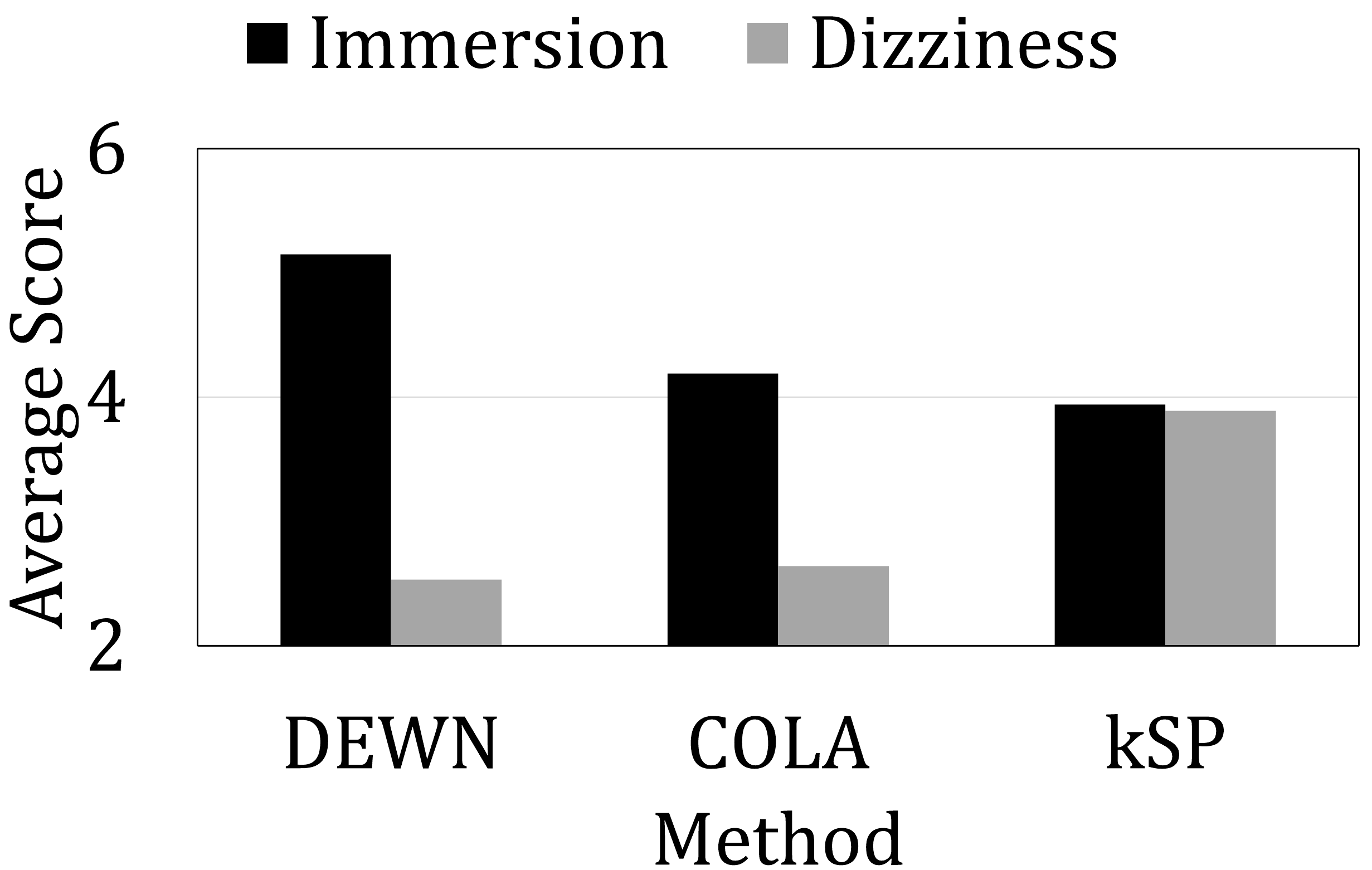}
    \caption{Effects on diff. methods.}
    \label{fig:user_method}
  \end{subfigure}
  \begin{subfigure}[b]{0.49\columnwidth}
    \includegraphics[width=\columnwidth, height=2.35cm]{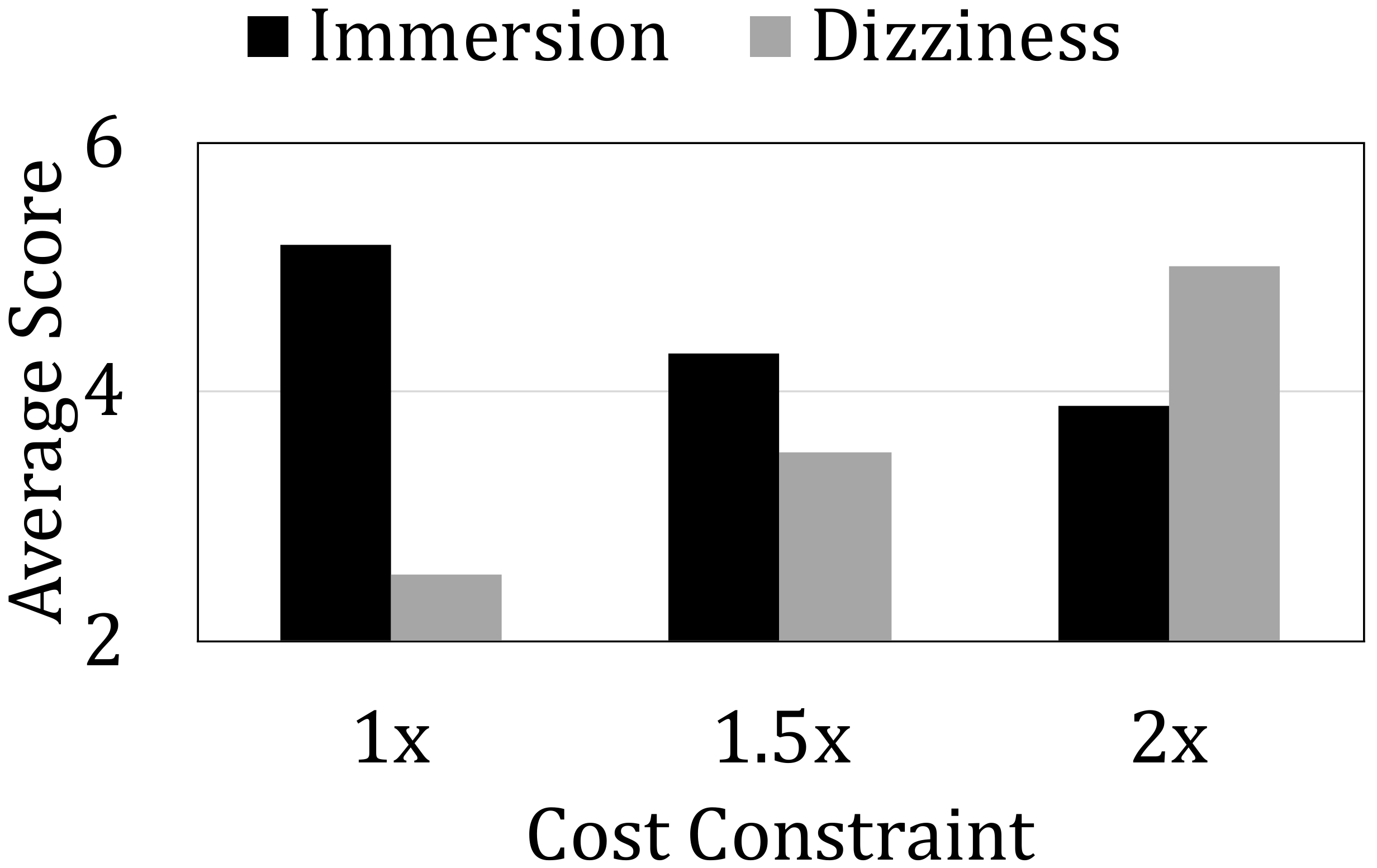}
    \caption{Effects on $C$.}
    \label{fig:user_cost}
  \end{subfigure}\\
\caption{Results of user study.}
\label{user_study}
\end{figure}
\sloppy Figure \ref{fig:user_method} compares the average immersion scores of DEWN, COLA, and kSP. DEWN outperforms kSP and COLA because it allows the users to follow the path with fewer and smoother RW operations. Figure \ref{fig:user_cost} presents the average immersion scores of DEWN with different RW cost constraints. As $C$ grows, the immersion slightly decreases with dizziness growing. However, according to user feedback, 93.5\% of the users are upset with kSP because it involves many Reset operations, and users almost bump into physical obstacles before RW operations are performed because the physical layout is not processed during the path search. All user recognizes that the RW operations in DEWN are much smoother for them.

\section{Conclusion} \label{sec:conclusions}
To the best of our knowledge, there exists no prior research that fully supports simultaneous movements in dual worlds for immersive user experience in VR. In this paper, we leverage Redirected Walking (RW) to formulate DROP, aiming to find the shortest-distance path in the virtual world, while constraining the RW cost to ensure immersive experience. Based on the idea of MIL Range, we design DEWN and propose various ordering and pruning strategies for efficient processing of DROP. Afterward, we show that the existing spatial query algorithms and index structures can leverage DEWN as a building block to support $k$NN and range queries in the dual worlds of VR. Experimental results and a user study manifest that DEWN can effectively find shorter v-paths with smoother RW operations compared with the baselines in various VR scenarios.

\begin{acks}
This work is supported in part by MOST in Taiwan through grant 107-2221-E-001-011-MY3.
\end{acks}

\bibliographystyle{plain}
\bibliography{ref}

\appendix
\section{Definitions of RW operations} \label{appen:rw_def}

Without loss of generality, when a user moves in the virtual world, it is assumed that the user first changes her orientation and then walks a straight step afterward.  Given the walking length $\ell$ of a user, the transition function from the previous loco-state $st_{n-1} = ((\gamma^{\text{v}}_{n-1}, \theta^{\text{v}}_{n-1}), (\gamma^{\text{p}}_{n-1}, \theta^{\text{p}}_{n-1}))$ to the next loco-state $st_n = ((\gamma^{\text{v}}_n, \theta^{\text{v}}_n), (\gamma^{\text{p}}_n, \theta^{\text{p}}_n))$ after walking a step is calculated as follows:
$$ \theta^{\text{p}}_{n} = \theta^{\text{p}}_{n-1} + \triangle \theta $$
$$ \theta^{\text{v}}_{n} = \theta^{\text{v}}_{n-1} + \triangle \theta $$
$$ \gamma^{\text{p}}_{n} = \gamma^{\text{p}}_{n-1} + \ell \times [\cos(\theta^{\text{p}}_{n}), \sin(\theta^{\text{p}}_{n})]^T$$
$$ \gamma^{\text{v}}_{n} = \gamma^{\text{v}}_{n-1} + \ell \times [\cos(\theta^{\text{v}}_n), \sin(\theta^{\text{v}}_n)]^T$$
where $\triangle \theta$ is the orientation difference between two loco-states.

\begin{definition}
\textit{Translation Gain} ($m_\text{T}$). When a Translation Operation (\textit{Translation}) is applied into an HMD, the change in the transition function is
$$ \gamma^{\text{v}}_{n} = \gamma^{\text{v}}_{n-1} + \ell \times m_\text{T} \times [\cos(\theta^{\text{v}}_{n}), \sin(\theta^{\text{v}}_{n})]^{\intercal}.$$
\end{definition}
In other words, when a user walks an $\ell$-length step in the physical world, she walks an $(\ell \times m_\text{T})$-length step in the virtual world, and $m_\text{T}$ is the \textit{translation gain}.

\begin{definition}
\textit{Rotation Gain} ($m_\text{R}$). The Rotation Operation (RO) manipulates the rotation speed in a VE when a user is turning into another direction, so that the rotation speeds in two worlds can be slightly different. The virtual orientation of the transition function becomes 
$$ \theta^{\text{v}}_{n} = \theta^{\text{v}}_{n-1} + m_\text{R} \times  \triangle \theta $$
where $m_\text{R}$ is the \textit{rotation gain} applied to an HMD. 
\end{definition}

\begin{definition}
\textit{Curvature Gain} ($m_\text{C}$). When a user walks straight in the virtual world, the Curvature Operation (CO) allows her to walk along a curve in the physical world. The physical orientation and position of the transition function are changed as follows.
$$ \theta^{\text{p}}_{n} = \theta^{\text{p}}_{n-1} + m_\text{C} \times \ell $$
$$ \gamma^{\text{p}}_{n} = \gamma^{\text{p}}_{n-1} + \frac{1}{m_\text{C}} \times \begin{bmatrix}
\sin(\theta^{\text{p}}_{n} + m_\text{C} \times \ell)-\sin(\theta^{\text{p}}_{n})\\ 
\cos(\theta^{\text{p}}_{n})-\cos(\theta^{\text{p}}_{n}+m_\text{C} \times \ell),
\end{bmatrix}$$
where $m_\text{C}$ is the \textit{curvature gain} applied to an HMD.
\end{definition}

\begin{definition}
\textit{Reset turning angle} ($\theta_{\text{Reset}}.$) Sometimes the \textit{Reset operation} (Reset) is required to explicitly ask the user to turn in a different direction in the physical world (but remains in the same direction in the VR world) in order to avoid the physical walls or obstacles. It is expected that Reset usually incurs a higher cost since it interrupts the user experience in the VR world. Given the \textit{Reset turning angle} $\theta_{\text{Reset}}$, a user is asked to turn $\theta_{\text{Reset}}$ in her current physical position, but the image display in the VR world is suspended during the turning. Therefore,
$$ \theta^{\text{p}}_{n} = \theta^{\text{p}}_{n-1} + \theta_{\text{Reset}}. $$
\end{definition}

While the above sets of transition functions represent the basic RW operations, one can define other transition functions to abstract other implementations of VR locomotion techniques, e.g., teleportation, where the user determines the next v-state, and the p-state remains unchanged. Thus, the notion of transition functions is general.

\section{Cost Model Approaches} \label{appen:rw_cost_model}

Here we briefly discuss several possible approaches to setup the cost model for RW operations. This cost model can be viewed as a cost function $C_{\text{RW}}(op, z)$ that takes both the type of the RW operation $op$ and the usage magnitude $z$ and maps to a positive RW cost $C_{\text{RW}}(op, z)$, where $z$ is $m_\text{T}$ for \textit{Translation}, $m_\text{R}$ for RO, $m_\text{C}$ for CO, and $\theta_{\text{Reset}}$ for Reset.

\begin{itemize}
    \item \textbf{Usage Count.} Using any RW operation $op$ incurs an RW cost of 1 unit, regardless of the magnitude $z$ and type of $op$. Thus, $C_{\text{RW}}(op, z) = 1$ for all $op$ and $z$.
    \item \textbf{Detection Likelihood.} When an RW operation $op$ is applied with a specific magnitude $z$, it incurs an RW cost proportional to the likelihood that it is detected by an average user. For example, according to \cite{SS13}, a \textit{Translation} with $m_\text{T} = 0.6$, i.e., down-scaling the walking distance by 40\%, has a roughly 90\% chance to be detected by the users. Thus, $C_{\text{RW}}(\text{\textit{Translation}}, 0.6) = 0.9$.
    \item \textbf{Detection Threshold.} When an RW operation $op$ is applied with a specific magnitude $z$, if $z$ is in the non-detectable region, e.g., $z = m_\text{R} \in (0.77, 1.10)$ for no-audio RO in \cite{FM16},  $C_{\text{RW}}(op, z) = 0$, indicating the user does not feel the modification, and $C_{\text{RW}}(op, z) = 1$ for all other values of $z$.
    \item \textbf{Other Threshold.} The detection thresholds used above can be changed to any other variations of threshold of RW operations, e.g., \textit{perception}, \textit{applicability}, or \textit{immersion} thresholds in \cite{MR18}.
    \item \textbf{Reset Cost.} Since Reset directly interrupts the user experience, it is not meaningful to quantify the detection-based cost for it. Instead, a possible approach is to set $C_{\text{RW}}(\text{Reset}, z) = c_{\text{Reset}}$ for some constant cost $c_{\text{Reset}}$ for all $z \neq 0$ to represent the inconvenience and degradation of immersion experienced by the user. Another possibility is to consider the angle that the user is asked to rotate in Reset, e.g., $C_{\text{RW}}(\text{Reset}, z) = c_{\text{Reset}} \cdot \frac{|z|}{180}$, where $z \in (-180,180]$ is the reset turning angle.
\end{itemize}

Note that in any cost model, $C_{\text{RW}}(op, z) = 0$ for $z = m_\text{T} = 1$ for \textit{Translation}, $z = m_\text{R} = 1$ for RO, $z = m_\text{C} = 0$ for CO, and $z = \theta_{\text{Reset}} = 0$ for Reset, since these values corresponds to no RW operations, i.e., the movements in the dual worlds are aligned. 

For walking operations, i.e., \textit{Translation} and CO, it is also applicable to further weight the RW cost by the total walking distance that the user is under the given RW operation. In other words, given $op$ and $z$, the total RW cost of the user walks under $op$ for a distance of $\ell$ is $C_{\text{RW}}(op, z) \cdot \ell$. RO and Reset only affects the turning movement and do not need to be weighted. 

\section{Experimental Results on Cost Model} \label{appen:exp_cost}

\begin{figure}[t]
  \begin{subfigure}[b]{0.48\columnwidth}
    \includegraphics[width=\columnwidth]{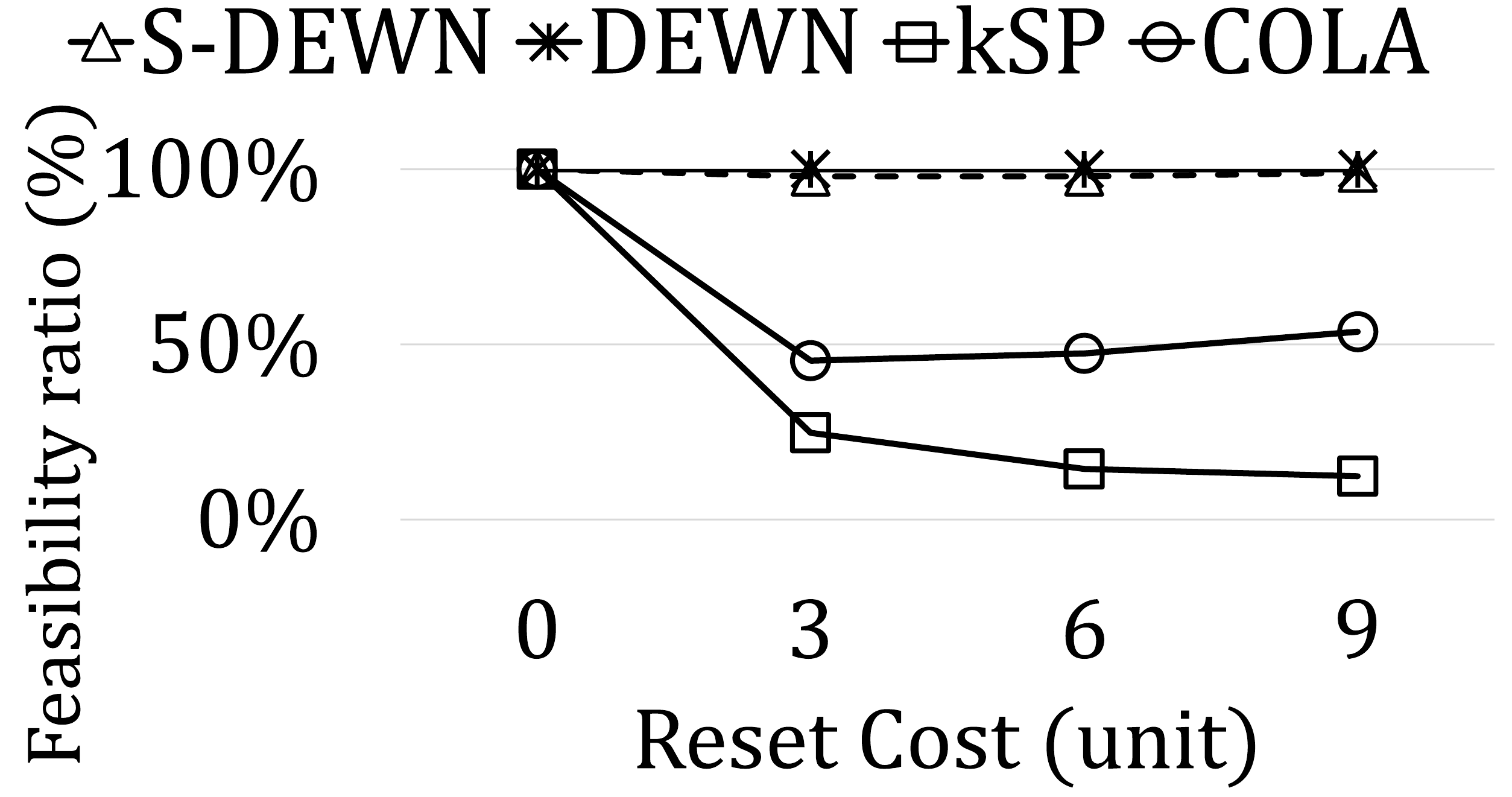}
    \caption{Feasibility ratio.}
    \label{fig:Reset_fea}
  \end{subfigure}
  \hfill 
  \begin{subfigure}[b]{0.48\columnwidth}
    \includegraphics[width=\columnwidth]{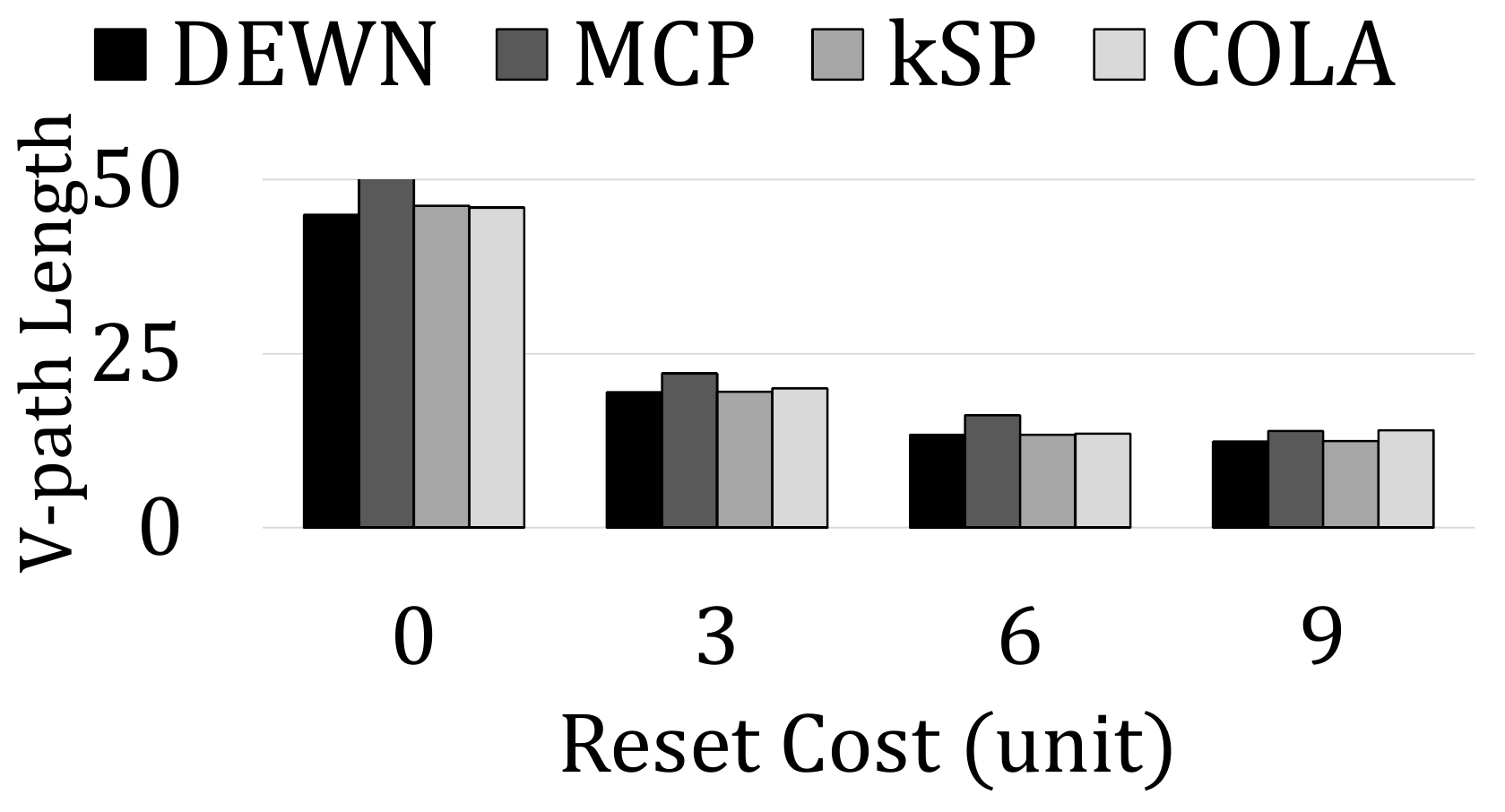}
    \caption{V-path length.}
    \label{fig:Reset_vpath}
  \end{subfigure}
  \caption{Experimental results on diff. $c_{\text{Reset}}$.}
  \label{fig:reset}
\end{figure}

Figure \ref{fig:reset} reports the experimental results for the detection threshold cost model with different values of $c_{\text{Reset}}$, where the query distances (distance between start and destination locations) are randomly distributed from 0 to 100, and the cost constraint is set to 10. The feasibility ratios are shown in Figure \ref{fig:Reset_fea}. MCP and DP share similar results (i.e., 100\% feasibility) with DEWN and thus are not shown here. All methods are 100\% feasible when $c_{\text{Reset}} = 0$, i.e., Reset is free, as it becomes feasible to abuse Reset to steer the user away from obstacles and boundaries in the physical world. As $c_{\text{Reset}}$ increases, the feasibility ratio of kSP is significantly affected as it relies on Reset as its only way to align the dual worlds. The feasibility ratio of COLA also significantly decreases as $c_{\text{Reset}}$ becomes nonzero. However, as $c_{\text{Reset}}$ grows large, Reset is less likely to be used in the RW operation configuration corresponding to the MIL values. Thus, the MIL Ranges are narrower, and the MIL upper bound values become more accurate in reflecting the total MIL values. Therefore, the feasibility ratio of COLA slightly improves as $c_{\text{Reset}}$ grows.

Figure \ref{fig:Reset_vpath} reports the average \textit{feasible} v-path lengths. All queries are feasible when $c_{\text{Reset}} = 0$. As $c_{\text{Reset}}$ grows, DROP queries with longer query distances become infeasible as Reset becomes costly. However, the feasibility of \textit{query instances} becomes stable after $c_{\text{Reset}}$ becomes sufficiently large (around $c_{\text{Reset}} = 6$). Since the feasible solutions do not rely on Reset when $c_{\text{Reset}}$ is large, continuing to increase $c_{\text{Reset}}$ does not affect the instance feasibility.

\end{document}